\documentclass[11pt]{amsart}
\usepackage[english]{babel}
\usepackage[utf8]{inputenc}
\usepackage[inner=1in, outer=1in, bottom =2.9cm, top =2.9cm]{geometry}
\makeatletter
\newcommand*\Cdot{\mathpalette\Cdot@{.5}}
\newcommand*\Cdot@[2]{\mathbin{\vcenter{\hbox{\scalebox{#2}{$\m@th#1\circled{$1$}$}}}}}
\makeatother

\usepackage{amsmath, amssymb, amsfonts, amsthm}
\usepackage{graphicx, overpic}
\usepackage{mathrsfs}
\usepackage{mathtools}
\usepackage[usenames,dvipsnames,svgnames,table]{xcolor}
\usepackage{enumerate}
\usepackage{enumitem}
\usepackage{shuffle} %Shuffle
\usepackage{lipsum}
\usepackage{color}
\usepackage{phoenician}
\usepackage{xkeyval}
\usepackage{tikz}
\usepackage{pst-node} % Commutative Diagrams
\usepackage{tikz-cd} % Commutative Diagrams
\usepackage[OT2,T1]{fontenc} %Shuffle
\usepackage{thm-restate} 
\usepackage{float}
\usepackage{bold-extra}
\usepackage{microtype}
\usepackage[all,cmtip]{xy}
\usepackage{etoolbox}
\usepackage[strict]{changepage} %Margins
\usepackage{esvect}
\usepackage{wrapfig}
\usepackage{caption}
\usepackage{cite}
\usepackage{eucal}
\usepackage{mathrsfs}
\DeclareMathAlphabet{\mathpzc}{OT1}{pzc}{m}{it}
\usepackage{bbding}
\usepackage{array}
\usepackage{pict2e}
\usepackage{makecell}
\usepackage{stmaryrd}
\usepackage{lmodern}
\usepackage{amsbsy}
\usepackage{esvect}
\usepackage{bbding}
\usepackage{thm-restate}
\usepackage[toc,page, titletoc,title]{appendix}
\graphicspath{ {figures/} }
\usepackage{url}
\usepackage[pagebackref]{hyperref} %PUT LAST
\newtagform{red}{\color{blue}(}{)}

\colorlet{linkequation}{blue}
\newcommand*{\refeqq}[1]{%
  \begingroup
    \hypersetup{
      linkcolor=linkequation,
      linkbordercolor=linkequation,
    }%
    \ref{#1}%
  \endgroup
}
\makeatletter
\newcommand{\colim@}[2]{%
  \vtop{\m@th\ialign{##\cr
    \hfil$#1\operator@font colim$\hfil\cr
    \noalign{\nointerlineskip\kern1.5\ex@}#2\cr
    \noalign{\nointerlineskip\kern-\ex@}\cr}}%
}
\newcommand{\colim}{%
  \mathop{\mathpalette\colim@{\rightarrowfill@\scriptscriptstyle}}\nmlimits@
}
\renewcommand{\varprojlim}{%
  \mathop{\mathpalette\varlim@{\leftarrowfill@\scriptscriptstyle}}\nmlimits@
}
\renewcommand{\varinjlim}{%
  \mathop{\mathpalette\varlim@{\rightarrowfill@\scriptscriptstyle}}\nmlimits@
}
\makeatother

\makeatletter
\providecommand*{\twoheadrightarrowfill@}{%
  \arrowfill@\relbar\relbar\twoheadrightarrow
}
\providecommand*{\twoheadleftarrowfill@}{%
  \arrowfill@\twoheadleftarrow\relbar\relbar
}
\providecommand*{\xtwoheadrightarrow}[2][]{%
  \ext@arrow 0579\twoheadrightarrowfill@{#1}{#2}%
}
\providecommand*{\xtwoheadleftarrow}[2][]{%
  \ext@arrow 5097\twoheadleftarrowfill@{#1}{#2}%
}
\makeatother
\makeatletter
\newcommand*{\relrelbarsep}{.386ex}
\newcommand*{\relrelbar}{%
  \mathrel{%
    \mathpalette\@relrelbar\relrelbarsep
  }%
}
\newcommand*{\@relrelbar}[2]{%
  \raise#2\hbox to 0pt{$\m@th#1\relbar$\hss}%
  \lower#2\hbox{$\m@th#1\relbar$}%
}
\providecommand*{\rightrightarrowsfill@}{%
  \arrowfill@\relrelbar\relrelbar\rightrightarrows
}
\providecommand*{\leftleftarrowsfill@}{%
  \arrowfill@\leftleftarrows\relrelbar\relrelbar
}
\providecommand*{\xrightrightarrows}[2][]{%
  \ext@arrow 0359\rightrightarrowsfill@{#1}{#2}%
}
\providecommand*{\xleftleftarrows}[2][]{%
  \ext@arrow 3095\leftleftarrowsfill@{#1}{#2}%
}
\makeatother
\DeclareSymbolFont{cyrletters}{OT2}{wncyr}{m}{n}
\DeclareMathSymbol{\Sh}{\mathalpha}{cyrletters}{"58}

\usetikzlibrary{tikzmark,decorations.pathreplacing}
\usetikzlibrary{shapes,shadows,arrows}
\usetikzlibrary{decorations.markings}
\usetikzlibrary{positioning,calc}
\tikzset{near start abs/.style={xshift=1cm}}
\DeclareSymbolFont{symbolsC}{U}{txsyc}{m}{n}
\DeclareMathSymbol{\Searrow}{\mathrel}{symbolsC}{117}

\hypersetup{
	colorlinks=true,
	linkcolor=black,
	filecolor=black,      
	urlcolor=blue,
	citecolor= 	red,
}
\DeclareSymbolFont{extraup}{U}{zavm}{m}{n}
\DeclareMathSymbol{\varheart}{\mathalpha}{extraup}{86}
\DeclareMathSymbol{\vardiamond}{\mathalpha}{extraup}{87}
 \DeclareMathSymbol{\varclub}{\mathalpha}{extraup}{84} 
\DeclareMathSymbol{\varspade}{\mathalpha}{extraup}{85}

\newcommand{\bigslant}[2]{{\raisebox{.2em}{$#1$}\left/\raisebox{-.2em}{$#2$}\right.}}
\theoremstyle{definition}
\newtheorem{thm}{Theorem}[section]
\newtheorem{cor}{Corollary}[thm]

\newtheorem{prop}[thm]{Proposition}

\theoremstyle{definition}

\newtheorem{ex}{Example}[section]

\newtheorem{remark}{Remark}[section]

\newcommand{\ga}{\alpha}

\newcommand{\sfVec}{\textsf{Vec}}
\newcommand{\sfS}{\textsf{S}}

\newcommand{\tH}{\mathtt{H}}
\newcommand{\tQ}{\mathtt{Q}}

\newcommand{\bR}{\mathbb{R}}
\newcommand{\bN}{\mathbb{N}}
\newcommand{\bC}{\mathbb{C}}

\newcommand{\bZ}{\mathbb{Z}}

\newcommand{\cF}{\CMcal{F}}

\newcommand{\cH}{\CMcal{H}}

\newcommand{\cS}{\CMcal{S}}
\newcommand{\cA}{\CMcal{A}}

\newcommand{\cK}{\CMcal{K}}
\newcommand{\cX}{\CMcal{X}}

\newcommand{\op}{\operatorname{op}}

\newcommand{\Hom}{\operatorname{Hom}}

\newcommand{\bE}{\textbf{E}}

\newcommand{\sS}{\textsf{S}}

\newcommand{\Lie}{ \textbf{Lie} }
\newcommand{\Zie}{\textbf{Zie}}
\newcommand{\Com}{\textbf{Com}}
\newcommand{\Ass}{\textbf{Ass}}

\newcommand{\Gam}{\boldsymbol{\Gamma}}
\newcommand{\Sig}{\boldsymbol{\Sigma}}  
\DeclareFontFamily{U}{mathx}{\hyphenchar\font45}
\DeclareFontShape{U}{mathx}{m}{n}{
      <5> <6> <7> <8> <9> <10>
      <10.95> <12> <14.4> <17.28> <20.74> <24.88>
      mathx10
      }{}
\DeclareSymbolFont{mathx}{U}{mathx}{m}{n}
\DeclareFontSubstitution{U}{mathx}{m}{n}
\DeclareMathAccent{\widecheck}{0}{mathx}{"71}
\DeclareMathAccent{\wideparen}{0}{mathx}{"75}

\newcommand{\res}{\parallel}

\newcommand{\formj}{\emph{\texttt{j}}\, }
\newcommand{\formg}{\emph{\texttt{g}}}

\newcommand{\bT}{\mathbb{T}}

\newcommand{\ssS}{\emph{\textsf{S}}}
\newcommand{\ssA}{\emph{\textsf{A}}}

\makeatletter
\newcommand*\bigcdot{\mathpalette\bigcdot@{.5}}
\newcommand*\bigcdot@[2]{\mathbin{\vcenter{\hbox{\scalebox{#2}{$\m@th#1\bullet$}}}}}
\makeatother

\makeatletter
\newcommand{\adjunction}{\@ifstar\named@adjunction\normal@adjunction}
\newcommand{\normal@adjunction}[4]{%
  #1\colon #2%
  \mathrel{\vcenter{%
    \offinterlineskip\m@th
    \ialign{%
      \hfil$##$\hfil\cr
      \longrightharpoonup\cr
      \noalign{\kern-.3ex}
      \smallbot\cr
      \longleftharpoondown\cr
    }%
  }}%
  #3 \noloc #4%
}
\newcommand{\named@adjunction}[4]{%
  #2%
  \mathrel{\vcenter{%
    \offinterlineskip\m@th
    \ialign{%
      \hfil$##$\hfil\cr
      \scriptstyle#1\cr
      \noalign{\kern.1ex}
      \longrightharpoonup\cr
      \noalign{\kern-.3ex}
      \smallbot\cr
      \longleftharpoondown\cr
      \scriptstyle#4\cr
    }%
  }}%
  #3%
}
\newcommand{\longrightharpoonup}{\relbar\joinrel\rightharpoonup}
\newcommand{\longleftharpoondown}{\leftharpoondown\joinrel\relbar}
\newcommand\noloc{%
  \nobreak
  \mspace{6mu plus 1mu}
  {:}
  \nonscript\mkern-\thinmuskip
  \mathpunct{}
  \mspace{2mu}
}
\newcommand{\smallbot}{%
  \begingroup\setlength\unitlength{.15em}%
  \begin{picture}(1,1)
  \roundcap
  \polyline(0,0)(1,0)
  \polyline(0.5,0)(0.5,1)
  \end{picture}%
  \endgroup
}
\makeatother
\newcommand{\leftrarrows}{\mathrel{\raise.75ex\hbox{\oalign{%
  $\scriptstyle\leftarrow$\cr
  \vrule width0pt height.5ex$\hfil\scriptstyle\relbar$\cr}}}}
\newcommand{\lrightarrows}{\mathrel{\raise.75ex\hbox{\oalign{%
  $\scriptstyle\relbar$\hfil\cr
  $\scriptstyle\vrule width0pt height.5ex\smash\rightarrow$\cr}}}}
\newcommand{\Rrelbar}{\mathrel{\raise.75ex\hbox{\oalign{%
  $\scriptstyle\relbar$\cr
  \vrule width0pt height.5ex$\scriptstyle\relbar$}}}}

\makeatletter
\def\leftrightarrowsfill@{\arrowfill@\leftrarrows\Rrelbar\lrightarrows}
\newcommand{\xleftrightarrows}[2][]{\ext@arrow 3399\leftrightarrowsfill@{#1}{#2}}
\makeatother

\newcommand{\la}{\langle}
\newcommand{\ra}{\rangle}
\newcommand{\wt}{\widetilde}
\newcommand{\wh}{\widehat}

\definecolor{Red}{rgb}{0.8,0,0.2}
\newcommand{\GG}[1]{}

\makeatletter
\def\@footnotecolor{red}
\define@key{Hyp}{footnotecolor}{%
 \HyColor@HyperrefColor{#1}\@footnotecolor%
}
\def\@footnotemark{%
    \leavevmode
    \ifhmode\edef\@x@sf{\the\spacefactor}\nobreak\fi
    \stepcounter{Hfootnote}%
    \global\let\Hy@saved@currentHref\@currentHref
    \hyper@makecurrent{Hfootnote}%
    \global\let\Hy@footnote@currentHref\@currentHref
    \global\let\@currentHref\Hy@saved@currentHref
    \hyper@linkstart{footnote}{\Hy@footnote@currentHref}%
    \@makefnmark
    \hyper@linkend
    \ifhmode\spacefactor\@x@sf\fi
    \relax
  }%
\makeatother

\hypersetup{footnotecolor=blue}
\makeatletter
\patchcmd{\@startsection}
  {\@afterindenttrue}
  {\@afterindentfalse}
  {}{}
\makeatother
\title[Species-Theoretic Foundations of Perturbative Quantum Field Theory]{Species-Theoretic Foundations of Perturbative\\ Quantum Field Theory}
\author{William Norledge}
\address{Pennsylvania State University}
\email{wxn39@psu.edu}
\begin{document}

\usetagform{red}

% % % % % % % % % % % % % % % % % % % % % % % % %
% Autoref names
\renewcommand{\chapterautorefname}{Chapter}
\renewcommand{\sectionautorefname}{Section}
\renewcommand{\subsectionautorefname}{Section}
% % % % % % % % % % % % % % % % % % % % % % % % % 

% % % % % % % % % % % % % % % % % % % % % % % % % % % % % 
% Autoref names
\renewcommand{\chapterautorefname}{Chapter}
\renewcommand{\sectionautorefname}{Section}
\renewcommand{\subsectionautorefname}{Section}
% % % % % % % % % % % % % % % % % % % % % % % % % % % % % 

\begin{abstract}
We develop an algebraic formalism for perturbative quantum field theory (pQFT) which is based on Joyal's combinatorial species. We show that certain basic structures of pQFT are correctly viewed as algebraic structures internal to species, constructed with respect to the Cauchy monoidal product. Aspects of this formalism have appeared in the physics literature, particularly in the work of Bogoliubov-Shirkov, Steinmann, Ruelle, and \hbox{Epstein-Glaser-Stora}. In this paper, we give a fully explicit account in terms of modern theory developed by \hbox{Aguiar-Mahajan}. We describe the central construction of causal perturbation theory as a homomorphism from the Hopf monoid of set compositions, decorated with local observables, into the Wick algebra of microcausal polynomial observables. The operator-valued distributions called (generalized) time-ordered products and (generalized) retarded products are obtained as images of fundamental elements of this Hopf monoid under the curried homomorphism. The perturbative \hbox{S-matrix} scheme corresponds to the so-called universal series, and the property of causal factorization is naturally expressed in terms of the action of the Hopf monoid on itself by Hopf powers, called the Tits product. Given a system of fully renormalized time-ordered products, the perturbative construction of the corresponding interacting products is via an up biderivation of the Hopf monoid, which recovers Bogoliubov's formula. 
%then correspond to various elements of the algebra. In particular, time-ordered products are the universal series, operator products of time-ordered products are the H-basis, and generalized retarded products are the Dynkin elements.
%The S-matrix is the image of a decorated universal series in the cocommutative Hopf monoid of set compositions under the homomorphism into the Wick algebra which is determined by the time-ordered products, being Wick algebra-valued functions of the decorations. 
\end{abstract}

\maketitle

% % % % % % % % % % % % % % % % % % % % % % % % % % % % %  gap between abstract and contents
\vspace{-6.5ex}
% % % % % % % % % % % % % % % % % % % % % % % % % % % % % 

% % % % % % % % % % % % % % % % % % % % % % % % % % % % % 
\setcounter{tocdepth}{1} % what appears
\hypertarget{foo}{ }
\tableofcontents
%\setlength{\cftbeforesecskip}{1pt} % spacing
% % % % % % % % % % % % % % % % % % % % % % % % % % % % % 

%%%%%%%%%%%%%%%%%%%%%%%%%%%%%%%%%%%%%%%%%%%%%%%%%%%%%%%%%%%%%%%%%%%%%%%%
%%%%%%%%%%%%%%%%%%%%%%%%%%%%%%%%%%%%%%%%%%%%%%%%%%%%%%%%%%%%%%%%%%%%%%%%
\section*{Introduction}\label{intro}
%%%%%%%%%%%%%%%%%%%%%%%%%%%%%%%%%%%%%%%%%%%%%%%%%%%%%%%%%%%%%%%%%%%%%%%%
%%%%%%%%%%%%%%%%%%%%%%%%%%%%%%%%%%%%%%%%%%%%%%%%%%%%%%%%%%%%%%%%%%%%%%%%

The theory of species valued in sets (or vector spaces) is a richer, categorified (and linearized) version of analyzing combinatorial structures in terms of generating functions, going back to André Joyal \cite{joyal1981theorie}, \cite{joyal1986foncteurs}, \cite{bergeron1998combinatorial}. In this approach, one sees additional structure by encoding processes of relabeling combinatorial objects, that is by modeling combinatorial objects as presheaves on the category $\sfS$ of finite sets $I$ (the labels) and bijections $\sigma$ (relabellings). See \cite[Foreword]{bergeron1998combinatorial}. The mathematician's gauge principle says that it is not redundant to account for all these relabelings, despite them not changing anything up to isomorphism. 

%One may consider such combinatorial species valued in sets as discrete cases of species valued in categories of spaces, which are often obtained as configuration spaces of points or moduli spaces with marked points. Note that (structures equivalent to) species appearing in this context are often called symmetric sequences, symmetric collections, or S-modules, and they often have the structure of (co)algebras over $\textsf{S}$-colored (co)operads. In this guise, algebraic structures in species already play a key role in cohomological field theory.

%symmetric collections, or S-modules.
%(or modules), which are often obtained as (the homology of) configuration spaces of points or moduli spaces with marked points. Note that species appearing in this context are often called symmetric sequences, symmetric collections, or S-modules. 

%Related to this is the consideration of species as an aspect of categorified differential calculus, e.g. the calculus of functors or Goodwillie calculus. 

In this paper, we are concerned with vector species $\textbf{p}$, which are presheaves of vector spaces over a field $\Bbbk$ on $\sfS$,
\[
\textbf{p}:\sfS^{\op}\to \sfVec, \qquad I\mapsto \textbf{p}[I] 
\quad \text{and} \quad 
\sigma \mapsto \textbf{p}[\sigma]
.\] 
A highly structured theory of gebras\footnote{\ meaning (co/bi/Hopf)algebras, Lie (co)algebras, etc.} internal to the functor category of vector species has been developed by Aguiar and Mahajan \cite{aguiar2010monoidal}, \cite{aguiar2013hopf}, building on the work of Barratt \cite{barratt1978twisted}, Joyal \cite{joyal1986foncteurs}, Schmitt \cite{Bill93}, Stover \cite{stover1993equivalence}, and others. For the internalization, one uses the Day convolution monoidal product $\textbf{p}\bigcdot\textbf{q}$ with respect to disjoint union and tensor product (\autoref{sec:Monoidal Products}), 
\[
\textbf{p}\bigcdot\textbf{q}[I] =  \textbf{p} \otimes_{\text{Day}} \textbf{q}[I]= \bigoplus_{S\sqcup T=I} \textbf{p}[S]\otimes \textbf{q}[T]
.\]
This is a categorification of the Cauchy product of formal power series. From the perspective of $\textsf{S}$-colored (co)operads, there is an equivalent description of these gebras as (co)algebras over the left (co)action (co)monads of the (co)operads $\Com^{ (\ast) }$, $\Ass^{ (\ast) }$, $\Lie^{ (\ast) }$ \cite[Appendix B.5]{aguiar2010monoidal}. This relates the gebras of this paper to structures such as cyclic operads, which appear in mathematical physics. Note that various decategorifications of \hbox{Aguiar-Mahajan's} theory recovers the plethora of graded combinatorial Hopf algebras which have been studied \cite[Chapter 5]{aguiar2010monoidal}. %In particular, the species analogs of the Connes-Kreimer graded Hopf algebras of rooted trees are constructed in \cite[Section 13.3]{aguiar2010monoidal}. 

On the other hand, quantum field theory (QFT) may be viewed as a kind of modern infinite dimensional calculus. Perturbative quantum field theory (pQFT) is the part of QFT which considers Taylor series approximations of smooth functions. By an argument of Dyson \cite{Dyson52}, Taylor series of realistic pQFTs are expected to have vanishing radius of convergence.
%They are formal power series in both Planck's constant $\hbar$, measuring the size of quantum effects, and the coupling constant $\formg$, measuring the size of interactions.\footnote{\ thus, in pAQFT, we study the formal infinitesimal neighborhoods of classical free field theories} 
Nevertheless, if an actual smooth function of a non-perturbative quantum field theory is being approximated, then they are asymptotic series, and so one might expect their truncations to agree to reasonable precision with experiment. This is indeed the case. 

%In contrast to pQFT, there are currently no known examples of realistic (non-perturbative) QFTs.

There are two main synthetic approaches to (non-perturbative) QFT, which grew out of the failure to make sense of the path integral analytically. There is functorial quantum field theory (FQFT), with axioms due to Atiyah and Segal, which formalizes the Schr\"odinger picture by assigning time evolution operators to cobordisms between spacetimes. There is also algebraic quantum field theory (AQFT), with axioms due to Haag and Kastler \cite{haagkas64}, which formalizes the Heisenberg picture by assigning $\text{C}^\ast$-algebras of observables to regions of spacetime. Low dimension examples of AQFTs/Wightman field theories were rigorously constructed in seminal work of \hbox{Glimm-Jaffe} and others \cite{MR247845}, \cite{MR272301}, \cite{MR363256}. 
%See also \cite{glimm2012quantum}. See \cite{Fred19} for some recent developments. 

Perturbative algebraic quantum field theory (pAQFT) \cite{rejzner2016pQFT}, \cite{dutsch2019perturbative}, \cite[\href{https://ncatlab.org/nlab/show/geometry+of+physics+--+perturbative+quantum+field+theory}{nLab}]{perturbative_quantum_field_theory}, due to Brunetti, D\"utsch, Fredenhagen, Hollands, Rejzner, Wald, and others, is (mathematically precise, realistic) pQFT based on causal perturbation theory \cite{steinbook71}, \cite{ep73roleofloc}, \cite{MR1359058}, due to St\"uckelberg, Bogoliubov, Steinmann, Epstein, Glaser, Stora, and others. See \cite[Foreword]{dutsch2019perturbative} for an account of the history. Following \cite{Slavnov78}, \cite{klaus2000micro}, \cite{dutfred00}, in which one takes the algebraic adiabatic limit to handle \hbox{IR-divergences}, pAQFT satisfies the Haag-Kastler axioms, but with \hbox{$\text{C}^\ast$-algebras} replaced by formal power series $\ast$-algebras, reflecting the fact that pQFT deals with Taylor series approximations. It is the construction and structure of these formal power series algebras which is naturally described in terms of gebra theory internal to species. Note that the popular Wilsonian effective field theory interpretation, sometimes called heuristic quantum field theory, is rigorously formulated within pAQFT \cite{dutfred09}, \cite{dut12}, \cite[Section 3.8]{dutsch2019perturbative}, \cite[Section 16]{perturbative_quantum_field_theory}.

For simplicity, we restrict ourselves to the Klein-Gordan real scalar field on Minkowski spacetime $\cX\cong \bR^{p,1}$, $p\in \bN$. In general, pAQFT deals with gauge theories on curved spacetimes \cite{MR2455327}. Therefore for us, an off-shell field configuration $\Phi$ is just any smooth function
\[
\Phi:\cX\to \bR
,\qquad
x\mapsto \Phi(x)
.\] 
%and to be on-shell, $\Phi$ must be in the kernel of the Klein-Gordan differential operator 
%\[
%\square + (m\text{c}/\hbar)^2, \qquad  m\in \bR_{\geq 0}
%.\] 
In particular, we do not impose conditions on the asymptotic behavior of $\Phi$ at infinite times. Let $\mathcal{F}_{\text{loc}}$ denote the space of local observables $\ssA$; these are functionals of field configurations which are obtained by integrating polynomials in $\Phi$ and its derivatives against bump functions on $\cX$. Let $\mathcal{F}$ denote the commutative $\ast$-algebra of microcausal polynomial observables $\emph{\textsf{O}}$; these are polynomial functionals of field configurations satisfying a \hbox{microlocal-theoretic} condition known as microcausality, with multiplication the pointwise multiplication of functionals, sometimes called the normal-ordered product. Then $\mathcal{F}[[\hbar]]$ is a formal power series $\ast$-algebra in formal Planck's constant $\hbar$, called the (abstract, off-shell) Wick algebra, with multiplication the Moyal star product for the Wightman propagator $\Delta_{\text{H}}$ of the Klein-Gordan field
\[
\mathcal{F}[[\hbar]] \otimes \mathcal{F}[[\hbar]] \to \mathcal{F}[[\hbar]]
,\qquad
\emph{\textsf{O}}_1\otimes \emph{\textsf{O}}_2 \mapsto \emph{\textsf{O}}_1 \star_{\text{H}}\! \emph{\textsf{O}}_2
,\] 
sometimes called the operator product. 

%The calculation of the Moyal star product involves the pointwise multiplication of generalized functions, which by H\"ormander's criterion \cite[Theorem 8.2.10]{horm90} is not well-defined without microcausality. 

Perhaps the most fundamental Hopf monoid of Aguiar-Mahajan's theory is the cocommutative Hopf algebra\footnote{\ we allow ourselves to say `algebra' since vector species form a linear category} of compositions $\Sig$ (\autoref{hopfofsetcomp}), which is a Hopf monoid internal to vector species defined with respect to the Day convolution. (More familiar is perhaps its decategorification, which is the graded Hopf algebra of noncommutative symmetric functions $\textbf{NSym}$.) A composition $F$ of $I$ is a surjective function of the form
\[
F:I\to   \{1,\dots,k\}   
,\qquad 
\text{for some} \quad  k\in \bN
.\] 
The ordering $1>\dots>k$ is understood, so that $F$ models the $k^{\text{th}}$ ordinal with $I$-marked points. We let $S_j=F^{-1}(j)$, called the lumps of $F$, and write $F=(S_1,\dots, S_k)$. Each component $\Sig[I]$ is the space of formal linear combinations $\mathtt{a}=\sum_{F} c_F \tH_F$, $c_F\in \Bbbk$, of compositions $F$ of $I$. The multiplication 
\[
\mu_{S,T}: \Sig[S]\otimes \Sig[T]\to \Sig[I]
,\qquad
\tH_F\otimes \tH_G\mapsto \tH_{FG}
\] 
is the linearization of concatenating compositions (`gluing' via ordinal sum), and the comultiplication   
\[\Delta_{S,T}:\Sig[I] \to\Sig[S]\otimes \Sig[T]
,\qquad
\tH_F \mapsto \tH_{F|_S} \otimes \tH_{F|_T}
\] 
is the linearization of restricting compositions to subsets (`forgetting marked points'), where $S\sqcup T=I$. 

Aspects of $\Sig$ appear in the physics literature as follows. Firstly, \hbox{Epstein-Glaser-Stora's} algebra of proper sequences \cite[Section 4.1]{epstein1976general} is the action of $\Sig$ on itself by Hopf powers (\autoref{sec:Tits}), called the Tits product \cite[Section 13]{aguiar2013hopf}, going back to Tits \cite{Tits74}. See \cite[Section 1.4.6]{brown08} for the context of other Dynkin types. The Tits product may be interpreted as the projection of the permutohedron onto its faces \cite[Introduction]{norledge2019hopf}. Secondly, the primitive part $\Zie=\mathcal{P}(\Sig)$\footnote{\ the name `Zie' comes from \cite{aguiar2017topics}} (\autoref{sec:Steain}), which is a Lie algebra internal to species, is essentially the Steinmann algebra from e.g. \cite[Section 6]{Ruelle}, \cite[Section III.1]{bros}. More precisely, the Steinmann algebra is a graded Lie algebra based on the structure map of the adjoint realization of $\Zie$ (\autoref{sec:Ruelle's Identity and the GLZ Relation}). Thirdly and fourthly, and outside the scope of this paper, see below regarding work of Losev-Manin and Feynman integrals.

We formalize the construction of a system of interacting generalized time-ordered products in causal perturbation theory as the construction of a homomorphism $\widetilde{\text{T}}$ of algebras internal to species of the form
\[       
\widetilde{\text{T}}:\Sig  \otimes \textbf{E}_{\mathcal{F}_{\text{loc}}[[\hbar]]}     
\to 
\textbf{U}_{\mathcal{F}[[\hbar, \formg]]} 
.\]
See \autoref{sec:Time-Ordered Products}. Here, $\otimes$ is the Hadamard monoidal product (componentwise tensoring), $\textbf{E}_{\mathcal{F}_{\text{loc}}[[\hbar]]}$ is the species $I\mapsto (\mathcal{F}_{\text{loc}}[[\hbar]])^{\otimes I}$, and $\textbf{U}_{\mathcal{F}[[\hbar, \formg]]}$ is the algebra which has the Wick algebra, with formal coupling constant $\formg$ adjoined, in each $I$-component. The data of $\widetilde{\text{T}}$ is equivalent to a homomorphism of $\bC$-algebras
\[
\hat{\Sig}(\mathcal{F}_{\text{loc}}[[\hbar]])\to \mathcal{F}[[\hbar, \formg]]
\]
where $\hat{\Sig}(-): \textsf{Vec}\to \textsf{Vec}$ is the analytic endofunctor on vector spaces associated to $\Sig$ \cite[Section 19.1.2]{aguiar2010monoidal}, called the Schur functor. Decategorified versions of this formalization appear in the graded Hopf algebra approaches to pQFT \cite{Brouder10}, \cite[p. 635]{Borcherds10}. In particular, there is an interpretation of the Moyal deformation quantization in terms of Laplace pairings (=coquasitriangular structures) \cite{Fauser01}, \cite[Section 2.4]{Brouder10}. 

Also related is the notion of a Losev-Manin cohomological field theory \cite[Theorem 3.3.1]{losevmanin}, \cite[Definition 1.3]{shadrin2011group}, where finite ordinals are replaced by strings of Riemann spheres glued at the poles, giving a Hopf monoid structure on the toric variety of the permutohedron, and $\Sig$ is replaced by the ordinary homology of this toric variety. The Hopf monoid structure of this toric variety is also central to modern approaches to Feynman integrals \cite{MR3713351}, \cite{schultka2018toric}. We shall study this Hopf monoid in future work.

Explicitly, the homomorphism $\widetilde{\text{T}}$ consists of component linear maps
\[       
\widetilde{\text{T}}_I:\Sig[I]  \otimes (\mathcal{F}_{\text{loc}}[[\hbar]])^{\otimes I}
\to 
\mathcal{F}[[\hbar, \formg]]
,\qquad
\tH_F\otimes \ssA_{i_1}\otimes \dots \otimes  \ssA_{i_n} \mapsto  \widetilde{\text{T}}_I(\tH_F\otimes \ssA_{i_1}\otimes \dots \otimes  \ssA_{i_n})
\]
for each finite set $I=\{i_1,\dots, i_n\}$. This homomorphism should also satisfy causal factorization, which says that
\[ 
\widetilde{\text{T}}_I( \mathtt{a} \otimes  \ssA_{i_1}\otimes \dots \ssA_{i_n} )
= 
\widetilde{\text{T}}_I(  \! \! \! \!   \underbrace{\mathtt{a} \triangleright \tH_{G}}_{\text{Tits product}}    \! \! \! \! \otimes  \ssA_{i_1}\otimes \dots \otimes  \ssA_{i_n}  )
\qquad  \text{for all} \quad 
\mathtt{a}\in \Sig[I]     
\]
whenever the local observables $\ssA_{i_1}, \dots ,  \ssA_{i_n}$ respect the ordering of $I$ induced by the composition $G$ (\autoref{prob:causalfac}). Additional properties are often included, such as translation equivariance. 

%\begin{figure}[t]
%	\centering
%	\includegraphics[scale=0.7]{tropperm}
%	\caption{The tropical toric variety which is the toric compactification of the type $A$ root space $\bR^I/(1,\dots,1)=(\bT^\times)^I/\bT^\times$ with respect to the braid arrangement fan. We have shown how the usual $\wt{A}_2$-tiling gets distorted in this picture. On the left is a typical subvariety in blue and its corresponding tropical polynomial underneath, and on the right we have shown the characteristic operation $\tH_{(12,3)}\triangleright (-)$ on the underlying Hopf monoid, see \cite[Introduction]{norledge2019hopf}.}
%	\label{fig:tropperm}
%\end{figure}

We can curry $\widetilde{\text{T}}$ with respect to the internal hom $\cH(-,-)$ for the Hadamard product (\autoref{sec:homforhad}), giving a homomorphism of algebras
\[
\Sig  \to \cH(  \textbf{E}_{\mathcal{F}_{\text{loc}}[[\hbar]]}  ,\textbf{U}_{ \mathcal{F}[[\hbar, \formg]]} )
,\qquad 
\tH_{F}=\tH_{(S_1,\dots, S_k)} \mapsto \widetilde{\text{T}}(S_1)\dots  \widetilde{\text{T}}(S_k)
.\]
The resulting linear maps
\[
\widetilde{\text{T}}(S_1)\dots  \widetilde{\text{T}}(S_k): (\mathcal{F}_{\text{loc}}[[\hbar]])^{\otimes I} \to \mathcal{F}[[\hbar, \formg]]
\]
are called interacting generalized time-ordered products. For each choice of a field polynomial, the curried homomorphism is a `representation' of $\Sig$ as $\mathcal{F}[[\hbar, \formg]]$-valued generalized functions on $\cX^I$, called operator-valued distributions since the Wick algebra is often represented on a Hilbert space. The composition of the time-ordered products $\widetilde{\text{T}}(I)$ with the Hadamard vacuum state
\[
\la - \ra_{0} :\mathcal{F}[[\hbar, \formg]] \to \bC[[\hbar, \formg]]
,\qquad
\emph{\textsf{O}} \mapsto  \emph{\textsf{O}}(\Phi=0)
\]
are then translation invariant $\bC[[\hbar, \formg]]$-valued generalized functions 
\[
\text{G}_I: \cX^I \to     \bC[[\hbar, \formg]]
,\qquad    
(x_{i_1}, \dots, x_{i_n}) \mapsto \text{G}_I(x_{i_1}, \dots, x_{i_n})   \footnote{\ we have used generalized function notation; $\text{G}_I$ is not a single function, but can be represented by a sequence of such functions}
\] 
called time-ordered $n$-point correlation functions, or Green's functions. After taking the adiabatic limit, and in the presence of vacuum stability, these functions may be interpreted as the probabilistic predictions made by the pQFT of the outcomes of scattering experiments, called scattering amplitudes (\autoref{sec:scatterung}). However, their values are formal power series in $\hbar$ and $\formg$, and so have to be truncated. 

Central to Aguiar-Mahajan's work is the interpretation of $\Sig$ (and other Hopf monoids) in terms of the geometry of the type $A$ reflection hyperplane arrangement, called the (essentialized) braid arrangement
\[
\text{Br}[I]=\big\{  \{  x_{i_1}-x_{i_2}=0  \}  \subseteq 
 \! \! \!   \! \! \!  \! 
 \underbrace{\bR^I/\bR \twoheadleftarrow \bR^I}_{\text{quotient by translations}}
\! \! \!  \! \! \!  \! 
 :  (i_1,i_2)\in I^2, \ i_1 \neq i_2  \big   \}
.\]
In causal perturbation theory, the braid arrangement appears as the space of time components of configurations \hbox{$\cX^I$} modulo translational symmetry \cite[Section 2]{Ruelle}, and the reflection hyperplanes are the coinciding interaction points. Every real hyperplane arrangement $\text{A}$ has a corresponding adjoint hyperplane arrangement $\text{A}^\vee$ \cite[Section 1.9.2]{aguiar2017topics}. The free vector space $\bR I$ on $I$ is naturally $\Hom(\bR^I,\bR)$, and so the adjoint of the braid arrangement is given by
\[
\text{Br}^\vee[I]=\bigg\{  \Big\{ \sum_{i\in S} x_i=\sum_{i\in T} x_i=0\Big  \}  \subseteq \underbrace{\Hom(\bR^I/\bR,\bR)\hookrightarrow \bR I}_{\text{sum-zero subspace}}  :  (S,T)\in 2^I,\   S,T\neq \emptyset \bigg  \}
.\]
In causal perturbation theory, the adjoint braid arrangement appears as the space of energy components \cite[Section 2]{Ruelle}, and the hyperplanes correspond to subsets going `on-shell'. The spherical representation of the adjoint braid arrangement is called the Steinmann sphere, or Steinmann planet, e.g. \cite[Figure A.4]{epstein2016}. The chambers of the adjoint braid arrangement are indexed by combinatorial gadgets called cells $\cS$ \cite[Definition 6]{epstein1976general}, also known as maximal unbalanced families \cite{billera2012maximal} and positive sum systems \cite{MR3467341}.

The primitive part Lie algebra $\Zie=\mathcal{P}(\Sig)$ (together with its dual Lie coalgebra $\Zie^\ast$) has a natural geometric realization over the adjoint braid arrangement \cite[Section 6]{Ruelle}, \cite[\href{https://www.youtube.com/watch?v=fUnr0f6mV4c}{Lecture 33}]{oc17}, \cite{lno2019}, \cite{norledge2019hopf}, which results in cells $\cS$ corresponding to certain special primitive elements $\mathtt{D}_\cS\in \Zie[I]$ (\autoref{adjoint}), named Dynkin elements by Aguiar-Mahajan \cite[Section 14.1 and 14.9.8]{aguiar2017topics}. It is shown in \cite{norledge2019hopf} that the Dynkin elements span $\Zie$, but they are not linearly independent. The relations which are satisfied by the Dynkin elements are known as the Steinmann relations \cite[Equation 44]{steinmann1960} (\autoref{stein}), first studied by Steinmann in settings where $\Sig$ is represented as operator-valued distributions. More recently, they have been studied in the context scattering amplitudes, where they appear to be related to cluster algebras \cite{drummond2018cluster}, \cite{caron2019cosmic}, \cite{Caron-Huot:2020bkp}. 

If we restrict a curried system of interacting generalized time-ordered products to the primitive part $\Zie$, then we obtain a Lie algebra homomorphism
\[       
\Zie\to\cH(\textbf{E}_{\mathcal{F}_{\text{loc}}[[\hbar]]},\textbf{U}_{\mathcal{F}[[\hbar, \formg]]}), \qquad \mathtt{D}_\cS \mapsto \widetilde{\text{R}}_\cS
.\]
The operator-valued distributions $\widetilde{\text{R}}_\cS$ which are the images of the Dynkin elements $\mathtt{D}_\cS$ are the interacting generalized retarded products of the system, see e.g. \cite{steinmann1960}, \cite{Huz1}, \cite[Equation 79]{ep73roleofloc}. 

Let $\textbf{L}\hookrightarrow \Sig$ be the Hopf subalgebra of linear orders (=compositions with singleton lumps), and let $\textbf{E}^\ast\hookrightarrow \Sig$ be the subcoalgebra of compositions with one lump. Then we have the following dictionary between products/vacuum expectation values and elements of $\Sig$. 
\bgroup
\def\arraystretch{1} 
\begin{table}[H] 
\begin{tabular}{|c|c|c|c|}
\hline
&
\begin{tabular}{@{}c@{}}spanning set\end{tabular}&
\begin{tabular}{@{}c@{}}operator-valued distributions\end{tabular}&
\begin{tabular}{@{}c@{}}vacuum expectation values\end{tabular}\\ \hline

\begin{tabular}{@{}c@{}}$\textbf{E}^\ast$\end{tabular}&   
\begin{tabular}{@{}c@{}}universal series\\$\mathtt{G}_I$\end{tabular}&
\begin{tabular}{@{}c@{}}time-ordered product\\ $\text{T}(I)$\end{tabular}&
\begin{tabular}{@{}c@{}}time-ordered $n$-point\\ function\end{tabular}\\ \hline

\begin{tabular}{@{}c@{}}$\textbf{L}$\end{tabular}&
\begin{tabular}{@{}c@{}}$\tH$-basis linear orders\\$\tH_\ell$\end{tabular}&
\begin{tabular}{@{}c@{}}$\text{T}(i_1)\dots \text{T}(i_n)$\end{tabular}& 
\begin{tabular}{@{}c@{}}Wightman $n$-point\\ functions \end{tabular}\\ \hline

\begin{tabular}{@{}c@{}}$\Sig$\end{tabular}&
\begin{tabular}{@{}c@{}}$\tH$-basis set compositions\\$\tH_F$\end{tabular}&
\begin{tabular}{@{}c@{}}generalized time-ordered products\\$\text{T}(S_1)\dots \text{T}(S_k)$\end{tabular}&
\begin{tabular}{@{}c@{}}generalized time-ordered\\ functions\end{tabular}\\ \hline

\begin{tabular}{@{}c@{}}$\Zie$\end{tabular}&
\begin{tabular}{@{}c@{}}Dynkin elements\\$\mathtt{D}_\cS$\end{tabular}&
\begin{tabular}{@{}c@{}}generalized retarded products\\$\text{R}_\cS$\end{tabular}&
\begin{tabular}{@{}c@{}}generalized retarded\\ functions\end{tabular}\\ \hline
\end{tabular}
\end{table}
\egroup 

We also formalize the construction of interacting products in casual perturbation theory as follows. Our starting point is a fully normalized system of generalized \hbox{time-ordered} products, that is a homomorphism of algebras
\[
\text{T}:\Sig\otimes \textbf{E}_{\mathcal{F}_\text{loc} [[\hbar]]}\to\textbf{U}_{\mathcal{F}((\hbar))}  
\]
satisfying causal factorization, and such that the singleton components $\text{T}_{\{i\}}$ are the natural inclusion 
\[
\mathcal{F}_\text{loc} [[\hbar]]\hookrightarrow  \mathcal{F}((\hbar))
,\qquad
\ssA \mapsto \, \,  :\! \ssA :  .
\]
The corresponding operator-valued distributions are determined everywhere on $\cX^I$ by causal factorization, apart from on the fat diagonal (coinciding interaction points). In particular, off the fat diagonal, the time-ordered products $\text{T}(I)$ are given by the Moyal star product $\star_{\text{F}}$ with respect to the Feynman propagator $\Delta_{\text{F}}$ for the Klein-Gordon field. The terms of the product $\star_{\text{F}}$ may be encoded in finite multigraphs, i.e. Feynman graphs. The remaining inherent ambiguity means one has to make choices when extending the $\text{T}(I)$ to the fat diagonal, and these choices form a torsor of the \hbox{St\"uckelberg-Petermann} renormalization group. This is Stora's elaboration \cite{stora16}, \cite{stora1993differential}, \cite{klaus2000micro} on \hbox{St\"uckelberg-Bogoliubov-Epstein-Glaser} normalization \cite{ep73roleofloc}, which constructs the $\text{T}(I)$ inductively in $n=|I|$. We leave species-theoretic aspects of renormalization, and possible connections to \hbox{Connes-Kreimer} theory \cite{MR1845168}, \cite{Bondia00}, \cite{Kreimer05}, \cite{FredHopf14}, to future work.

%\footnote{\ alternatively, to define the homomorphism $\text{T}$ one can normalize the operator-valued distributions which are the images of the Dynkin elements $\mathtt{D}_i\in \Sig[Y\sqcup \{i\}]$, as in \cite{steinbook71}, \cite{dutfredretard04}, \cite{dutsch2019perturbative}, called total retarded products $\text{R}(Y,i)$; unlike $\tH_{(I)}$, the $\mathtt{D}_i$ do not freely generate $\Sig$, and so one must include the GLZ relation \cite[Equation 11]{GLZ1957} (\autoref{sec:Ruelle's Identity and the GLZ Relation})} 

In the original formulation by Tomonaga, Schwinger, Feynman and Dyson, would-be \hbox{time-ordered} products are obtained by informally multiplying Wick algebra products by step functions, which is in general ill-defined by H\"ormander's criterion. This leads to individual terms of the formal power series diverging, called UV-divergences. Then informal methods are used to obtain finite values from these infinite terms \cite[Preface and Section 4.3]{MR1359058}.

The so-called exponential species $\textbf{E}$, given by $\textbf{E}[I]=\Bbbk$ with basis element $1_\Bbbk\in \textbf{E}[I]$ denoted $\tH_I$, has the structure of an algebra in species by linearizing taking unions of sets,
\[
\mu_{S,T}: \textbf{E}[S]\otimes \textbf{E}[T]\to \textbf{E}[I]
,\qquad
\tH_S\otimes \tH_T\mapsto \tH_{I}
.\] 
An $\textbf{E}$-module $\textbf{m}=(\textbf{m},\rho)$ is an associative and unital morphism 
\[\rho:\textbf{E}\bigcdot\textbf{m}\to \textbf{m}\] 
for $\textbf{m}$ a species (\autoref{sec:Modules}). Moreover, taking the inverse of $\mu_{S,T}$ as the comultiplication turns $\textbf{E}$ into a connected (co)commutative bialgebra, and so the category of $\textbf{E}$-modules $\textsf{Rep}(\textbf{E})$ is a symmetric monoidal category with monoidal product the Cauchy product of $\textbf{E}$-modules. In particular, we may consider Hopf/Lie algebras internal to $\textsf{Rep}(\textbf{E})$, which we call Hopf/Lie \hbox{$\textbf{E}$-algebras}.

The retarded $Y\downarrow(-)$ and advanced $Y\uparrow(-)$ Steinmann arrows are (we formalize as) raising operators on $\Sig$, whose precise definition is due to \hbox{Epstein-Glaser-Stora} \cite[p.82-83]{epstein1976general}. They define two $\textbf{E}$-module structures on $\Sig$ (\autoref{sec:The Steinmann Arrows}), 
\[
\textbf{E}\bigcdot \Sig \to \Sig 
,\quad
\tH_Y \otimes \tH_F \mapsto Y  \downarrow\tH_F
\qquad \text{and} \qquad 
\textbf{E}\bigcdot \Sig \to \Sig
,\quad
\tH_Y \otimes \tH_F \mapsto Y  \uparrow \tH_F
.\]
In particular, the retarded arrow is generated by putting $\{\ast\} \downarrow \tH_{(I)}= -\tH_{(\ast, I)} +\tH_{(\ast I)}$.\footnote{\ $(\ast I)$ denotes the composition of $\{ \ast \}\sqcup I$ which has a single lump} Then
\[
Y\! \downarrow \tH_{(I)}= \underbrace{\sum_{Y_1\sqcup Y_2=Y} \mu_{Y_1, Y_2\sqcup I}\big ( \text{s}(\tH_{(Y_1)}) \otimes \tH_{(Y_2\sqcup I)} \big )}_{\text{denoted $\mathtt{R}_{(Y;I)}$}}
\]
where $\text{s}:\Sig \to \Sig$ is the antipode of $\Sig$. The Steinmann arrows were first studied by Steinmann \cite[Section 3]{steinmann1960} in settings where $\Sig$ is represented as operator-valued distributions. Here, the \hbox{operator-valued} distribution which is the image of $\mathtt{R}_{(Y;I)}\in \Sig[Y\sqcup I]$ is called the retarded product $\text{R}(Y;I)$.\footnote{\ note that some authors, e.g. \cite{dutsch2019perturbative}, call $\text{R}(Y;i)$ the retarded product, and then call $\text{R}(Y;I)$ the generalized retarded product}

Since $\{\ast\} \downarrow (-)$ is a commutative biderivation of $\Sig$ (\autoref{steinmannarrowaredercoder}), the retarded Steinmann arrow gives $\Sig$ the structure of a Hopf $\textbf{E}$-algebra, and $\Zie$ the structure of a Lie $\textbf{E}$-algebra (similarly for the advanced arrow). There is an interesting description of these Lie $\textbf{E}$-algebras in terms of the adjoint braid arrangement (\autoref{sec:The Steinmann Arrows and Dynkin Elements}). The Steinmann arrows are ``two halves'' of the restricted adjoint representation $\textbf{L}\bigcdot \Sig \to \Sig$ of $\Sig$, which is reflected in \cite[Equation 13]{steinmann1960}. This directly corresponds to how the retarded $\Delta_-$ and advanced $\Delta_+$ propagators are two halves of the causal propagator $\Delta_{\text{S}}=\Delta_+ - \Delta_-$. 

Let $\cH^{\bigcdot}(-,-)$ be the internal hom for the Cauchy product of species (\autoref{sec:inthomcauchy}), and let 
\[
(-)^{\textbf{E}}=\cH^{\bigcdot} ( \textbf{E} , -)
.\] 
Then $(-)^{\textbf{E}}$ is an endofunctor on species, which is lax monoidal with respect to the Cauchy product (\autoref{sec:theendoE}). Therefore $\Sig^{\textbf{E}}$ is naturally an algebra, with multiplication inherited from $\Sig$. Then, by currying the retarded Steinmann action 
\[
(-)\downarrow(-): \textbf{E}\bigcdot \Sig \to \Sig
\] 
we obtain a homomorphism $\Sig \to \Sig^{\textbf{E}}$. Similarly for the setting with decorations, given a choice of adiabatically switched interaction action functional $\ssS_{\text{int}}\in \mathcal{F}_{\text{loc}}[[\hbar]]$, we obtain the homomorphism (\autoref{prop:isEalg})
\begin{align*}
\Sig \otimes \textbf{E}_{\mathcal{F}_{\text{loc}}[[\hbar]]} &\to    (\Sig \otimes \textbf{E}_{\mathcal{F}_{\text{loc}}[[\hbar]]})^{\textbf{E}}\\[6pt]
\tH_F\otimes \ssA_{i_1}\otimes \dots \otimes \ssA_{i_n} &\mapsto \,    \sum_{r=0}^\infty   \underbrace{\downarrow\dots  \downarrow}_{\text{$r$ times}} \tH_F  \otimes  \underbrace{\ssS_{\text{int}}\otimes \dots \otimes \ssS_{\text{int}}}_{\text{$r$ times}} \,  \otimes\,    \ssA_{i_1}\otimes \dots \otimes \ssA_{i_n}.
\end{align*}
Compare this with the formalism for creation-annihilation operators in \cite[Chapter 19]{aguiar2010monoidal}. Then, finally, the corresponding system of perturbed interacting time-ordered products $\widetilde{\text{T}}$ is given by composing this homomorphism with the image of $\text{T}$ under the endofunctor $(-)^{\textbf{E}}$ (\autoref{sec:Perturbation of T-Products by Steinmann Arrows}),
\[
\widetilde{\text{T}} : \Sig\otimes \textbf{E}_{\mathcal{F}_\text{loc} [[\hbar]]} \to   (\Sig\otimes \textbf{E}_{\mathcal{F}_\text{loc} [[\hbar]]})^{\textbf{E}}   \xrightarrow{\text{T}^{\textbf{E}}}        (\textbf{U}_{ \mathcal{F}((\hbar))})^{\textbf{E}} \cong   \textbf{U}_{\mathcal{F}((\hbar))[[\formg]]}
.\]
It is a theorem of pAQFT that this does indeed land in $\textbf{U}_{\mathcal{F}[[\hbar,\formg]]}$.

We also formalize S-matrices, or time-ordered exponentials. Let $\Hom(-,-)$ denote the external hom for species, which lands in vector spaces $\sfVec$ (\autoref{sec:externalhom}). We let 
\[
\mathscr{S}(-)=\Hom(\textbf{E},-)
.\] 
This is lax monoidal with respect to the Cauchy product (\autoref{sec:Series of Species}). In the presence of a generic system of products on an algebra $\textbf{a}$,
\[
\varphi:\textbf{a}\otimes \textbf{E}_V\to \textbf{U}_\cA,
\] 
series $\mathtt{s}\in \mathscr{S}(\textbf{a})$ of $\textbf{a}$
\[\mathtt{s}:\textbf{E}\to\textbf{a}
,\qquad 
\tH_I\mapsto \mathtt{s}_I\] 
induce $\mathscr{S}(\textbf{U}_\cA)\cong \cA[[\formj]]$-valued functions on $V$ as follows (\autoref{sec:decseries}),
\[
\mathcal{S}_{\mathtt{s}} :    V \to \cA[[\formj]]
,\qquad
\ssA \mapsto   \mathcal{S}_{\mathtt{s}}(\formj\! \ssA) =  \sum_{n=0}^{\infty}  \dfrac{\formj^n}{n!}  \varphi_{n} ( \mathtt{s}_{n} \otimes \underbrace{\ssA\otimes \dots \otimes \ssA}_{\text{$n$ times}})
.\]
If $\varphi$ is a homomorphism of algebras, then 
\[
\mathcal{S}_{(-)}: \mathscr{S}(\textbf{a})\to \text{Func}(V, \cA[[\formj]])
\] 
is a homomorphism of $\Bbbk$-algebras (\autoref{homo}). As a basic example, if we put $\textbf{a}=\textbf{E}$, $\cA=C^\infty(V^\ast)$, and set $\formj=1$ at the end, then one can recover the classical exponential function in this way (\autoref{ex:polyfun}). 

For $c\in \Bbbk$, the so-called (scaled) universal series $\mathtt{G}(c)$ of $\Sig$ is given by sending each finite set to the (scaled) composition with one lump,
\[   
\mathtt{G}(c):  \textbf{E} \to  \Sig
,\qquad  
\tH_{I}\mapsto \mathtt{G}(c)_{I}:= c^n\,  \tH_{(I)}    
.\]
If we set $c=1/\text{i}\hbar$, then the function $\mathcal{S}=\mathcal{S}_{\mathtt{G}(1/\text{i}\hbar)}$ above for a fully normalized system of generalized time-ordered products $\text{T}:\Sig\otimes \textbf{E}_{\mathcal{F}_\text{loc} [[\hbar]]}\to\textbf{U}_{\mathcal{F}((\hbar))}$ recovers the usual perturbative S-matrix scheme of pAQFT,
\[
\mathcal{S}:\mathcal{F}_{\text{loc}}[[\hbar]]\to\mathcal{F}((\hbar))[[\formj]]
,\qquad
\ssA \mapsto   \mathcal{S}(\formj\!  \ssA) =  \sum_{n=0}^{\infty}  \bigg( \dfrac{1}{\text{i} \hbar} \bigg)^n \dfrac{\formj^n}{n!}  \text{T}_{n} ( \tH_{(n)} \otimes \underbrace{\ssA\otimes \dots \otimes \ssA}_{\text{$n$ times}})
.\] 
The image of $\mathcal{S}(\formj\! \ssA)$ after applying perturbation by the retarded Steinmann arrow $(-)\downarrow (-)$ and a choice of interaction $\ssS_{\text{int}}\in \mathcal{F}_{\text{loc}}[[\hbar]]$ is
\[
\mathcal{Z}_{\formg \ssS_{\text{int}}}(\formj\! \ssA)
=
\sum_{n=0}^\infty \sum_{r=0}^\infty 
\bigg(\dfrac{1}{\text{i}\hbar}\bigg)^{\! r+n}
\dfrac{\formg^{r} \formj^n}{r!\, n!}\,  \text{R}_{r;n} (\underbrace{\ssS_{\text{int}}\otimes \dots \otimes \ssS_{\text{int}}}_{\text{$r$ times}}\,   ;\,  \underbrace{\ssA\otimes \dots \otimes \ssA}_{\text{$n$ times}} )    
\]
where, by our previous expression for $\mathtt{R}_{(Y;I)}=Y\downarrow \tH_{(I)}$ (and letting $\overline{\text{T}}$ denote the precomposition of $\text{T}$ with the antipode of $\Sig \otimes \textbf{E}_{\mathcal{F}_{\text{loc}}[[\hbar]]}$), we have
\[
\text{R}_{Y;I}(\ssS_{\text{int}}^{\, Y};\ssA^I)
= 
\text{T}_{Y\sqcup I}(Y\downarrow \tH_{(I)} \otimes   \ssS_{\text{int}}^{\, Y} \otimes \ssA^I)
=
\sum_{Y_1 \sqcup Y_2=Y}  \overline{\text{T}}_{Y_1}(\ssS_{\text{int}}^{\, Y_1}) \star_{\text{H}} \text{T}_{Y_2\sqcup I}( \ssS_{\text{int}}^{\, Y_2}\otimes  \ssA^I)
.\]
Then, since
\[
\mathcal{S}_{(-)}:\mathscr{S}(\Sig)\to \text{Func}\big (\mathcal{F}_{\text{loc}}[[\hbar]] , \mathcal{F}((\hbar))[[\formg]]\big )
\] 
is a homomorphism of $\Bbbk$-algebras, it follows that $\mathcal{Z}_{\formg \ssS_{\text{int}}}$ is given by
\[
\mathcal{Z}_{\formg \ssS_{\text{int}}}(\formj\! \ssA)
=
\mathcal{S}^{-1}( \formg \ssS_{\text{int}})\star_{\text{H}} \mathcal{S}(\formg \ssS_{\text{int}} +\formj\! \ssA )
.\]
This is the generating function, or partition function, for time-ordered products of interacting field observables, see e.g. \cite[Section 8.1]{ep73roleofloc}, \cite[Section 6.2]{dutfred00}, going back to Bogoliubov \cite[Chapter 4]{Bogoliubov59}. In this paper, we arrive at the generating function $\mathcal{Z}_{\formg \ssS_{\text{int}}}$ through purely Hopf-theoretic considerations. However, it was originally motivated by attempts to make sense of the path integral synthetically. For some recent developments, see \cite{collini2016fedosov}, \cite{MR4109798}.

%%%%%%%%%%%%%%%%%%%%%%%%%%%%%%%%%%%%%%%%%%%%%%%%%%%%%%%%%%%%%%%%%%%%%%%%
\subsection*{Structure.} 
%%%%%%%%%%%%%%%%%%%%%%%%%%%%%%%%%%%%%%%%%%%%%%%%%%%%%%%%%%%%%%%%%%%%%%%%

This paper is divided into three parts. In part one we develop general theory of species, in part two we focus on the theory for the Hopf algebra of compositions $\Sig$ and its primitive part $\Zie$, and in part three we specialize to pAQFT for the case of a real scalar field on Minkowski spacetime.

%By translation invariance, or equivalently momentum conservation, the generalized time-ordered functions are translation invariant. If we restrict to the time component, then we obtain a realization of $\Sig$ as distributions on the braid arrangement, or functions on the adjoint braid arrangement. 

%%%%%%%%%%%%%%%%%%%%%%%%%%%%%%%%%%%%%%%%%%%%%%%%%%%%%%%%%%%%%%%%%%%%%%%%
\subsection*{Acknowledgments.} 
%%%%%%%%%%%%%%%%%%%%%%%%%%%%%%%%%%%%%%%%%%%%%%%%%%%%%%%%%%%%%%%%%%%%%%%%

We thank Adrian Ocneanu for his support and useful discussions. This paper would not have been written without Nick Early's discovery that certain relations appearing in Ocneanu's work were known in quantum field theory as the Steinmann relations. We thank Yiannis Loizides and Maria Teresa Chiri for helpful discussions during an early stage of this project. We thank Arthur Jaffe for his support, useful suggestions, and encouragement to pursue this topic. We thank Penn State maths department for their continued support. 

%%%%%%%%%%%%%%%%%%%%%%%%%%%%%%%%%%%%%%%%%%%%%%%%%%%%%%%%%%%%%%%%%%%%%%%%
%%%%%%%%%%%%%%%%%%%%%%%%%%%%%%%%%%%%%%%%%%%%%%%%%%%%%%%%%%%%%%%%%%%%%%%%
\part{General Theory}
%%%%%%%%%%%%%%%%%%%%%%%%%%%%%%%%%%%%%%%%%%%%%%%%%%%%%%%%%%%%%%%%%%%%%%%%

We begin by recalling some basic aspects of species, and of Aguiar-Mahajan's theory of gebras internal to species. We recall the various homs on species and define associated structures, in particular we define the endomorphism algebra of raising operators on a species (i.e. higher up operators). We develop theory for modules internal to species, constructed with respect to the Cauchy product. We also develop theory for decorated series of species, and introduce the notion of a system of products for a species. We show that systems of products for modules over the exponential species may be `perturbed' in a certain sense.

%%%%%%%%%%%%%%%%%%%%%%%%%%%%%%%%%%%%%%%%%%%%%%%%%%%%%%%%%%%%%%%%%%%%%%%%
\section{Preliminaries} 
%%%%%%%%%%%%%%%%%%%%%%%%%%%%%%%%%%%%%%%%%%%%%%%%%%%%%%%%%%%%%%%%%%%%%%%%
%%%%%%%%%%%%%%%%%%%%%%%%%%%%%%%%%%%%%%%%%%%%%%%%%%%%%%%%%%%%%%%%%%%%%%%%
%%%%%%%%%%%%%%%%%%%%%%%%%%%%%%%%%%%%%%%%%%%%%%%%%%%%%%%%%%%%%%%%%%%%%%%%
%%%%%%%%%%%%%%%%%%%%%%%%%%%%%%%%%%%%%%%%%%%%%%%%%%%%%%%%%%%%%%%%%%%%%%%%

We recall species, and gebras internal to species, following \cite{aguiar2010monoidal}, \cite{aguiar2013hopf}.

%%%%%%%%%%%%%%%%%%%%%%%%%%%%%%%%%%%%%%%%%%%%%%%%%%%%%%%%%%%%%%%%%%%%%%%%
\subsection{Species}
%%%%%%%%%%%%%%%%%%%%%%%%%%%%%%%%%%%%%%%%%%%%%%%%%%%%%%%%%%%%%%%%%%%%%%%% 

Let $\textsf{Vec}$ denote the category of vector spaces over a field $\Bbbk$ of characteristic zero. 
Let $\sS$ denote the monoidal category with objects finite sets $I,J\dots$, morphisms bijective functions $\sigma:J\to I$, and monoidal product the restriction of the disjoint union of sets to finite sets. We let $n:=|I|$ throughout. We pick a section of the decategorification functor 
\[
\sfS\to \bN
,\qquad 
I\mapsto n
\] 
by letting $[n]$ denote the finite set of integers $[n]:=\{1,\dots, n\}$, for $n\in \bN$. A (vector) \emph{species} $\textbf{p}$ is a presheaf of vector spaces on $\sS$, denoted
\[  
\textbf{p}:\sS^{\op}\to\textsf{Vec},\qquad  I\mapsto\textbf{p}[I] 
\quad\text{and}\quad  
\sigma \mapsto \textbf{p}[\sigma]
.\]
We abbreviate $\textbf{p}[n]:=\textbf{p}\big[\{1,\dots,n\}\big]$. Explicitly, to every finite set $I$ we assign a vector space $\textbf{p}[I]$, and to every bijection of finite sets $\sigma:J\to I$ we assign a (bijective) linear map $\textbf{p}[\sigma]:\textbf{p}[I]\to \textbf{p}[J]$ such that identities and the composition of bijections are preserved. Often, $\textbf{p}[I]$ is the collection of formal linear combinations of labelings/`probes' of a set of combinatorial objects by $I$, and $\textbf{p}[\sigma]:\textbf{p}[I]\to \textbf{p}[J]$ sends an \hbox{$I$-labeling} to its precomposition with $\sigma$. We shall tend to define species by giving their values on finite sets only, with their values on bijections then being induced in an obvious way. 

A species $\textbf{p}$ is called \emph{positive} if $\textbf{p}[\emptyset]=0$, and \emph{connected} if $\textbf{p}[\emptyset]\cong \Bbbk$. Every species $\textbf{p}$ determines a positive species $\textbf{p}_+$ by putting $\textbf{p}_+[I]:=\textbf{p}[I]$ for $I$ nonempty. We denote the functor category of species by $[\sfS^{\op},\sfVec]$. Explicitly, a \emph{morphism} of species $\eta:\textbf{p}\to \textbf{q}$ is a family of linear maps 
\[
\eta_I:\textbf{p}[I]\to \textbf{q}[I],\qquad  I\in \sfS
\] 
such that $\eta_J\circ \textbf{p}[\sigma]=\textbf{q}[\sigma]\circ \eta_I$ for all bijections $\sigma:J\to I$ in $\sfS$. We abbreviate $\eta_n:=\eta_{[n]}$.

%%%%%%%%%%%%%%%%%%%%%%%%%%%%%%%%%%%%%%%%%%%%%%%%%%%%%%%%%%%%%%%%%%%%%%%%
\subsection{Monoidal Products}\label{sec:Monoidal Products}
%%%%%%%%%%%%%%%%%%%%%%%%%%%%%%%%%%%%%%%%%%%%%%%%%%%%%%%%%%%%%%%%%%%%%%%%

We now equip the category of species $[\sfS^{\op},\textsf{Vec}]$ with three monoidal products.

The category $[\sfS^{\op},\textsf{Vec}]$ may be equipped with the symmetric monoidal product known as the Day convolution, defined with respect to the disjoint union of finite sets and the tensor product of vector spaces. This is often called the Cauchy product of species. Let us write $S\sqcup T=I$ to indicate an ordered pair $(S,T)$ such that
\[
S,T\subseteq I, \qquad S\cap T=\emptyset
,\qquad 
S\cup T=I
.\] 
Given species $\textbf{p}$ and $\textbf{q}$, their \emph{Cauchy product} $\textbf{p}\bigcdot \textbf{q}$ is the species given by
\[    
\textbf{p}\bigcdot  \textbf{q}[I]
:=     
\bigoplus_{S\sqcup T=I} \textbf{p}[S] \otimes \textbf{q}[T]    
.\]
This extends in an obvious way to a symmetric monoidal product for $[\sfS^{\op},\textsf{Vec}]$. The unit for the Cauchy product is the species $\textbf{1}$, given by
\[
\textbf{1}[I]:=
\begin{cases}
\Bbbk &\quad \text{if}\ I=\emptyset \\
0 &\quad \text{otherwise}. 
\end{cases}
\]
Given a morphism of species $\eta:\textbf{p}\bigcdot \textbf{q}\to \textbf{r}$, we denote the restriction of $\eta_I$ to $\textbf{p}[S] \otimes \textbf{q}[T]$ by $\eta_{S,T}$. Given a morphism of species $\eta:\textbf{r}\to \textbf{p} \bigcdot \textbf{q}$, we denote the composition of $\eta_I$ with the orthogonal projection onto $\textbf{p}[S] \otimes \textbf{q}[T]$ by $\eta_{S,T}$. 

Given species $\textbf{p}$ and $\textbf{q}$, their \emph{Hadamard product} $\textbf{p}\otimes \textbf{q}$ is the species given by
\[     
\textbf{p}\otimes  \textbf{q}[I]:= \textbf{p}[I]\otimes  \textbf{q}[I]
.\] 
This extends to a second symmetric monoidal product for $[\sfS^{\op},\textsf{Vec}]$. The unit for the Hadamard product is the \emph{exponential species} $\textbf{E}$, given by
\[
\textbf{E}[I]:= \Bbbk
\qquad \text{for all} \quad I\in \sfS
\qquad \quad  \text{and} \qquad \quad 
\textbf{E}[\sigma]:=\text{id}_\Bbbk 
\qquad \text{for all} \quad \sigma\in \sfS
.\]
We let $\tH_I$ denote the unit $1_\Bbbk\in \textbf{E}[I]$. 

Given species $\textbf{p}$ and $\textbf{q}$, with $\textbf{q}[\emptyset]=0$, their \emph{plethystic product} $\textbf{p}\boldsymbol{\circ} \textbf{q}$, also known as composition or substitution, is the species given by
\[    
\textbf{p}\boldsymbol{\circ} \textbf{q}[I]:= \bigoplus_{ P } \textbf{p}[P] \otimes \bigotimes_{S_j\in P} \textbf{q}[S_j] 
.\]
The direct sum is over all set partitions $P=\{S_1, \dots , S_k\}$ of $I$ (the $S_j\subseteq I$ are disjoint, nonempty, and their union is $I$). This extends to a third monoidal product for $[\sfS^{\op},\textsf{Vec}]$, and we direct the reader to \cite[Appendix B.5]{aguiar2010monoidal} for important subtleties here. The unit for the plethystic product is the species $\textbf{X}$, given by
\[ 
\textbf{X}[I]=\begin{cases}
\Bbbk  &\quad  \text{if }I\text{ is a singleton} \\
0 &\quad \text{otherwise}.
\end{cases}   
\]  
We let $\tH_i$ denote the unit $1_\Bbbk\in \textbf{X}\big [\{ i \}\big ]$. Finally, we denote the category-theoretic coproduct (and product) of species by $\oplus$, which is given by
\[
\textbf{p}\oplus \textbf{q}[I]:= \textbf{p}[I]\oplus  \textbf{q}[I]
.\]
The unit for the coproduct is the species $\textbf{0}$, given by
\[
\textbf{0}[I]:=0 \qquad \quad \text{for all} \quad  I\in \sfS
.\]

\begin{remark}
Species are equivalent to analytic endofunctors on $\textsf{Vec}$ via a categorified version of the $\bZ$-transform \cite[Section 19.1.2]{aguiar2010monoidal}. The Cauchy convolution product is induced by pointwise multiplying endofunctors, and the plethystic product is induced by composing endofunctors.
\end{remark}

%%%%%%%%%%%%%%%%%%%%%%%%%%%%%%%%%%%%%%%%%%%%%%%%%%%%%%%%%%%%%%%%%%%%%%%%
\subsection{Gebras in Species}
%%%%%%%%%%%%%%%%%%%%%%%%%%%%%%%%%%%%%%%%%%%%%%%%%%%%%%%%%%%%%%%%%%%%%%%%

Let a \emph{(co/bi/Hopf/Lie)algebra in species} be a (co/bi/Hopf)monoid/Lie algebra internal to species $[\sfS^{\op},\textsf{Vec}]$, constructed with respect to the Cauchy product.\footnote{\ regarding the other monoidal products, gebras constructed with respect to the Hadamard product are equivalently presheaves of $\Bbbk$-gebras on $\sfS$, and (co)monoids constructed with respect to the plethystic product are linear (co)operads by another name \cite[Appendix B.5]{aguiar2010monoidal}} We now make these definitions explicit.

An algebra in species $\textbf{a}=(\textbf{a},\mu,\iota)$ consists of two morphisms, the \emph{multiplication} $\mu: \textbf{a}\bigcdot \textbf{a} \to \textbf{a}$ and the \emph{unit} $\iota:\textbf{1}\to \textbf{a}$, which satisfy associativity and unitality. Explicitly, we have linear maps
\[   
\mu_{S,T}:  \textbf{a}[S]\otimes \textbf{a}[T]\to \textbf{a}[I] 
\qquad \text{and} \qquad
\iota_{\emptyset}  :\Bbbk\to  \textbf{a}[\emptyset]
\]
where $S\sqcup T=I$, $I\in \sfS$, which commute with the action of $\sfS$,\footnote{\ that is $\textbf{a}[\sigma]\big(\mu_{S,T}(\mathtt{a}\otimes \mathtt{b})\big)=\mu_{\sigma(S),\sigma(T)}\big( \textbf{a}[\sigma|_S](\mathtt{a}),\textbf{a}[\sigma|_T](\mathtt{b})\big)$} 
and satisfy associativity
\[\mu_{S,T\sqcup U}\big( \mathtt{a}  \otimes  \mu_{T,U}(\mathtt{b}\otimes \mathtt{c})\big)
=
\mu_{S\sqcup T,U}\big( \mu_{S,T}( \mathtt{a}  \otimes  \mathtt{b} )\otimes \mathtt{c}\big)
\]
and unitality
\[
\mu_{I,\emptyset}( \mathtt{a} \otimes \mathtt{1}_\textbf{a} )
=
\mu_{\emptyset,I}(\mathtt{1}_\textbf{a}\otimes \mathtt{a} )
=
\mathtt{a}
\]
where $\mathtt{1}_\textbf{a}:=\iota_{\emptyset}(1_\Bbbk)$.
  
A coalgebra in species $\textbf{c}=(\textbf{c},\Delta,\epsilon)$ consists of two morphisms, the \emph{comultiplication} $\Delta:\textbf{c}\to \textbf{c}\bigcdot \textbf{c}$ and the \emph{counit} $\epsilon:\textbf{c}\to \textbf{1}$, which satisfy coassociativity and counitality. Explicitly, we have linear maps
\[   \Delta_{S,T}: \textbf{c}[I]\to \textbf{c}[S]\otimes \textbf{c}[T]
\qquad \text{and} \qquad
\epsilon_{\emptyset}:  \textbf{c}[\emptyset]\to \Bbbk 
\]
where $S\sqcup T=I$, $I\in \sfS$, which commute with the action of $\sfS$, and satisfy coassociativity
\[
(\Delta_{S,T}\otimes \text{id})   \circ  \Delta_{S\sqcup T,U}(\mathtt{a})
=
(\text{id}\otimes \Delta_{T,U})   \circ  \Delta_{S,T\sqcup U}(\mathtt{a})
\]
and counitality 
\[
(\text{id}\otimes  \epsilon) \circ \Delta_{I,\emptyset}(\mathtt{a})
=
(\epsilon\otimes  \text{id}) \circ \Delta_{\emptyset,I}(\mathtt{a})
=
\mathtt{a} 
\]
where we have identified $\textbf{a}[I]\otimes \Bbbk= \Bbbk\otimes\textbf{a}[I]=\textbf{a}[I]$.
 
A bialgebra in species $\textbf{h}=(\textbf{h}, \mu, \Delta, \iota, \epsilon)$ simultaneously has the structure of an algebra and a coalgebra, such that the usual four compatibility conditions are satisfied, see \cite[Section 8.3.1]{aguiar2010monoidal}. A Hopf algebra in species is a bialgebra in species such that there exists an additional morphism $\text{s}:\textbf{h}\to \textbf{h}$ called the \emph{antipode}. Following \cite[Section 2.8]{aguiar2013hopf}, the antipode is a family of linear maps
\[ 
\text{s}_I:  \textbf{h}[I]\to  \textbf{h}[I], \qquad I\in \sfS
\]
such that $\text{s}_{\emptyset}:\textbf{h}[\emptyset]\to \textbf{h}[\emptyset]$ is an antipode for the induced $\Bbbk$-bialgebra $\textbf{h}[\emptyset]$, and for each nonempty $I$, we have  
\[
\sum_{S\sqcup T=I} \mu_{S,T}\circ ( \text{id} \otimes \text{s}_T )\circ  \Delta_{S,T} =0
\qquad \text{and} \qquad
\sum_{S\sqcup T=I}\mu_{S,T}\circ (\text{s}_S \otimes \text{id})\circ \Delta_{S,T} =0
.\]
%All our Hopf algebras in species will be connected, and so  $\text{s}_{\emptyset}$ will be trivial. 
After \textcolor{blue}{(\refeqq{eq:convolprod})} below, we notice that these conditions equivalently say that $\text{s}:\textbf{h}\to \textbf{h}$ is the inverse of the identity map $\text{id}_\textbf{h}:\textbf{h}\to \textbf{h}$ in the convolution $\Bbbk$-algebra $\Hom(\textbf{h},\textbf{h})$. This inverse exists if and only if the antipode $\text{s}_{\emptyset}$ exists \cite[Proposition 1]{aguiar2013hopf}.  

To define a bialgebra structure on a connected species, it suffices to specify the linear maps $\mu_{S,T}$ and $\Delta_{S,T}$ for $S$ and $T$ nonempty, with everything else then being determined. This connected bialgebra is then necessarily a Hopf algebra. 

A Lie algebra in species $\textbf{g}=(\textbf{g}, \partial^\ast)$ consists of a single morphism $\partial^\ast: \textbf{g}\bigcdot \textbf{g}\to \textbf{g}$ called the \emph{Lie bracket}, which satisfies antisymmetry and the Jacobi identity. Explicitly, we have linear maps
\[   
\partial^\ast_{S,T}:  \textbf{g}[S]\otimes \textbf{g}[T]\to \textbf{g}[I] 
,\qquad
[\mathtt{a}, \mathtt{b}]_{S,T}:= \partial^\ast_{S,T}( \mathtt{a}\otimes \mathtt{b} )
\]
where $S\sqcup T=I$, $I\in \sfS$, which commute with the action of $\sfS$, and satisfy antisymmetry 
\[
[\mathtt{a}, \mathtt{b}]_{S,T} + [\mathtt{b}, \mathtt{a}]_{T,S}=0
\]
and the Jacobi identity
\[
[[\mathtt{a}, \mathtt{b}]_{S,T} , \mathtt{c}]_{S\sqcup T, U}
+
[[\mathtt{c}, \mathtt{a}]_{U,S} , \mathtt{b}]_{U\sqcup S, T}
+
[[\mathtt{b}, \mathtt{c}]_{T,U} , \mathtt{a}]_{T\sqcup U, S}
=0
.\]

\begin{ex}\label{ex:E}
The exponential species $\textbf{E}$ is an algebra with
\[
\mu_{S,T} (\tH_S\otimes \tH_T ):= \tH_I 
\qquad \text{and} \qquad
\mathtt{1}_\textbf{E}:= \tH_\emptyset
.\]
Let $\textbf{E}^\ast$ also denote the exponential species, with $1_\Bbbk \in \textbf{E}^\ast[I]$ now denoted by $\tH_{(I)}$. Then $\textbf{E}^\ast$ is a coalgebra with
\[
\Delta_{S,T}(\tH_{(I)}):= \tH_{(S)}\otimes \tH_{(T)} 
\qquad \text{and} \qquad
\epsilon_{\emptyset}(\tH_\emptyset):= 1_\Bbbk
.\]
Equipped with all four morphisms, $\textbf{E}$ is a connected bialgebra. The antipode is given by 
\[
\text{s}_I(\tH_I)=(-1)^n\tH_I
\]
where $n=|I|$ as usual.
\end{ex}

\begin{ex}
Let $\text{L}[I]$ denote the set of linear orderings $\ell$ of $I$. We let $\emptyset\in \text{L}[\emptyset]$ denote the unique linear ordering of the empty set. For $S\subseteq I$, let $\ell|_S$ denote the restriction of $\ell$ to $S$. For $S\sqcup T=I$, $\ell_1\in \text{L}[S]$ and $\ell_2\in \text{L}[T]$, let $\ell_1 \ell_2\in \text{L}[I]$ denote the concatenation of $\ell_1$ with $\ell_2$. We have the species $\textbf{L}$ given by
\[
\textbf{L}[I]:= \big\{\text{formal $\Bbbk$-linear combinations of linear orderings of $I$}\big\}
.\]
We let $\tH_\ell$ denote the basis element of $\textbf{L}[I]$ corresponding to $\ell\in \text{L}[I]$. Then $\textbf{L}$ is a connected bialgebra, with multiplication and comultiplication given by
\[
\mu_{S,T} (\tH_{\ell_1} \otimes \tH_{\ell_2} ):= \tH_{\ell_1 \ell_2} 
\qquad \text{and} \qquad
\Delta_{S,T} (\tH_{\ell}):=   \tH_{\ell|_S} \otimes \tH_{\ell|_T} 
\]
and unit and counit given by
\[
\mathtt{1}_{\textbf{L}}:= \tH_{\emptyset} 
\qquad \text{and} \qquad
\epsilon_{\emptyset}(\tH_{\emptyset} ) :=1_\Bbbk
.\]
The antipode is given by
\[
\text{s}_I(\tH_\ell)=(-1)^n\tH_{\bar{\ell}}
\]
where $\bar{\ell}$ is the linear ordering of $I$ obtained by reversing the linear ordering $\ell$.
\end{ex}

%Moreover, a bialgebra structure on a connected species is necessarily a Hopf algebra, and there are explicit formulas for the antipode.

%%%%%%%%%%%%%%%%%%%%%%%%%%%%%%%%%%%%%%%%%%%%%%%%%%%%%%%%%%%%%%%%%%%%%%%%
\subsection{Poincar\'e-Birkhoff-Witt and Cartier-Milnor-Moore} \label{sec:Poincar\'e-Birkhoff-Witt and Cartier-Milnor-Moore}
%%%%%%%%%%%%%%%%%%%%%%%%%%%%%%%%%%%%%%%%%%%%%%%%%%%%%%%%%%%%%%%%%%%%%%%%

We recall the Poincar\'e-Birkhoff-Witt (PBW) and Cartier-Milnor-Moore (CMM) theorems for Hopf algebras in species. Every (perhaps nonunital) associative algebra in species $\textbf{a}$ gives rise to a Lie algebra in species via the \emph{commutator bracket} 
\[
[-,-] : \textbf{a}\bigcdot \textbf{a} \to \textbf{a}, \qquad 
[\mathtt{a}, \mathtt{b}]_{S,T}
:= 
\mu_{S,T}(\mathtt{a}\otimes \mathtt{b})-\mu_{T,S}(\mathtt{b}\otimes \mathtt{a}) 
.\] 
We sometimes abbreviate $[\mathtt{a}, \mathtt{b}]:=[\mathtt{a}, \mathtt{b}]_{S,T}$. For $\textbf{h}$ a Hopf algebra in species, the associated species of \emph{primitive elements} $\mathcal{P}(\textbf{h})$ is given by
\begin{equation}\label{eq:prim}
\mathcal{P}(\textbf{h})[I]:
=
\big\{ 
\mathtt{a}\in \textbf{h}[I] : 
\Delta_I(\mathtt{a})= \mathtt{a}\otimes \mathtt{1}_\textbf{h} +\mathtt{1}_\textbf{h}\otimes \mathtt{a}
\big\}
.
\end{equation}
The restriction of the commutator bracket of $\textbf{h}$ gives $\mathcal{P}(\textbf{h})$ the structure of a Lie algebra. Notice that if $\textbf{h}$ is connected, then $\mathcal{P}(\textbf{h})$ is a positive Lie algebra $\mathcal{P}(\textbf{h})[\emptyset]=0$, and 
\[
\mathcal{P}(\textbf{h})[I]
=\bigcap_{\substack{S\sqcup T=I\\[2pt] S,T\neq \emptyset}}  \text{ker} \big  (    \Delta_{S,T} :\textbf{h}[I]\to \textbf{h}[S]\otimes \textbf{h}[T]    \big)      
\]
for $I$ nonempty. %In the connected case, the antipode acts by negation on $\mathcal{P}(\textbf{h})$ \cite[Proposition 12]{aguiar2013hopf}. 

Let $\textbf{g}$ be a positive Lie algebra in species (so $\textbf{g}[\emptyset]=0$). The \emph{universal enveloping algebra} $\mathcal{U}(\textbf{g})$ of $\textbf{g}$ is the connected cocommutative Hopf algebra which is the quotient of the plethystic product $\textbf{L}\boldsymbol{\circ} \textbf{g}$ by the ideal generated by elements of the form
\[
\mathtt{a}\otimes\mathtt{b}-\mathtt{b}\otimes\mathtt{a}
-[\mathtt{a},\mathtt{b}]
\qquad \quad
\text{for} \quad  \mathtt{a}\in \textbf{g}[S]\quad \text{and} \quad  \mathtt{b}\in \textbf{g}[T] 
.\]   
See \cite[Section 8.2]{aguiar2013hopf} for more details. The plethystic product of the coalgebra homomorphism
\[\textbf{E}^\ast \to \textbf{L}, 
\qquad \tH_{(I)}\mapsto \dfrac{1}{n!} \sum_{\ell\in \text{L}[I]} \tH_\ell\] 
with $\text{id}_\textbf{g}$ gives a coalgebra homomorphism $\textbf{E}^\ast\boldsymbol{\circ} \textbf{g}\to \textbf{L}\boldsymbol{\circ} \textbf{g}$. The following is Aguiar-Mahajan's \cite[Theorem 119]{aguiar2013hopf} elaboration on Joyal's PBW theorem for species \cite[Section 4.2, Theorem 2]{joyal1986foncteurs}. 

\begin{thm}[PBW]
Let $\textbf{g}$ be a positive Lie algebra in species. The composition of the coalgebra homomorphism $\textbf{E}^\ast \boldsymbol{\circ} \textbf{g}\to \textbf{L}\boldsymbol{\circ}\textbf{g}$ with the quotient map $\textbf{L}\boldsymbol{\circ} \textbf{g} \twoheadrightarrow \mathcal{U}(\textbf{g})$ is an isomorphism of coalgebras
\[
\textbf{E}^\ast \boldsymbol{\circ} \textbf{g} \xrightarrow{\sim}\mathcal{U}(\textbf{g})
.\]
\end{thm}

This isomorphism commutes with the canonical images of $\textbf{g}$. In particular, $\textbf{g}$ is a Lie subalgebra of its universal enveloping algebra $\mathcal{U}(\textbf{g})$. The following is due to Stover \cite[Proposition 7.10 and Theorem 8.4]{stover1993equivalence}.

\begin{thm}[CMM] \label{CMM}
The constructions universal enveloping algebra $\mathcal{U}(-)$ and primitive elements $\mathcal{P}(-)$ form an adjoint equivalence between positive Lie algebras in species and connected cocommutative Hopf algebras in species,
\begin{center}
\begin{tikzcd}[column sep=huge,row sep=large] 
^{\text{co}}\textsf{Hopf}([\sfS^{\op},\textsf{Vec}]^\circ)   
\arrow[r, shift right=1.142ex,  "\bot" , "\mathcal{P}"']   
&   
\textsf{Lie}([\sfS^{\op},\textsf{Vec}]_+)  
\arrow[l, shift right=1.142ex, "\mathcal{U}"'] .
\end{tikzcd}
\end{center}
\end{thm}

\begin{ex}
We have
\[
\mathcal{P}(\textbf{E})= \textbf{X}
\qquad \text{and} \qquad
\mathcal{U}(\textbf{X})= \textbf{E}
\]
and also
\[
\mathcal{P}(\textbf{L})= \Lie
\qquad \text{and} \qquad
\mathcal{U}(\Lie)= \textbf{L}
\]
where $\Lie$ is the underlying species of the (positive) Lie operad \cite[Example B.5]{aguiar2010monoidal}. 
\end{ex}

%%%%%%%%%%%%%%%%%%%%%%%%%%%%%%%%%%%%%%%%%%%%%%%%%%%%%%%%%%%%%%%%%%%%%%%%
\section{Internal and External Homs}
%%%%%%%%%%%%%%%%%%%%%%%%%%%%%%%%%%%%%%%%%%%%%%%%%%%%%%%%%%%%%%%%%%%%%%%% 

We now discus the $\sfVec$-enriched external hom for species $\Hom(-,-)$, and the internal homs for the Hadamard product $\cH(-,-)$ and the Cauchy product $\cH^{\bigcdot}(-,-)$.
  
%%%%%%%%%%%%%%%%%%%%%%%%%%%%%%%%%%%%%%%%%%%%%%%%%%%%%%%%%%%%%%%%%%%%%%%%
\subsection{External Hom}\label{sec:externalhom}
%%%%%%%%%%%%%%%%%%%%%%%%%%%%%%%%%%%%%%%%%%%%%%%%%%%%%%%%%%%%%%%%%%%%%%%% 

For species $\textbf{p}$ and $\textbf{q}$, we have the vector space $\Hom(\textbf{p}, \textbf{q})$ given by
\[     
\Hom(\textbf{p}, \textbf{q}):
=
\big\{\text{morphisms of species $\xi:\textbf{p}\to \textbf{q}$}\big\}   
.\]
The vector space structure is given by
\[             
(\xi + \nu)_I :=  {\xi}_I + {\nu}_I 
\qquad 
\text{and} 
\qquad
(c\, \xi)_I:=   c\,  \xi_I \quad \text{for} \quad  c\in \Bbbk
.\] 
For morphisms of species $\zeta:\textbf{r}\to \textbf{p}$ and $\eta:\textbf{q}\to \textbf{s}$, we have the linear map $\Hom(\zeta,\eta)$ given by 
\[          
\Hom(\zeta,\eta)  : \Hom( \textbf{p}, \textbf{q}) \to  \Hom(\textbf{r}, \textbf{s}) , 
\qquad  
\xi\mapsto  \eta \circ \xi \circ \zeta
.\]
This defines the external hom for species,
\[       
\Hom(-,-):  [\sfS^{\op},\textsf{Vec}]^{\op}\times [\sfS^{\op},\textsf{Vec}]  
\to  
\textsf{Vec}  
.\]
Let us recall the construction of \cite[Section 2.7]{aguiar2013hopf}. If $\textbf{c}$ is a coalgebra and $\textbf{a}$ is an algebra, then $\Hom(\textbf{c},\textbf{a})$ is a $\Bbbk$-algebra (with multiplication denoted $\ast$) as follows. For $\xi,\nu\in \Hom(\textbf{c},\textbf{a})$, put 
\begin{equation}\label{eq:convolprod}
\xi\ast \nu:
=
\mu  \circ  (  \xi\bigcdot \nu )  \circ    \Delta 
\footnote{
in terms of components we have  $(\xi\ast \nu)_I:
=\sum_{S\sqcup T=I}
\mu_{S,T}  \circ  (  \xi_S\otimes \nu_T )  \circ    \Delta_{S,T}$   
} 
\qquad \text{and} \qquad
1_{\Hom(\textbf{c},\textbf{a})}:= \iota \circ \epsilon
. 
\end{equation}  
Then associativity and unitality follow from (co)associativity and (co)unitality of $\textbf{a}$ (and $\textbf{c}$). If $\zeta:\textbf{c}_2\to \textbf{c}_1$ is a homomorphism of coalgebras and $\eta:\textbf{a}_1\to \textbf{a}_2$ is a homomorphism of algebras, then it is straightforward to check that $\Hom(\zeta,\eta)$ is a homomorphism of $\Bbbk$-algebras.

Given any species $\textbf{p}$, the vector space $\text{End}(\textbf{p}):=\Hom(\textbf{p},\textbf{p})$ is a $\Bbbk$-algebra under composition of morphisms. 

%%%%%%%%%%%%%%%%%%%%%%%%%%%%%%%%%%%%%%%%%%%%%%%%%%%%%%%%%%%%%%%%%%%%%%%%
\subsection{Internal Hom for the Hadamard Product}\label{sec:homforhad}
%%%%%%%%%%%%%%%%%%%%%%%%%%%%%%%%%%%%%%%%%%%%%%%%%%%%%%%%%%%%%%%%%%%%%%%% 

For species $\textbf{p}$ and $\textbf{q}$, we have the species $\cH(\textbf{p}, \textbf{q})$ given by
\[
\cH(\textbf{p},\textbf{q})[I]
:=
\big\{ \text{linear maps } f:\textbf{p}[I]\to \textbf{q}[I]  \big  \} 
\qquad \text{and} \qquad   
\cH(\textbf{p},\textbf{q})[\sigma](f)
:= 
\textbf{q}[\sigma] \circ  f \circ \textbf{p}[\sigma]^{-1}    
 .   \]
For morphisms of species $\zeta:\textbf{r}\to \textbf{p}$ and $\eta:\textbf{q}\to \textbf{s}$, we have the morphism of species $\cH(\zeta,\eta)$ given by
\[          
\cH(\zeta,\eta)_I  : \cH( \textbf{p}, \textbf{q} )[I] \to  \cH(\textbf{r}, \textbf{s})[I] , 
\qquad  
f\mapsto  \eta_I \circ f \circ \zeta_I
.\] 
Then $\cH(-,-)$ is the internal hom for the Hadamard product of species. We have the isomorphism of vector spaces called \emph{currying} 
\[    
\Hom(  \textbf{p}\otimes \textbf{q} ,  \textbf{r}  )\xrightarrow{\sim}  \Hom\! \big(   \textbf{p} , \cH(  \textbf{q},  \textbf{r} )  \big )
,\qquad    
\eta\mapsto \wh{\eta}
\]
where
\[    
\wh{\eta}_I( \mathtt{v} ): \textbf{q}[I]\to \textbf{r}[I]
,\qquad  
\mathtt{w}\mapsto \eta_I( \mathtt{v}\otimes \mathtt{w})
.\]
Since the exponential species $\textbf{E}$ is the unit for the Hadamard product, by putting $\textbf{p}=\textbf{E}$, we recover the external hom from the internal hom as follows,
\[       
\Hom( - ,  -  )   =   \Hom( \textbf{E} ,  -  ) \circ      \cH( - ,  -  )      
.\]
Let us recall a construction of \cite[Section 3.2]{aguiar2013hopf}. If $\textbf{c}$ is a coalgebra and $\textbf{a}$ is an algebra, then $\cH(\textbf{c},\textbf{a})$ is an algebra in species as follows,
\begin{equation} \label{eq:curlyHalg}
\mu_{S,T}(f\otimes g):
=\mu_{S,T}  \circ (f\otimes g)  \circ    \Delta_{S,T}  
\qquad
\text{and}
\qquad 
\mathtt{1}_{\cH(\textbf{c},\textbf{a})}  : 
=\iota_{\emptyset}  \circ \epsilon_{\emptyset}
\end{equation}  
where associativity and unitality follow from (co)associativity and (co)unitality of $\textbf{a}$ (and $\textbf{c}$). If $\zeta:\textbf{c}_2\to \textbf{c}_1$ is a homomorphism of coalgebras and $\eta:\textbf{a}_1\to \textbf{a}_2$ is a homomorphism of algebras, then it is straightforward to check that $\cH(\zeta,\eta)$ is a homomorphism of algebras. We see in \textcolor{blue}{(\refeqq{eq:Smulttran})} that $\Hom( \textbf{E} ,  -  )$ is a lax monoidal functor, which recovers \textcolor{blue}{(\refeqq{eq:convolprod})} as the image of \textcolor{blue}{(\refeqq{eq:curlyHalg})}. This is discussed in \cite[Section 12.11]{aguiar2013hopf}.

Given any species $\textbf{p}$, the species $\CMcal{E}(\textbf{p}):=\cH(\textbf{p},\textbf{p})$ is a presheaf of $\Bbbk$-algebras with multiplication the composition of linear maps, equivalently a monoid internal to species constructed with respect to the Hadamard product. 

%%%%%%%%%%%%%%%%%%%%%%%%%%%%%%%%%%%%%%%%%%%%%%%%%%%%%%%%%%%%%%%%%%%%%%%%
\subsection{Internal Hom for the Cauchy Product}\label{sec:inthomcauchy}
%%%%%%%%%%%%%%%%%%%%%%%%%%%%%%%%%%%%%%%%%%%%%%%%%%%%%%%%%%%%%%%%%%%%%%%%

We now distinguish between two copies of $\sfS$. The first $\sfS$ has `color' the formal symbol $\formg$ (physically, the coupling constant). We call the corresponding species $\formg$\emph{-colored}. The second $\sfS$ has color the formal symbol $\formj$ (physically, the source field). We call the corresponding species $\formj$\emph{-colored}. We identify the $\sfS$ which has appeared hitherto with the $\formj$-colored $\sfS$. 

We denote sets of the $\formg$-colored $\sfS$ by $Y$. We let $r:=|Y|$ throughout, and we pick a section of the decategorification functor  
\[
\sfS\to \bN
,\qquad Y\mapsto r
\] 
by letting $[r]$ denote the finite set of symbols 
\[
[r]:=\{\ast_1,\dots, \ast_r\},\qquad   r\in \bN 
.\]
We sometimes abbreviate $\ast:= \ast_1$ and $\ast I:=\{\ast\}\sqcup I$. If $\eta$ is a morphism of $\formg$-colored species, we abbreviate $\eta_r:= \eta_{[r]}$.

Let a $\{\! \formg,\! \formj\! \}$\emph{-colored} \emph{species} be a presheaf of vector spaces on the product $\sfS\times \sfS$ of our two copies of $\sfS$. We let the $\formj$-colored $\sfS$ be a module category over the $\formg$-colored $\sfS$ where the action is disjoint union,
\[
\sfS\times\sfS\to\sfS
,\qquad
(Y,I)\mapsto Y\sqcup I
,\quad
(\sigma_1,\sigma_2)\mapsto \sigma_1\sqcup \sigma_2
.\footnote{\ similarly, the category of $\formg$-colored species acts on the category of $\formj$-colored species via the Cauchy product, and as an instance of the microcosm principle, our theory of modules internal to species will be with respect to this module category}\] 
By precomposing $\formj$-colored species with the action $\sfS \times \sfS \to \sfS$, we obtain a functor $(-)_{(2)}$ on $\formj$-colored species into $\{\! \formg,\! \formj\! \}$-colored species (see also \cite[Section 14.1.4]{aguiar2010monoidal}),
\[
(-)_{(2)}:[\sfS^{\op},\sfVec] \to \big [(\sfS\times \sfS)^{\op},\sfVec\big ]
,\qquad \textbf{p} \mapsto  \textbf{p} \circ  \big ((-)\sqcup (-)\big) 
.\]
%This functor is bistrong monoidal with respect to the Day convolution \cite[Proposition 14.4]{aguiar2010monoidal}. 
Explicitly, for a species $\textbf{p}$ we have
\[
\textbf{p}_{(2)}:(\sfS\times \sfS)^{\op}\to \sfVec
,\qquad
(Y,I)\mapsto  \textbf{p}^{[Y]}[I]:=\textbf{p}[Y\sqcup I]
\quad \text{and} \quad
(\sigma_1,\sigma_2)\mapsto  \textbf{p}^{[\sigma_1]}[\sigma_2]:=\textbf{p}[\sigma_1\sqcup \sigma_2]
.\] 
Fixing $Y\in \sfS$, by precomposing $\{\! \formg,\! \formj\! \}$-colored species with the functor 
\[
\sfS\xrightarrow{\sim}  \text{id}_Y\times  \sfS \hookrightarrow   \sfS \times \sfS
\] 
we obtain a functor on $\{\! \formg,\! \formj\! \}$-colored species back into $\formj$-colored species. We denote the precomposition of this functor with $(-)_{(2)}$ by $(-)^{[Y]}$. Then $(-)^{[Y]}$ is an endofunctor on $\formj$-colored species. Explicitly, for a species $\textbf{p}$ we have
\[
\textbf{p}^{[Y]}:\sfS^{\op}\to \sfVec
, \qquad 
I\mapsto \textbf{p}^{[Y]}[I] \quad \text{and} \quad \sigma \mapsto \textbf{p}^{[\text{id}_Y]}[\sigma]
\]
and for a morphism of species $\eta:\textbf{p}\to \textbf{q}$ we have
\[
\eta^{[Y]}: \textbf{p}^{[Y]}\to \textbf{q}^{[Y]}
,\qquad
\eta^{[Y]}_I = \eta_{Y\sqcup I}
.\]
For $\formj$-colored species $\textbf{p}$ and $\textbf{q}$, we have the $\formg$-colored species $\cH^{\bigcdot}(\textbf{p},\textbf{q})$ given by
\[
\cH^{\bigcdot}(\textbf{p},\textbf{q})[Y]:=
\Hom(\textbf{p},\textbf{q}^{[Y]}) 
\qquad \text{and} \qquad   
\big(\cH^{\bigcdot}(\textbf{p},\textbf{q})[\sigma](u)\big)_I:= \textbf{q}^{[\sigma]}[\text{id}_I]  \circ  u_I 
.\]
We may view $\cH^{\bigcdot}(\textbf{p},\textbf{q})$ as the species of `raising morphisms' from $\textbf{p}$ to $\textbf{q}$. For morphisms of species $\zeta:\textbf{p}_2\to \textbf{p}_1$ and $\eta:\textbf{q}_1\to \textbf{q}_2$, we have the morphism of species $\cH^{\bigcdot}(\zeta,\eta)$ given by
\[          
\cH^{\bigcdot}(\zeta,\eta)_Y: 
\cH^{\bigcdot}( \textbf{p}_1, \textbf{q}_1)[Y] 
\to  
\cH^{\bigcdot}(\textbf{p}_2, \textbf{q}_2)[Y] 
,\qquad  
u\mapsto  \eta^{[Y]} \circ u \circ \zeta
.\] 
Then $\cH^{\bigcdot}(-,-)$ is the internal hom for the Cauchy product of species \cite[Proposition 8.51]{aguiar2010monoidal}. Thus, given a $\formg$-colored species $\textbf{r}$ and $\formj$-colored species $\textbf{p}$ and $\textbf{q}$, we have the isomorphism of vector spaces called \emph{currying}
\[
\Hom( \textbf{r}  \bigcdot \textbf{p}, \textbf{q})
\xrightarrow{\sim}
\Hom\! \big (\textbf{r},   \cH^{\bigcdot}(\textbf{p},\textbf{q}) \big  )
,\qquad
\rho\mapsto \widehat{\rho}
\]
where
\[
 \widehat{\rho}_Y(\mathtt{w}): \textbf{p}\to \textbf{q}^{[Y]}
,\qquad    
\mathtt{v}\mapsto \rho_{Y,I}( \mathtt{w}\otimes \mathtt{v} )
.\] 

%%%%%%%%%%%%%%%%%%%%%%%%%%%%%%%%%%%%%%%%%%%%%%%%%%%%%%%%%%%%%%%%%%%%%%%%
\subsection{The Endomorphism Algebra of Raising Operators} \label{sec:inthomcauchy2}
%%%%%%%%%%%%%%%%%%%%%%%%%%%%%%%%%%%%%%%%%%%%%%%%%%%%%%%%%%%%%%%%%%%%%%%%

Given a species $\textbf{p}$, we have the species $\CMcal{E}^{\bigcdot}(\textbf{p}):=\cH^{\bigcdot}(\textbf{p},\textbf{p})$ of \emph{raising operators} $u\in \CMcal{E}^{\bigcdot}(\textbf{p})[Y]$,
\[
u:\textbf{p}\to \textbf{p}^{[Y]}
,\qquad
\mathtt{v} \mapsto u(\mathtt{v}).
\footnote{\ for raising operators, we always abbreviate $u(\mathtt{v}):=u_I(\mathtt{v})$}
\]  
Following \cite[Section 8.12.1]{aguiar2010monoidal}, in the case of the singleton $Y=\{\ast\}$, we let $(-)':=(-)^{[\{\ast\}]}$ and we call an operator of the form $u:\textbf{p}\to \textbf{p}'$ an \emph{up operator}. Notice that $\CMcal{E}^{\bigcdot}(\textbf{p})$ is naturally an algebra in species, with multiplication of raising operators given by composition,   
\begin{equation}\label{eq:multcomp}
\mu_{Y_1,Y_2}(u\otimes v) 
:= 
u^{[Y_2]}  \circ   v.
\end{equation}
Associativity follows from associativity of disjoint union. The unit $\mathtt{1}_{\CMcal{E}^{\bigcdot}(\textbf{p})}\in \CMcal{E}^{\bigcdot}(\textbf{p})[\emptyset]$ is the canonical isomorphism of species
\[
\mathtt{1}_{\CMcal{E}^{\bigcdot}(\textbf{p})}: \textbf{p}\xrightarrow{\sim} \textbf{p}^{[\emptyset]}
.\]

Let $\textbf{a}$ be an algebra in species. Let a $Y$-\emph{derivation} $D$ of $\textbf{a}$ be a raising operator $D:\textbf{a} \to \textbf{a}^{[Y]}$ such that
\[
D\big ( \mu_{S,T}( \mathtt{a} \otimes \mathtt{b} ) \big) = \mu_{Y\sqcup S,T}\big ( D(\mathtt{a}) \otimes \mathtt{b} \big )  +     \mu_{S,Y\sqcup T}\big ( \mathtt{a} \otimes D(\mathtt{b}) \big )
.\]
Following \cite[Section 8.12.4]{aguiar2010monoidal}, in the case of the singleton $Y=\{\ast\}$, we just say \emph{up derivation}. Given $\mathtt{a}\in \textbf{a}[Y]$, it is straightforward to check that
\[
\textbf{a}\to \textbf{a}^{[Y]}
,\qquad
\mathtt{b}\mapsto [\mathtt{a}, \mathtt{b} ]= \mu_{Y,I} ( \mathtt{a} \otimes \mathtt{b} )-\mu_{I,Y} ( \mathtt{b} \otimes \mathtt{a} )
\]
is a $Y$-derivation of $\textbf{a}$. Define the species $\text{Der}(\textbf{a})$ by letting $\text{Der}(\textbf{a})[Y]\subseteq \CMcal{E}^{\bigcdot}(\textbf{a})[Y]$ be the subspace of $Y$-derivations of $\textbf{a}$. 

\begin{prop}
Let $\textbf{a}$ be an algebra in species.  The species of derivations $\text{Der}(\textbf{a})$ is a Lie subalgebra of $\CMcal{E}^{\bigcdot}(\textbf{a})$. 
\end{prop}
\begin{proof}
We need to show that $\text{Der}(\textbf{a})$ is closed under the commutator bracket of $\CMcal{E}^{\bigcdot}(\textbf{a})$. Let $Y=Y_1\sqcup Y_2$, and let $D_1$ be a $Y_1$-derivation and $D_2$ be a $Y_2$-derivation of $\textbf{a}$. We have
\begin{align*}
[D_1 ,D_2 ] \big ( \mu_{S,T}( \mathtt{a} \otimes \mathtt{b} ) \big) \ \ 
&=
\ \ \ {D_1}^{[Y_2]}  \circ   {D_2}  \big (   \mu_{S,T}( \mathtt{a} \otimes \mathtt{b}) \big )    -{D_2}^{[Y_1]}  \circ   {D_1}(   \mu_{S,T}\big ( \mathtt{a} \otimes \mathtt{b}   ) \big)  \\[6pt]
&=
\ \ \ {D_1}^{[Y_2]} \Big (       \mu_{Y_2\sqcup S,T}\big (  {D_2}(\mathtt{a}) \otimes \mathtt{b} \big )  +    \mu_{S,Y_2\sqcup T}\big ( \mathtt{a} \otimes  {D_2}(\mathtt{b})  \big )    \Big)          \\[6pt]
&
\ \ \ -{D_2}^{[Y_1]} \Big (       \mu_{Y_1 \sqcup S,T}\big (  {D_1}(\mathtt{a}) \otimes \mathtt{b} \big )  +    \mu_{S,Y_1 \sqcup T}\big ( \mathtt{a} \otimes  {D_1}(\mathtt{b})  \big )    \Big)          \\[6pt]
&=
 \ \ \ \mu_{Y\sqcup S,T}\big (  [D_1,D_2] (\mathtt{a}) \otimes \mathtt{b} \big )+  \mu_{S,Y\sqcup T}\big (   \mathtt{a} \otimes [D_1,D_2]( \mathtt{b}) \big ) .
\end{align*}
Therefore the bracket $[D_1 ,D_2 ]$ is a $Y$-derivation of $\textbf{a}$.
\end{proof}

Notice that if $\textbf{a}$ is connected, then $\text{Der}(\textbf{a})$ is a positive Lie algebra.

%%%%%%%%%%%%%%%%%%%%%%%%%%%%%%%%%%%%%%%%%%%%%%%%%%%%%%%%%%%%%%%%%%%%%%%%
\subsection{$(-)^{\textbf{r}}$-Coalgebras} \label{sec: coalg}
%%%%%%%%%%%%%%%%%%%%%%%%%%%%%%%%%%%%%%%%%%%%%%%%%%%%%%%%%%%%%%%%%%%%%%%%

%We now give an alternative description of $\cH^{\bigcdot}(-,-)$ where the arguments have opposite colors. 

Now fixing $I\in \sfS$ instead, by precomposing $\{\! \formg,\! \formj\! \}$-colored species with the functor 
\[
\sfS\xrightarrow{\sim}\sfS\times \text{id}_I \hookrightarrow  \sfS \times \sfS
\] 
we obtain a functor on $\{\! \formg,\! \formj\! \}$-colored species into $\formg$-colored species. We adopt the following notation for the precomposition of this functor with $(-)_{(2)}$; we denote the image of $\textbf{p}$   by $\textbf{p}^{[-]}[I]$, explicitly $\textbf{p}^{[-]}[I]$ is the $\formg$-colored species
\[
\textbf{p}^{[-]}[I]:\sfS^{\op}\to \sfVec
, \qquad 
Y\mapsto \textbf{p}^{[Y]}[I] \quad \text{and} \quad \sigma \mapsto \textbf{p}^{[\sigma]}[\text{id}_I]
,\]
and we denote the image of a morphism $\eta:\textbf{p}\to \textbf{q}$ by $\eta_I^{[-]}$, explicitly
\[
\eta_I^{[-]}: \textbf{p}^{[-]}[I]\to \textbf{q}^{[-]}[I]
,\qquad
(\eta_I^{[-]})_Y: =\eta^{[Y]}_{I}
.\]
For $Y=[r]$, we abbreviate $\eta^r_{I}:=\eta_I^{[Y]}$.

For a $\formg$-colored species $\textbf{r}$ and a $\formj$-colored species $\textbf{p}$, we have the $\formj$-colored species $\textbf{p}^\textbf{r}$ given by 
\[
\textbf{p}^\textbf{r}[I]:=
\Hom(\textbf{r},\textbf{p}^{[-]} [I] ) 
\qquad \text{and} \qquad   
\textbf{p}^\textbf{r}[\sigma](\chi):=
\textbf{p}^{[\text{id}_Y]}[\sigma]  \circ  \chi_Y 
.\] 
We shall tend to fix $\textbf{r}$, and view $\textbf{p}^\textbf{r}$ as the species of `$\textbf{r}$-elements' of $\textbf{p}$. For morphisms of species $\zeta:\textbf{r}_2\to \textbf{r}_1$ and $\eta:\textbf{p}\to \textbf{q}$, we have the morphism of species $\eta^\zeta$ given by
\[          
\eta^\zeta_I: 
\textbf{p}^{\textbf{r}_1}[I] 
\to  
\textbf{q}^{\textbf{r}_2}[I] 
,\qquad  
\chi\mapsto  \eta_I^{[-]} \circ \chi \circ \zeta
.\] 
We have the isomorphism of vector spaces
\[
\Hom( \textbf{r}  \bigcdot \textbf{p}, \textbf{q})
\xrightarrow{\sim}
\Hom\! \big (\textbf{p},  \textbf{q}^{\textbf{r}} \big  )
,\qquad
\rho\mapsto \widecheck{\rho} 
\]
where $\widecheck{\rho}:\textbf{p}\to \textbf{q}^\textbf{r}$ is given by
\[
\widecheck{\rho}_I(\mathtt{v}): \textbf{r}\to \textbf{p}^{[-]}[I]
,\qquad  
\mathtt{w}\mapsto \rho_{Y,I}( \mathtt{w}\otimes \mathtt{v} ) 
.\]

If $\textbf{c}$ is a coalgebra and $\textbf{a}$ is an algebra, then $\textbf{a}^\textbf{c}$ is an algebra as follows. Given $\textbf{c}$-elements
\[\chi: \textbf{c}\to \textbf{a}^{[-]}[S] \qquad \text{and} \qquad \omega: \textbf{c}\to \textbf{a}^{[-]}[T]\] 
define the $\textbf{c}$-element $\mu_{S,T}(\chi\otimes  \omega):\textbf{c}\to \textbf{a}^{[-]}[I]$ by
\begin{equation}\label{eq:multofa^c}
\big (\mu_{S,T}(\chi\otimes  \omega)\big )_Y
:=    
\sum_{Y_1 \sqcup Y_2 =Y} \mu_{ Y_1\sqcup S, Y_2\sqcup T } \circ(\chi_{Y_1}\otimes \omega_{Y_2})\circ\Delta_{Y_1,Y_2}
.
\end{equation}
The unit $\mathtt{1}_{\textbf{a}^\textbf{c}}\in \textbf{a}^\textbf{c}[\emptyset]=\Hom(\textbf{c},\textbf{a}^{[-]}[\emptyset])$ is given by
\begin{equation}\label{eq:unitofa^c}
\mathtt{1}_{\textbf{a}^\textbf{c}}:=\iota^{[-]}_\emptyset  \circ \epsilon
.
\end{equation}
Associativity and unitality follow from (co)associativity and (co)unitality of $\textbf{a}$ (and $\textbf{c}$). If $\zeta:\textbf{c}_2\to \textbf{c}_1$ is a homomorphism of coalgebras and $\eta:\textbf{a}_1\to \textbf{a}_2$ is a homomorphism of algebras, then it is straightforward to check that $\eta^{\zeta}$ is a homomorphism of algebras.   

Given a $\formg$-colored species $\textbf{r}$, we have the endofunctor $(-)^{\textbf{r}}$ on $\formj$-colored species, given by
\[ 
(-)^{\textbf{r}}:   
[\sfS^{\op},\textsf{Vec}]\to[\sfS^{\op},\textsf{Vec}]
,\qquad
\textbf{p}\mapsto \textbf{p}^{\textbf{r}} \quad \text{and} \quad \eta\mapsto \eta^{\textbf{r}}:=\eta^{\text{id}_\textbf{r}}
.\]
We have
\[
(\textbf{p}\bigcdot \textbf{q})^{\textbf{r}}[I] 
=
\Hom\! \big(\textbf{r},  (\textbf{p}\bigcdot \textbf{q})^{[-]} [I] \big )  
=
\bigoplus_{S\sqcup T=I}\Hom\! \Big (
\textbf{r},\bigoplus_{Y_1\sqcup Y_2=(-)}\textbf{p}^{[Y_1]}[S]\otimes \textbf{q}^{[Y_2]}[T]  
\Big )  
.\]
For $\textbf{c}$ a coalgebra, the functor $(-)^{\textbf{c}}$ is lax monoidal with respect to the Cauchy product, with multiplication transformation
%\begin{equation} \label{eq:laxx}
\[
\textbf{p}^{\textbf{c}} \bigcdot \textbf{q}^{\textbf{c}} \to  (\textbf{p}  \bigcdot  \textbf{q})^{\textbf{c}}
,\qquad  
\chi \otimes \omega
\mapsto 
\sum_{Y_1 \sqcup Y_2 =(-)}  (\chi_{Y_1} \otimes \omega_{Y_2})
\circ \Delta_{Y_1, Y_2}
\]
%\end{equation}
and unit transformation
\[
\textbf{1}\to  \textbf{1}^{\textbf{c}}
,\qquad
1_\Bbbk \mapsto\iota_{\emptyset}  \circ \epsilon_\emptyset
.\]
If $\textbf{a}$ is an algebra, then the induced algebraic structure for $\textbf{a}^{\textbf{c}}$ recovers \textcolor{blue}{(\refeqq{eq:multofa^c})} and \textcolor{blue}{(\refeqq{eq:unitofa^c})}.

An $(-)^{\textbf{r}}$-\emph{coalgebra} is a morphism of species of the form $\textbf{p}\to \textbf{p}^\textbf{r}$. If $\textbf{a}$ is an algebra and $\textbf{c}$ is a coalgebra, then one can additionally require that a $(-)^\textbf{c}$-coalgebra $\textbf{a}\to \textbf{a}^\textbf{c}$ is a homomorphism of algebras. Particularly important for us is the endofunctor $(-)^\textbf{E}$. 

\begin{remark}
The endofunctor $(-)^{\textbf{E}}$ may alternatively be obtained by composing $(-)_{(2)}$ with currying to obtain species-valued species, and then composing with $\Hom(\textbf{E},-)$ to get back ordinary species. 
\end{remark}

%\[       
%\eta^{\zeta}_I  :\textbf{a}_1^{\textbf{c}_1}[I] \to  \textbf{a}_2^{\textbf{c}_1}[I] , \qquad  f\mapsto  \eta_I \circ f \circ \zeta_I 
%\]
%defines a homomorphism of algebras. 
%%%%%%%%%%%%%%%%%%%%%%%%%%%%%%%%%%%%%%%%%%%%%%%%%%%%%%%%%%%%%%%%%%%%%%%%
\section{Two Lax Monoidal Functors}
%%%%%%%%%%%%%%%%%%%%%%%%%%%%%%%%%%%%%%%%%%%%%%%%%%%%%%%%%%%%%%%%%%%%%%%%
We now describe in more detail the special cases the functor $\Hom(\textbf{E},-)$ and the endofunctor $(-)^{\textbf{E}}$.
%%%%%%%%%%%%%%%%%%%%%%%%%%%%%%%%%%%%%%%%%%%%%%%%%%%%%%%%%%%%%%%%%%%%%%%%
\subsection{Series of Species}\label{sec:Series of Species}
%%%%%%%%%%%%%%%%%%%%%%%%%%%%%%%%%%%%%%%%%%%%%%%%%%%%%%%%%%%%%%%%%%%%%%%%
We recall series of species, following \cite[Section 12]{aguiar2013hopf}. Let us denote the functor $\Hom( \textbf{E},-)$ by $\mathscr{S}(-)$, thus
\[   
\mathscr{S}(-):   [\sfS^{\op},\textsf{Vec}]  \to \textsf{Vec}, \qquad
 \mathscr{S}(\textbf{p}):=  \Hom( \textbf{E} ,  \textbf{p}  ) \quad \text{and} \quad   \mathscr{S}(\eta):=  \Hom( \text{id}_\textbf{E} ,  \eta  ) 
.\]
Elements of the vector space $\mathscr{S}(\textbf{p})$ are called \emph{series} of $\textbf{p}$, which we denote by
\[  
\mathtt{s}=\mathtt{s}_{(-)}:\textbf{E}\to \textbf{p}, \qquad  \tH_I \mapsto    \mathtt{s}_{I}
.\]
Explicitly, given a morphism of species $\eta:\textbf{p}\to \textbf{q}$, then $\mathscr{S}(\eta)$ is the linear map given by  
\[
\mathscr{S}(\eta): \mathscr{S}(\textbf{p}) \to \mathscr{S}(\textbf{q})
,\qquad
\mathtt{s}\mapsto \eta \circ \mathtt{s}
.\] 
A series of $\textbf{p}$ is equivalent to picking an $\sfS_n$-invariant vector of $\textbf{p}[n]$, for each $n\in \bN$. Therefore, we have the isomorphism of vector spaces
\[
\mathscr{S}(\textbf{p})\xrightarrow{\sim}\prod_{n=0}^\infty\big (\textbf{p}[n]\big)^{\sfS_n}, 
\qquad  
\mathtt{s}\mapsto   \sum^\infty_{n=0} \mathtt{s}_n 
.\footnote{\ if $V$ is a $\Bbbk[\sfS_n]$-module, then we let $V^{\sfS_n}\subseteq V$ denote the subspace of $\sfS_n$-invariant vectors, and recall that $\mathtt{s}_n$ denotes $\mathtt{s}_{[n]}$, where $[n]=\{1,\dots, n\}$}
\]
Indeed, as a special case of a general fact about presheaves, $\mathscr{S}(\textbf{p})$ is equivalently the \hbox{category-theoretic} limit of $\textbf{p}$ as a diagram of shape $\sfS^{\op}$ in $\textsf{Vec}$.

The functor $\mathscr{S}$ is lax monoidal with respect to the Cauchy product and tensor product, with multiplication transformation 
\begin{equation}\label{eq:Smulttran}
\mathscr{S}(\textbf{p})\otimes  \mathscr{S}(\textbf{q})\to\mathscr{S}(\textbf{p}\bigcdot \textbf{q} )
,\qquad     
\mathtt{s}_{(-)}  \otimes \mathtt{t}_{(-)}\mapsto   \sum_{S\sqcup T=(-)}  \mathtt{s}_{S}\otimes  \mathtt{t}_{T} 
\end{equation}
and unit transformation the canonical isomorphism
\[
\Bbbk \xrightarrow{\sim} \mathscr{S}(\textbf{1}) 
.\]
Thus, if $\textbf{a}$ is an algebra in species, then $\mathscr{S}(\textbf{a})$ is a $\Bbbk$-algebra with multiplication given by
\begin{equation}\label{eq:multinS(a)}
\mathscr{S}(\textbf{a})\otimes \mathscr{S}(\textbf{a}) \to \mathscr{S}(\textbf{a}), \qquad   {\mathtt{s}}_{(-)}\otimes {\mathtt{t}}_{(-)} \mapsto {\mathtt{s}}\ast{\mathtt{t}}_{(-)} 
:=   
\sum_{S\sqcup T=(-)} \mu_{S,T}({\mathtt{s}}_{S} \otimes {\mathtt{t}}_{T})
\end{equation}
and unit the series given by
\begin{equation}\label{eq:unitseries}
\mathtt{u}:\textbf{E}\to \textbf{a}, 
\qquad                
\mathtt{u}_I :=  
\begin{cases}
\mathtt{1}_\textbf{a} &\quad \text{if $I=\emptyset$}\\
0            &\quad \text{otherwise}.
\end{cases}
\end{equation}  

Let $\textbf{a}$ be an algebra in species and let $\textbf{c}$ be a coalgebra in species. Recall from \autoref{ex:E} that $\textbf{E}$ is a Hopf algebra. We have the subset
\[        
\mathscr{E}(\textbf{a}):= \{  \mathtt{s}\in \mathscr{S}(\textbf{a}): \mathtt{s} \text{ is a homomorphism of algebras}\}\subseteq \mathscr{S}(\textbf{a})
.\]
Elements $\mathtt{e}\in \mathscr{E}(\textbf{a})$ are called \emph{exponential series}. Non-trivial examples exist as long as $\textbf{a}$ has some elements which commute. It is shown in \cite[Section 12.2]{aguiar2013hopf} that if $\textbf{a}$ is commutative, then $\mathscr{E}(\textbf{a})$ is a commutative group of units in $\mathscr{S}(\textbf{a})$, which is naturally isomorphic to $\textbf{a}[1]$ (under addition of vectors). In \autoref{prop:exp}, we give an analog of this result in a setting with decorations.

We also have the subset
\[        
\mathscr{G}(\textbf{c}):= \{  \mathtt{s}\in \mathscr{S}(\textbf{c}): \mathtt{s} \text{ is a homomorphism of coalgebras}\}\subseteq \mathscr{S}(\textbf{c})
.\]
Elements $\mathtt{g}\in \mathscr{G}(\textbf{c})$ are called \emph{group-like} series. If $\textbf{h}$ is a Hopf algebra in species, then the set $\mathscr{G}(\textbf{h})$ is a group of units in the $\Bbbk$-algebra $\mathscr{S}(\textbf{h})$.\footnote{\ moreover, $\mathscr{S}(\textbf{h})$ is naturally a complete Hopf algebra with group-like elements $\mathscr{G}(\textbf{h})$ \cite[Section 12.5]{aguiar2013hopf}} If $\text{s}:\textbf{h}\to \textbf{h}$ is the antipode of $\textbf{h}$, then the inverse of $\mathscr{G}(\textbf{h})$ is 
\[
\mathscr{S}(\text{s}):\mathscr{G}(\textbf{h})\to \mathscr{G}(\textbf{h})
.\] 
Thus
\begin{equation} \label{interse}
\mathtt{g}\ast (\text{s}\circ\mathtt{g})=(\text{s}\circ\mathtt{g})\ast \mathtt{g}= 1_{\mathscr{S}( \textbf{h} )} 
.
\end{equation}
The fact that $\mathscr{G}(\textbf{h})$ is closed under multiplication follows from the compatibility property between the multiplication and comultiplication of $\textbf{h}$. %See \cite[Section 12.3]{aguiar2013hopf} for more details. 

%\[
%\Hom_{\textsf{Mon}([\sfS^{\op},\textsf{Vec}] )}(\textbf{E}, \textbf{a}),  
%\qquad     
%\mathscr{G}(\textbf{c}):=  \Hom_{\textsf{Comon}([\sfS^{\op},\textsf{Vec}] )}(\textbf{E}^\ast , %\textbf{c})   
%.\]
%%%%%%%%%%%%%%%%%%%%%%%%%%%%%%%%%%%%%%%%%%%%%%%%%%%%%%%%%%%%%%%%%%%%%%%%
\subsection{The Endofunctor $(-)^{\textbf{E}}$} \label{sec:theendoE}
%%%%%%%%%%%%%%%%%%%%%%%%%%%%%%%%%%%%%%%%%%%%%%%%%%%%%%%%%%%%%%%%%%%%%%%%
Recall the endofunctor $(-)^{\textbf{E}}$ from \autoref{sec: coalg}, 
\[   
(-)^{\textbf{E}}:   
[\sfS^{\op},\textsf{Vec}]\to[\sfS^{\op},\textsf{Vec}]
,\qquad
\textbf{p}\mapsto \textbf{p}^{\textbf{E}} \quad \text{and} \quad \eta\mapsto \eta^{\textbf{E}}:=\eta^{\text{id}_\textbf{E}}
.\]
Explicitly, elements of $\textbf{p}^{\textbf{E}}[I]$ are series of the $\formg$-colored species $\textbf{p}^{[-]}[I]$, denoted
\[
\mathtt{x}=\mathtt{x}_{(-)}: \textbf{E}\to \textbf{p}^{[-]}[I]
,\qquad
\tH_{Y}\mapsto \mathtt{x}_Y
.\]
We have the following isomorphism of vector spaces,
\begin{equation}\label{eq:isoseriesr}
\textbf{p}^{\textbf{E}}[I]=\mathscr{S}\big (\textbf{p}^{[-]}[I]\big )
\xrightarrow{\sim}\prod_{r=0}^\infty\big (\textbf{p}^{[r]}[I]\big )^{ \sfS_r }, 
\qquad  
\mathtt{x}_{(-)}\mapsto   \sum_{r=0}^\infty \mathtt{x}_r
.
\end{equation}
Since $\textbf{E}$ is a coalgebra, the functor $(-)^{\textbf{E}}$ is lax monoidal. Recalling \textcolor{blue}{(\refeqq{eq:multofa^c})} and \textcolor{blue}{(\refeqq{eq:unitofa^c})}, if $\textbf{a}$ is an algebra in species, then so is $\textbf{a}^{\textbf{E}}$, with multiplication given by
\begin{equation}\label{eq:multofNderiv}
\big (\mu_{S,T}(\mathtt{x}\otimes \mathtt{y})\big)_{(-)}
= \! \! 
\underbrace{
\sum_{Y_1 \sqcup Y_2 =(-)}
\mu_{Y_1 \sqcup S, Y_2 \sqcup T}(\mathtt{x}_{Y_1} \otimes \mathtt{y}_{Y_2})
=
\sum_{r=0}^\infty \sum_{r_1 + r_2 =r}
\dfrac{r!}{r_1 !\,  r_2 !}
\mu_{[r_1] \sqcup S, [r_2] \sqcup T}(\mathtt{x}_{r_1} \otimes \mathtt{y}_{r_2})}_{\text{via the isomorphism \textcolor{blue}{(\refeqq{eq:isoseriesr})}}}
\end{equation}
and unit $\mathtt{1}_{\textbf{a}^{\textbf{E}}}\in \textbf{a}^{\textbf{E}}[\emptyset]$ the series given by
\begin{equation}\label{eq:unitofNderiv}
\mathtt{1}_{\textbf{a}^\textbf{E}} :\textbf{E}\to \textbf{a}^{[-]}[\emptyset], 
\qquad                
{(\mathtt{1}_{\textbf{a}^\textbf{E}})}_Y :=  
\begin{cases}
\mathtt{1}_\textbf{a} &\quad \text{if $Y=\emptyset$}\\
0            &\quad \text{otherwise}.
\end{cases}
\end{equation}
%%%%%%%%%%%%%%%%%%%%%%%%%%%%%%%%%%%%%%%%%%%%%%%%%%%%%%%%%%%%%%%%%%%%%%%%
\section{From $\Bbbk$-Algebras to Algebras in Species and Back}
%%%%%%%%%%%%%%%%%%%%%%%%%%%%%%%%%%%%%%%%%%%%%%%%%%%%%%%%%%%%%%%%%%%%%%%%
Given a $\Bbbk$-algebra $\cA$, we now show how the formal power series algebras $\cA[[\formj]]$ and  $\cA[[\formg,\! \formj]]$ arise as algebras of series of species. 
%%%%%%%%%%%%%%%%%%%%%%%%%%%%%%%%%%%%%%%%%%%%%%%%%%%%%%%%%%%%%%%%%%%%%%%%
\subsection{The Algebra $\textbf{U}_\cA$} 
%%%%%%%%%%%%%%%%%%%%%%%%%%%%%%%%%%%%%%%%%%%%%%%%%%%%%%%%%%%%%%%%%%%%%%%%
Given a vector space $V$, we have the species $\textbf{U}_V$ given by 
\[
\textbf{U}_V[I]:=V 
\qquad 
\text{and} 
\qquad       
\textbf{U}_V[\sigma]:= \text{id}_{V}      
.\]
Given a linear map $f:V \to W$, we have the morphism of species $\textbf{U}_f$ given by
\[
\textbf{U}_f: \textbf{U}_{V}\to \textbf{U}_{W}
,\qquad
(\textbf{U}_f)_I:= f
.\]
This defines a functor 
\[
\textbf{U}_{(-)}: \textsf{Vec} \to [ \sfS^{\op}, \textsf{Vec}],
\qquad
V\mapsto \textbf{U}_V \quad \text{and} \quad f\mapsto \textbf{U}_f
.\]
In fact, $\textbf{U}_{(-)}$ is the left adjoint of $\mathscr{S}(-)$. Let $\cA$ be a $\Bbbk$-algebra, with multiplication and unit denoted respectively by
\[    
 \cA \otimes \cA \to \cA, 
\qquad   
a\otimes b \mapsto a\star  b     
\qquad \quad 
\text{and} 
\qquad       \quad   
\Bbbk\to \cA, 
\qquad 
1_\Bbbk\mapsto 1_\cA   
.\]
%For example, $\cA$ may be an algebra of functionals on a space of functions. 
Then $\textbf{U}_\cA$ is an algebra in species, with multiplication and unit given respectively by
\begin{equation}\label{eq:algU}
\mu_{S,T}(a\otimes b):= a\star b 
\qquad\text{and}\qquad 
\mathtt{1}_{\textbf{U}_\cA}:=1_\cA
.
\end{equation}
This construction was given in \cite[Section 2.10]{aguiar2013hopf}. 

%We have the adjoint pair
%\begin{center}
%\begin{tikzcd}[column sep=huge,row sep=large] 
%\text{[}\sfS^{\text{op}},\textsf{Vec}\text{]} 	    
%\arrow[r, shift right=1.142ex,  "\bot" , "\text{ $\mathscr{S}(-)$ }"']   
%&   
%\textsf{Vec}   
%\arrow[l, shift right=1.142ex, "\textbf{U}_{(-)}"'] .
%\end{tikzcd}
%\end{center}

%\begin{remark}
%We have the adjoint pair
%\begin{center}
%\begin{tikzcd}[column sep=huge,row sep=large] 
%\text{[}\sfS^{\text{op}},\textsf{Vec}\text{]} 	    
%\arrow[r, shift right=1.142ex,  "\bot" , "\text{ $\mathscr{S}(-)$ }"']   
%&   
%\textsf{Vec}   
%\arrow[l, shift right=1.142ex, "\textbf{U}_{(-)}"'] .
%\end{tikzcd}
%\end{center}
%The counit $\varepsilon: \textbf{U}_{\mathscr{S}(-)} \to \text{id}_{[\sfS^{\op}, \textsf{Vec}]}$ is given by
%\[
%\varepsilon_\textbf{p}:\textbf{U}_{\mathscr{S}(\textbf{p})}\to \textbf{p},
%\qquad
%(\varepsilon_\textbf{p})_I (\mathtt{s}) :=  \mathtt{s}_I
%.\] 
%Indeed, given $\eta:\textbf{U}_V\to \textbf{p}$, then $f:V\to \mathscr{S}(\textbf{p})$, $f(v)_I:=\eta_I(v)$ is the unique morphism such that $\varepsilon_\textbf{p} \circ \textbf{U}_f=\eta$. The unit $\iota: \text{id}_{\textsf{Vec}} \to %\mathscr{S}(\textbf{U}_{(-)}) $ is given by
%\[
%\iota_V : V\to \mathscr{S}(\textbf{U}_{V}),
%\qquad
%\iota_V (v)_I:=v
%.\]
%Indeed, given $f:V\to \mathscr{S}(\textbf{p})$, then $\eta:\textbf{U}_V\to \textbf{p}$, $\eta_I(v):= f(v)_I$ is the unique morphism such that $\mathscr{S}(\eta)\circ \iota_V=f$.
%\end{remark}
%%%%%%%%%%%%%%%%%%%%%%%%%%%%%%%%%%%%%%%%%%%%%%%%%%%%%%%%%%%%%%%%%%%%%%%%
\subsection{Series of a $\Bbbk$-Algebra}
%%%%%%%%%%%%%%%%%%%%%%%%%%%%%%%%%%%%%%%%%%%%%%%%%%%%%%%%%%%%%%%%%%%%%%%%
Let $\cA$ be a $\Bbbk$-algebra, and let $\cA[[\formj]]$ denote the $\Bbbk$-algebra of formal power series in the formal symbol $\formj$ with coefficients in $\cA$. The multiplication of $\cA[[\formj]]$, which we also denote by $\star$, is given by
\begin{equation}\label{eq:multofformalaj}    
\bigg(
\sum^\infty_{n=0} \formj^n  a_n \bigg) \star \bigg( \sum^\infty_{n=0}   \formj^n b_n 
\bigg) 
:=
\sum^\infty_{n=0} \formj^n \sum_{s+t=n} a_s \star b_t    
.
\end{equation}
The unit $1_{\cA[[\formj]]}$ is the formal power series $\sum^\infty_{n=0} \formj^n  a_n$ with $a_0=1_\Bbbk$ and $a_n=0$ otherwise. Also, $\cA[[\formj]]$ may be equipped with the derivation called \emph{formal differentiation}, given by
\begin{equation}\label{eq:formalder}
\dfrac{d}{d\formj}: 
\cA[[\formj]] \to \cA[[\formj]]
,\qquad
\dfrac{d}{d\formj} \sum^\infty_{n=0} \formj^n  a_n 
:=
\sum^\infty_{n=0} \formj^{n} (n+1)a_{n+1} 
.
\end{equation}

\begin{prop} \label{prop:iso}
The isomorphism of vector spaces
\[
\mathscr{S}(\textbf{U}_\cA ) \xrightarrow{\sim} \cA[[\formj]],
\qquad
\mathtt{s}\mapsto   \sum^\infty_{n=0} 
\dfrac{\formj^n}{n!} \mathtt{s}_n 
\]
is an isomorphism of $\Bbbk$-algebras.
\end{prop}
\begin{proof}
For series $\mathtt{s},\mathtt{t}\in \mathscr{S}(\textbf{U}_\cA )$, the multiplication of their images in $\cA[[\formj]]$ is given by \textcolor{blue}{(\refeqq{eq:multofformalaj})}, 
\[
\bigg(\sum^\infty_{n=0} 
\dfrac{\formj^n}{n!} \mathtt{s}_n  \bigg)
\star 
\bigg(\sum^\infty_{n=0} 
\dfrac{\formj^n}{n!} \mathtt{t}_n \bigg)
=
\sum^\infty_{n=0} \formj^n \sum_{ s +t=n  } \dfrac{1}{s!\, t!}  \mathtt{s}_s \star \mathtt{t}_s   
.\]
On the other hand, their multiplication ${\mathtt{s}}\ast{\mathtt{t}}$ in $\mathscr{S}(\textbf{U}_\cA)$ is given by \textcolor{blue}{(\refeqq{eq:multinS(a)})},
\[
{\mathtt{s}}\ast{\mathtt{t}}_{(-)}=
\sum_{S\sqcup T=(-)} {\mathtt{s}}_{S} \star {\mathtt{t}}_{T}
.\]
Then the image of ${\mathtt{s}}\ast{\mathtt{t}}$ in $\cA[[\formj]]$ is 
\[
\sum^\infty_{n=0} 
\dfrac{\formj^n}{n!} \sum_{S\sqcup T=[n]} {\mathtt{s}}_{S} \star {\mathtt{t}}_{T}
=
\sum^\infty_{n=0} 
\dfrac{\formj^n}{n!}   
\sum_{ s +t=n  } 
\dfrac{n!}{s!\,  t!} {\mathtt{s}}_{s} \star {\mathtt{t}}_{t}
=
\sum^\infty_{n=0} \formj^n \sum_{ s +t=n  } 
\dfrac{1}{s!\,  t!} {\mathtt{s}}_{s} \star {\mathtt{t}}_{t}
.\]
Thus, the isomorphism $\mathscr{S}(\textbf{U}_\cA ) \xrightarrow{\sim} \cA[[\formj]]$ preserves the multiplication. The preservation of units is clear. 
\end{proof}

\begin{remark}\label{rem:der}
Given a series $\mathtt{s}\in \mathscr{S}(\textbf{U}_\cA)$, we can take the image of $\mathtt{s}$ under the endofunctor $(-)'$ from \autoref{sec:inthomcauchy2},
\[
\mathtt{s}': \textbf{E}' \to \textbf{U}'_\cA
.\]
We have natural isomorphisms $\textbf{E}'\cong \textbf{E}$ and $\textbf{U}'_\cA\cong \textbf{U}_\cA$, so that we may take $\mathtt{s}'\in \mathscr{S}(\textbf{U}_\cA)$. Then, via the isomorphism of \autoref{prop:iso}, we have
\[
\mathtt{s}'= \dfrac{d \mathtt{s}}{d\formj}
.\]
\end{remark}

\begin{ex}
Putting $\cA=\Bbbk$, then $\textbf{U}_\Bbbk=\textbf{E}$, and
\[
\mathscr{S}(\textbf{E})\cong \Bbbk[[\formj]] 
.\]
This is \cite[Equation 157]{aguiar2013hopf}.
\end{ex}

%When we take limits of homomorphisms $\textbf{a}\to \textbf{U}_\cA$, we shall compose with this isomorphism to obtain homomorphisms $\mathscr{S}(\textbf{a})\to \cA[[\formj]]$.

%Since the symmetric group actions are trivial, this is also induced by the inverse norm transform $N^{-1}_{\textbf{U}_\cA}$. 

We have that $\textbf{U}^{\textbf{E}}_\cA:=(\textbf{U}_\cA)^{\textbf{E}}$ is an algebra in species via the lax monoidal structure of the endofunctor $(-)^{\textbf{E}}$. Let $\cA[[\formg]]$ denote the $\Bbbk$-algebra of formal power series in the formal symbol $\formg$  with coefficients in $\cA$. 

\begin{prop}\label{prop:g}
The isomorphism of species
\[  
\textbf{U}^{\textbf{E}}_\cA \xrightarrow{\sim} \textbf{U}_{ \cA[[\formg]] }   
,\qquad  
\mathtt{x} 
\mapsto 
\sum^\infty_{r=0} \dfrac{\formg^{r}}{r!} \mathtt{x}_r  
\]
is an isomorphism of algebras in species.
\end{prop}
\begin{proof}
From \textcolor{blue}{(\refeqq{eq:multofNderiv})}, we see that the multiplication of $\textbf{U}^{\textbf{E}}_\cA$ is given by
\[
\mu_{S,T}(\mathtt{x}_{(-)} \otimes  \mathtt{y}_{(-)})
=
\sum_{Y_1 \sqcup Y_2 =(-)} 
\mathtt{x}_{Y_1} \star \mathtt{y}_{Y_2} 
.\] 
Then
\[
\sum_{r=0}^\infty \dfrac{\formg^{r}}{r!}
\sum_{Y_1 \sqcup Y_2 =[r]} 
\mathtt{x}_{Y_1} \star \mathtt{y}_{Y_2} 
= 
\sum_{r=0}^\infty \dfrac{\formg^{r}}{r!}
\sum_{r_1+r_2 =r}  \dfrac{r!}{r_1 ! r_2 !}\, 
\mathtt{x}_{r_1} \star \mathtt{y}_{r_2} 
=
\mu_{S,T}\bigg( \sum^\infty_{r=0} \dfrac{\formg^{r}}{r!} \mathtt{x}_r , \sum^\infty_{r=0} \dfrac{\formg^{r}}{r!} \mathtt{y}_r \bigg)
.\]
Thus, the isomorphism $\textbf{U}^{\textbf{E}}_\cA \xrightarrow{\sim}\textbf{U}_{\cA[[\formg]]}$ preserves the multiplication. The preservation of units is clear. 
\end{proof}

\begin{cor}
The isomorphism of vector spaces
\[  
\mathscr{S}(\textbf{U}^{\textbf{E}}_\cA) \xrightarrow{\sim} \cA[[\formg,\! \formj]]    
,\qquad  
\mathtt{s}
\mapsto 
\sum^\infty_{n=0} \sum^\infty_{r=0} \dfrac{\formg^{\, r}\formj^{n}}{r!\,  n!} (\mathtt{s}_n)_r  
\]
is an isomorphism of $\Bbbk$-algebras.
\end{cor}
%%%%%%%%%%%%%%%%%%%%%%%%%%%%%%%%%%%%%%%%%%%%%%%%%%%%%%%%%%%%%%%%%%%%%%%%
%%%%%%%%%%%%%%%%%%%%%%%%%%%%%%%%%%%%%%%%%%%%%%%%%%%%%%%%%%%%%%%%%%%%%%%%
%%%%%%%%%%%%%%%%%%%%%%%%%%%%%%%%%%%%%%%%%%%%%%%%%%%%%%%%%%%%%%%%%%%%%%%%
\section{Modules in Species}\label{sec:Modules in Species}
%%%%%%%%%%%%%%%%%%%%%%%%%%%%%%%%%%%%%%%%%%%%%%%%%%%%%%%%%%%%%%%%%%%%%%%%
%%%%%%%%%%%%%%%%%%%%%%%%%%%%%%%%%%%%%%%%%%%%%%%%%%%%%%%%%%%%%%%%%%%%%%%%
%%%%%%%%%%%%%%%%%%%%%%%%%%%%%%%%%%%%%%%%%%%%%%%%%%%%%%%%%%%%%%%%%%%%%%%%
We recall aspects of the general theory of modules over Hopf monoids internal to symmetric monoidal categories, specialized to Hopf monoids in species. We then focus on gebras internal to $\textbf{E}$-modules, and view them (via currying) as examples of $(-)^{\textbf{E}}$-coalgebras. 

%The case relevant to pAQFT will be the realization of the Hopf algebra of compositions $\Sig$ as a Hopf monoid internal to $\textbf{E}$-modules. 
%%%%%%%%%%%%%%%%%%%%%%%%%%%%%%%%%%%%%%%%%%%%%%%%%%%%%%%%%%%%%%%%%%%%%%%%
\subsection{$\textbf{h}$-Modules and $\textbf{h}$-Gebras} \label{sec:Modules}
%%%%%%%%%%%%%%%%%%%%%%%%%%%%%%%%%%%%%%%%%%%%%%%%%%%%%%%%%%%%%%%%%%%%%%%%
%In this section, we recall analogs of familiar constructions regarding representations of Hopf algebras, but in the setting of species. 

Given a Hopf algebra in $\formg$-colored species $\textbf{h}$, let an $\textbf{h}$\emph{-module} $\textbf{m}=(\textbf{m},\rho)$ be a $\formj$-colored species $\textbf{m}$ equipped with a morphism of species of the form
\[
\rho:\textbf{h}\bigcdot \textbf{m}\to \textbf{m}
,\qquad
\mathtt{h} \otimes \mathtt{v} \mapsto \rho_{Y,I}( \mathtt{h} \otimes \mathtt{v} )
\]
called the (uncurried) \emph{action}, which satisfies associativity
\[
\rho_{Y_1\sqcup Y_2,I}\big( \mu_{Y_1,Y_2}(\mathtt{h}_1 \otimes \mathtt{h}_2)   \otimes \mathtt{v} \big)
=
\rho_{Y_1,Y_2\sqcup I}\big(  (\mathtt{h}_1 \otimes  \rho_{Y_2,I}(\mathtt{h}_2   \otimes \mathtt{v})\big)
\]
and unitality
\[
\rho_{\emptyset,I}( \mathtt{1}_\textbf{h}\otimes \mathtt{v})=\mathtt{v} 
.\]
Recall the algebra of raising operators $\CMcal{E}^{\bigcdot}(\textbf{m})$ from \textcolor{blue}{(\refeqq{eq:multcomp})}. The \emph{curried action} is the homomorphism of algebras given by
\[
\wh{\rho}:\textbf{h}\to \cH^{\bigcdot}(\textbf{m}, \textbf{m})
,\qquad
\mathtt{h} \mapsto \wh{\rho}_Y(\mathtt{h})
\]
where $\wh{\rho}_Y(\mathtt{h})$ is the raising operator given by
\[
\wh{\rho}_Y(\mathtt{h}): \textbf{m}\to \textbf{m}^{[Y]}
,\qquad
\mathtt{v}\mapsto \rho_{Y,I}(\mathtt{h}\otimes\mathtt{v})
.\]
In this way, an $\textbf{h}$-module $\textbf{m}$ is equivalent to a homomorphism (`representation') $\textbf{h}\to \CMcal{E}^{\bigcdot}(\textbf{m})$. Notice that $\textbf{h}$-modules also induce $(-)^{\textbf{h}}$-coalgebras, given by   
\[
\widecheck{\rho}:\textbf{m}\to \textbf{m}^{\textbf{h}}
,\qquad
\mathtt{v}\mapsto \widecheck{\rho}_I(\mathtt{v})
\] 
where $\widecheck{\rho}_I(\mathtt{v})$ is the $\textbf{h}$-element given by
\[
\widecheck{\rho}_I(\mathtt{v}):
\textbf{h} \to \textbf{m}^{[-]}[I]
,\qquad
\mathtt{h}\mapsto \rho_{Y,I}(\mathtt{h}\otimes\mathtt{v})
.
\]
For $\textbf{h}$-modules $\textbf{m}=(\textbf{m},\rho)$ and $\textbf{n}=(\textbf{n},\tau)$, a \emph{morphism} $\eta:\textbf{m}\to \textbf{n}$ of $\textbf{h}$-modules is a morphism of the underlying species $\eta: \textbf{m}\to \textbf{n}$ such that 
\[
\eta_{Y\sqcup I}\big( \rho_{Y,I}(\mathtt{h}\otimes \mathtt{v})\big)  =\tau_{Y,I}\big ( \mathtt{h}\otimes \eta_I(\mathtt{v})\big ) 
.\]
For the uncurried action, this is an intertwiner of actions; that is, a morphism $\eta: \textbf{m}\to \textbf{n}$ such that the following diagram commutes,
\begin{center}
\begin{tikzcd}[column sep=large,row sep=large] 
\textbf{h}\bigcdot\textbf{m} \arrow[d, "\text{id}\, \bigcdot\, \eta"'] \arrow[r,   "\rho"]   
& 
\textbf{m} \arrow[d,   "\eta"]     
\\
\textbf{h}\bigcdot \textbf{n} \arrow[r,   "\tau"']     
& 
\textbf{n}.
\end{tikzcd}
\end{center}
%For the curried action, it is a morphism $\eta: \textbf{m}\to \textbf{n}$ such that
%\[
%\cH^{\bigcdot}(\eta, \eta ) \circ \widehat{\rho} = \widehat{\tau} 
%.\]
It also faithfully determines a morphism of $(-)^{\textbf{h}}$-coalgebras. A morphism of $(-)^{\textbf{h}}$-coalgebras is a commutative diagram in $[\sfS^{\op}, \sfVec]$ of the form
\begin{center}
\begin{tikzcd}[column sep=large,row sep=large] 
\textbf{m} \arrow[d, "\eta"'] \arrow[r,   "\widecheck{\rho}"]   
& 
\textbf{m}^{\textbf{h}} \arrow[d,   "\eta^{\textbf{h}}"]     
\\
\textbf{n} \arrow[r,   "\widecheck{\tau}"']     
& 
\textbf{n}^{\textbf{h}}.
\end{tikzcd}
\end{center}
Morphism composition is the obvious pasting of diagrams. (One can identify the morphism with $\eta$, however we shall prefer to think of the morphism as being the whole diagram.) 

For $\textbf{h}$-modules $\textbf{m}=(\textbf{m},\rho)$ and $\textbf{n}=(\textbf{n},\tau)$, their \emph{Cauchy product} $\textbf{m} \bigcdot \textbf{n}$ is the $\textbf{h}$-module whose action is given by the composition
\begin{equation}\label{eq:tensormodules}
\textbf{h}\bigcdot( \textbf{m} \bigcdot \textbf{n})
\xrightarrow{\Delta\,  \bigcdot\,  \text{id}} 
(\textbf{h}\bigcdot \textbf{h}) \bigcdot( \textbf{m} \bigcdot \textbf{n})
\xrightarrow{\sim} 
(\textbf{h}\bigcdot \textbf{m}) \bigcdot( \textbf{h} \bigcdot \textbf{n})
\xrightarrow{\rho\, \bigcdot\, \tau} 
\textbf{m} \bigcdot \textbf{n}
.
\end{equation}
This defines the monoidal category \hbox{$\textsf{Rep}(\textbf{h})$} of \hbox{$\textbf{h}$-modules}. The unit object $\textbf{1}$ is the $\textbf{h}$-module given by 
\[
\textbf{h}\bigcdot\textbf{1}\to \textbf{1}
,\qquad 
\mathtt{h}\otimes 1_\Bbbk \mapsto \epsilon_Y(\mathtt{h})
\] 
where $\epsilon$ is the counit of $\textbf{h}$.

%The fiber functor \hbox{$\textsf{Rep}(\textbf{h})\to [\sfS^{\op}, \sfVec]$}, which forgets the $\textbf{h}$-action, is strong monoidal with respect to the Cauchy product. %\footnote{\ One can use Sweedler notation to write this more explicitly. The comultiplications of the particular cases we are concerned with are `set-theoretic', and so we shall not need Sweedler notation.}

Let $\textbf{h}$ be a cocommutative Hopf algebra, so that $\textsf{Rep}(\textbf{h})$ is a symmetric monoidal category. Let a $\textbf{h}$-\emph{(bi/co)algebra} be a (bi/co)monoid internal to $\textsf{Rep}(\textbf{h})$. Let a \emph{Hopf/Lie} $\textbf{h}$-\emph{algebra} be a Hopf monoid/Lie algebra internal to $\textsf{Rep}(\textbf{h})$. We make these definitions fully explicit for the cases which concern us below, namely $\textbf{h}\in \{ \textbf{L}, \textbf{E} \}$.

Let $\textbf{a}=(\textbf{a},\rho)$ be an $\textbf{h}$-algebra. Then, for primitive elements $\mathtt{h}\in \mathcal{P}(\textbf{h})[Y]$, the operator \hbox{$\wh{\rho}_Y(\mathtt{h}):\textbf{a}\to \textbf{a}^{[Y]}$} is a $Y$-derivation of $\textbf{a}$. Therefore the curried action $\wh{\rho}:\textbf{h}\to \cH^{\bigcdot}(\textbf{a}, \textbf{a})$ restricts to a homomorphism of Lie algebras
\[
\mathcal{P}(\textbf{h})\to \text{Der}(\textbf{a})
,\qquad
\mathtt{h} \mapsto \wh{\rho}_Y(\mathtt{h})
.\] 

%\begin{remark}
%If we set $\textbf{h}=\textbf{1}$, then we recover usual species and algebras in species. 
%\end{remark}

%In this paper, we only need the cases $\textbf{h}\in \{\textbf{E}, \textbf{L}\}$, which can be consider in terms of up operators.

\begin{thm} \label{prop:is a homo}
Let $\textbf{a}$ be an algebra in species, and suppose that the underlying species of $\textbf{a}$ has the structure of an $\textbf{h}$-module $(\textbf{a},\rho)$. Then $\textbf{a}$ is an $\textbf{h}$-algebra if and only if the associated $(-)^{\textbf{h}}$-coalgebra \[
\widecheck{\rho}:\textbf{a}\to\textbf{a}^\textbf{h}
\] 
is a homomorphism of algebras.
\end{thm}
\begin{proof}
%Let $(\mathtt{h}_x)_{x\in \cX}$ be a basis of $\textbf{h}$. 
Given $\mathtt{h}\in \textbf{h}[Y]$, let
\[
\Delta_{Y_1,Y_2}(\mathtt{h})=: 
\sum_{x\in X} {\mathtt{h}_{Y_1}}_x \otimes {\mathtt{h}_{Y_2}}_x
\]
where $X$ is some set of indices. Then $\textbf{a}$ is an $\textbf{h}$-algebra if its multiplication and unit are morphisms of $\textbf{h}$-modules, that is if
\begin{equation}\label{eq:actionohh}
\rho_{Y,I}\big (\mathtt{h}\otimes\mu_{S,T}(\mathtt{a}\otimes \mathtt{b})\big) 
= 
\sum_{Y_1\sqcup Y_2 =Y} 
\sum_{x\in X}
\mu_{Y_1\sqcup S,Y_2\sqcup T}\big 
( \rho_{Y_1,S}({\mathtt{h}_{Y_1}}_x \otimes \mathtt{a})
\otimes 
\rho_{Y_2,T}({\mathtt{h}_{Y_2}}_x\otimes  \mathtt{b})
\big )
\end{equation}
and
\[
\rho_{Y,\emptyset}(\mathtt{h}\otimes \mathtt{1}_\textbf{a})= \epsilon_Y(\mathtt{h}) \mathtt{1}_\textbf{a}
.\]
On the other hand, the $(-)^{\textbf{h}}$-coalgebra $\widecheck{\rho}:\textbf{a}\to\textbf{a}^\textbf{h}$ preserves the multiplication if
\begin{equation}\label{eq:actionohh2}
\textbf{a}\bigcdot \textbf{a} \xrightarrow{\mu} \textbf{a} 
\xrightarrow{\widecheck{\rho}}
\textbf{a}^\textbf{h} 
\ = \ 
\textbf{a}\bigcdot \textbf{a} 
\xrightarrow{\widecheck{\rho}\, \bigcdot\, \widecheck{\rho}}  
\textbf{a}^{\textbf{h}}\bigcdot \textbf{a}^{\textbf{h}}
\xrightarrow{\mu}  
\textbf{a}^{\textbf{h}}
.
\end{equation}
(Recall the multiplication $\textbf{a}^{\textbf{h}} \bigcdot \textbf{a}^{\textbf{h}}
\xrightarrow{\mu}  
\textbf{a}^{\textbf{h}}$ from \textcolor{blue}{(\refeqq{eq:multofa^c})}.) We then observe that the left-hand sides of \textcolor{blue}{(\refeqq{eq:actionohh})} and \textcolor{blue}{(\refeqq{eq:actionohh2})} are equal, and similarly for the right-hand sides. The $(-)^{\textbf{h}}$-coalgebra $\widecheck{\rho}:\textbf{a}\to\textbf{a}^\textbf{h}$ preserves the unit if
\[
\widecheck{\rho}(\mathtt{1}_\textbf{a})= \mathtt{1}_{\textbf{a}^{\textbf{h}}}
.\]
That is, if
\[
\rho_{Y,\emptyset}( \mathtt{h} \otimes \mathtt{1}_\textbf{a} )   = ( \iota^{[-]}_\emptyset \circ \epsilon)_Y  (\mathtt{h})=\epsilon_Y(\mathtt{h}) \mathtt{1}_\textbf{a}
.\]
(Recall the unit of $\textbf{a}^{\textbf{h}}$ from \textcolor{blue}{(\refeqq{eq:unitofa^c})}.)
\end{proof}
%%%%%%%%%%%%%%%%%%%%%%%%%%%%%%%%%%%%%%%%%%%%%%%%%%%%%%%%%%%%%%%%%%%%%%%%
\subsection{$\textbf{g}$-Modules}
%%%%%%%%%%%%%%%%%%%%%%%%%%%%%%%%%%%%%%%%%%%%%%%%%%%%%%%%%%%%%%%%%%%%%%%%
%In this section, we recall analogs of familiar constructions regarding representations of Lie algebras, but in the setting of species. 

Given a Lie algebra in species $\textbf{g}$, let a $\textbf{g}$\emph{-module} $\textbf{m}=(\textbf{m},\rho)$ be a species $\textbf{m}$ equipped with a morphism of species of the form
\[
\rho:\textbf{g}\bigcdot \textbf{m}\to \textbf{m}
,\qquad
\mathtt{g} \otimes \mathtt{v} \mapsto \rho_{Y,I}(\mathtt{g} \otimes \mathtt{v})
\]
such that 
\[
\rho_{Y,I}\big ( [\mathtt{g}_1, \mathtt{g}_2]_{Y_1, Y_2} \otimes \mathtt{v} \big) 
=  
\rho_{Y_1, Y_2\sqcup I}\big( \mathtt{g}_1 \otimes \rho_{Y_2,I} (\mathtt{g}_2 \otimes \mathtt{v})\big) 
-
\rho_{Y_1\sqcup I,Y_2}\big(\rho_{Y_1,I} (\mathtt{g}_1 \otimes \mathtt{v})\otimes \mathtt{g}_2 \big) 
.\]
Notice that a morphism of species $\rho:\textbf{g}\bigcdot \textbf{m}\to \textbf{m}$ is a $\textbf{g}$-module if and only if its currying
\[
\wh{\rho}:\textbf{g}\to \CMcal{E}^{\bigcdot}(\textbf{m})
\]
is a homomorphism of Lie algebras (where $\CMcal{E}^{\bigcdot}(\textbf{m})$ is equipped with the commutator bracket). The restriction of an $\textbf{h}$-module to the primitive part $\mathcal{P}(\textbf{h})$ is a $\mathcal{P}(\textbf{h})$-module. Conversely, a $\textbf{g}$-module determines a $\mathcal{U}(\textbf{g})$-module since a homomorphism of Lie algebras $\textbf{g}\to \CMcal{E}^{\bigcdot}(\textbf{m})$ will uniquely extend to a homomorphism of algebras $\mathcal{U}(\textbf{g}) \to \CMcal{E}^{\bigcdot}(\textbf{m})$. This defines an equivalence of categories
\[
\textsf{Rep}(\textbf{g})\xrightarrow{\sim} \textsf{Rep}\big(\mathcal{U}(\textbf{g})\big )
.\]
To make this equivalence a monoidal functor, for $\textbf{g}$-modules $(\textbf{m},\rho)$ and $(\textbf{n},\tau)$, we define their \emph{Cauchy product} to be the $\textbf{g}$-module given by
\[
\textbf{g}\bigcdot( \textbf{m} \bigcdot \textbf{n})\to  \textbf{m} \bigcdot \textbf{n}
,\qquad 
\mathtt{g} \otimes (\mathtt{v} \otimes \mathtt{w})
\mapsto 
\rho(\mathtt{g}\otimes \mathtt{v}) \otimes \mathtt{w}   
+ 
\mathtt{v} \otimes \tau( \mathtt{g} \otimes \mathtt{w}) 
.\]
Let a $\textbf{g}$-\emph{(bi/co)algebra} be a (bi/co)monoid internal to $\textsf{Rep}(\textbf{g})$. Let a \emph{Hopf/Lie} $\textbf{g}$-\emph{algebra} be a Hopf monoid/Lie algebra internal to $\textsf{Rep}(\textbf{g})$. 

%Thus, for $\textbf{g}$ a positive Lie algebra, we obtain a symmetric monoidal category $\textsf{Rep}(\textbf{g})$ of $\textbf{g}$-modules. Let a $\textbf{g}$-(bi/co)algebra be a (bi/co)monoid internal to $\textsf{Rep}(\textbf{g})$. Let a Hopf/Lie $\textbf{g}$-algebra be a Hopf monoid/Lie algebra internal to $\textsf{Rep}(\textbf{g})$.
%%%%%%%%%%%%%%%%%%%%%%%%%%%%%%%%%%%%%%%%%%%%%%%%%%%%%%%%%%%%%%%%%%%%%%%%
\subsection{$\textbf{L}$-Modules and $\textbf{L}$-(Co)algebras}
%%%%%%%%%%%%%%%%%%%%%%%%%%%%%%%%%%%%%%%%%%%%%%%%%%%%%%%%%%%%%%%%%%%%%%%%
Recall the endofunctor $(-)'$ from \autoref{sec:inthomcauchy2}. The species $\textbf{m}'$ is called the \emph{derivative} of $\textbf{m}$, given by
\[
\textbf{m}': \sfS^{\op}\to \sfVec
,\qquad
I\mapsto    
\textbf{m}\big [\{\ast\} \sqcup  I\big ] 
\quad\text{and}\quad 
\sigma \mapsto
\textbf{m}[\text{id}_{\{\ast\}} \sqcup \sigma]  
.\] 
We have
\begin{equation} \label{prod}
(\textbf{m} \bigcdot \textbf{n})'=\textbf{n}'\bigcdot \textbf{n}\oplus  \textbf{m}\bigcdot \textbf{n}'
\qquad \text{and} \qquad
\textbf{1}'=\textbf{0}
.
\end{equation}
This follows from the fact that a decomposition of $\ast I$ into two blocks is of the form $\ast S \sqcup T$ or $S\sqcup \ast T$ for $S\sqcup T=I$. Thus, $(-)'$ behaves like a derivation on the category $[\textsf{S}^{\op},\sfVec]$ with respect to the Cauchy product. The classical product rule for power series is recovered from \textcolor{blue}{(\refeqq{prod})} via \autoref{rem:der}. The classical chain rule is obtained similarly from the plethystic product. We have the isomorphism of species
\begin{equation}\label{eq:ast}
\textbf{m}^\textbf{X} \xrightarrow{\sim} \textbf{m}'
,\qquad 
\chi \mapsto \chi_\ast( \tH_\ast )
.
\end{equation}
Since $\textbf{L}$ is the free algebra on $\textbf{X}$, an $\textbf{L}$-module $(\textbf{m}, \rho)$ is completely determined by the up operator 
\[
\wh{\rho}(\tH_{(\ast)}): \textbf{m}\to \textbf{m}'
.\] 
In this way, an $\textbf{L}$-module is equivalently a species $\textbf{m}$ equipped with a choice of up operator $\textbf{m}\to \textbf{m}'$. The Cauchy product of species equipped with up operators as defined in \cite[Section 8.12.2]{aguiar2010monoidal} coincides with the Cauchy product of $\textbf{L}$-modules, as defined in \textcolor{blue}{(\refeqq{eq:tensormodules})}. Therefore the (co)monoids with up (co)derivations of \cite[Section 8.12.5]{aguiar2010monoidal} are equivalently $\textbf{L}$-(co)algebras. %In particular, one requires here that the up operator is an up derivation because $\tH_{(\ast)}\in \mathcal{P}(\textbf{L})$. 

We call an up operator $u:\textbf{m}\to \textbf{m}'$ \emph{commutative} if its corresponding $\textbf{L}$-module 
\[
\rho:\textbf{L} \bigcdot \textbf{m} \to \textbf{m}
,\qquad
\rho_{\ast,I}( \tH_{(\ast)}  \otimes \mathtt{v})  :=   u(\mathtt{v})
\] 
factors through $\textbf{L} \bigcdot \textbf{m} \twoheadrightarrow \textbf{E} \bigcdot \textbf{m}$, $\tH_\ell\otimes \mathtt{v}\mapsto \tH_Y\otimes \mathtt{v}$, thus inducing an $\textbf{E}$-module.
%%%%%%%%%%%%%%%%%%%%%%%%%%%%%%%%%%%%%%%%%%%%%%%%%%%%%%%%%%%%%%%%%%%%%%%%
\subsection{$\textbf{E}$-Modules}
%%%%%%%%%%%%%%%%%%%%%%%%%%%%%%%%%%%%%%%%%%%%%%%%%%%%%%%%%%%%%%%%%%%%%%%%
We now make $\textbf{X}$-modules and $\textbf{E}$-modules explicit, where $\textbf{X}$ is taken to be an abelian Lie algebra. Thus, an $\textbf{X}$-module $(\textbf{m}, \rho)$ is a morphism of the form
\[
\rho:\textbf{X}\bigcdot \textbf{m}\to \textbf{m}
,\qquad
\tH_{y}\otimes \mathtt{v} \mapsto \rho_{y,I}( \tH_{y}\otimes \mathtt{v} )
\]
such that, letting $u_{y}:=\wh{\rho}_y(\tH_{y}):\textbf{m}\to \textbf{m}^{[\{ y \}]}$, we have
\begin{equation}\label{com}
u_{y_2}(u_{y_1}(\mathtt{v})) = u_{y_1} (u_{y_2}(\mathtt{v}))
.
\end{equation}
We also let $u:=u_\ast=\wh{\rho}_{\ast }(\tH_{\ast}):\textbf{m}\to \textbf{m}'$, and this up operator determines the $\textbf{X}$-module.
%is determined by this up operator. 
A generic up operator $\textbf{m}\to \textbf{m}'$ is commutative if and only if \textcolor{blue}{(\refeqq{com})} holds. Thus, $\textbf{X}$-modules are equivalent to commutative up operators. 

Note that $\textbf{X}$-modules are also equivalent to $\textbf{E}$-modules, since $\textbf{X}$ is the primitive part of $\textbf{E}$. Explicitly, an $\textbf{E}$-module
\[
\rho:\textbf{E}\bigcdot \textbf{m}\to \textbf{m}
\] 
induces an $\textbf{X}$-module by restricting the action to the primitive part of $\textbf{E}$. Conversely, an $\textbf{X}$-module
\[
\rho:\textbf{X}\bigcdot \textbf{m}\to \textbf{m}
\]
induces an $\textbf{E}$-module, also denoted $\rho$, given by
\begin{equation}\label{eq:inducedE}
\rho:\textbf{E}\bigcdot \textbf{m} \to \textbf{m}
,\qquad
\tH_Y\otimes \mathtt{v} \mapsto 
\rho_{Y,I}(\mathtt{v}):=
u_{y_r}( \dots (u_{y_2}(u_{y_1}(\mathtt{v})))\dots )
\footnote{\ the order of the $y_j$ does not matter}
\end{equation}
where $Y=\{ y_1,\dots, y_r \}$. We let 
\[
u_{Y}:=\wh{\rho}_Y(\tH_Y):\textbf{m}\to \textbf{m}^{[Y]} 
\qquad \text{and} \qquad 
u_{r}:=\wh{\rho}_r(\tH_{[r]}):\textbf{m}\to \textbf{m}^{[[r]]}
. 
\]
Then the corresponding $(-)^{\textbf{E}}$-coalgebra is given by
\[
\widecheck{\rho}:\textbf{m} \to \textbf{m}^{\textbf{E}}
,\qquad 
\mathtt{v} \mapsto \underbrace{ u_{(-)}(\mathtt{v})= \sum_{r=0}^\infty u_r(\mathtt{v})}_{\text{via the isomorphism \textcolor{blue}{(\refeqq{eq:isoseriesr})}}}
.\] 
If $u:\textbf{m} \to \textbf{m}'$ is a commutative up operator, then we let $(\textbf{m},u)$ denote the $\textbf{E}$-module generated by $u$.  
%%%%%%%%%%%%%%%%%%%%%%%%%%%%%%%%%%%%%%%%%%%%%%%%%%%%%%%%%%%%%%%%%%%%%%%%
\subsection{$\textbf{E}$-Algebras}
%%%%%%%%%%%%%%%%%%%%%%%%%%%%%%%%%%%%%%%%%%%%%%%%%%%%%%%%%%%%%%%%%%%%%%%%
We now make $\textbf{E}$-algebras explicit. We present everything from the \hbox{$(-)^\textbf{E}$-coalgebra} point of view, since this shall be our perspective in the application to QFT. If $\textbf{m}=(\textbf{m},u)$ and $\textbf{n}=(\textbf{n},v)$ are $\textbf{E}$-modules, then a morphism $\textbf{m}\to \textbf{n}$ of $\textbf{E}$-modules is a commutative diagram in $[\sfS^{\op}, \sfVec]$ of the form
\begin{center}
\begin{tikzcd}[column sep=huge,row sep=large] 
\textbf{m} \arrow[d, "\eta"'] \arrow[r,   "u_{(-)}(-)"]   
& 
\textbf{m}^{\textbf{E}} \arrow[d,   "\eta^{\textbf{E}}"]     
\\
\textbf{n} \arrow[r,   "v_{(-)}(-)"']     
& 
\textbf{n}^{\textbf{E}}. 
\end{tikzcd}
\end{center}
Since $u$ and $v$ generate the $\textbf{E}$-actions, this diagram commutes if (and only if)
\begin{equation}\label{eq:commute}
v\big (\eta_I(\mathtt{v})\big ) =  \eta_{\ast I} \big (u(\mathtt{v}) \big) \qquad \text{for all} \quad \mathtt{v}\in \textbf{m}[I] \quad \text{and} \quad I\in \sfS 
.
\end{equation}
%(If $u$ and $v$ are the commutative up operators which define $\textbf{m}$ and $\textbf{n}$, then via the isomorphism \textcolor{blue}{(\refeqq{eq:ast})}, we have $u\cong \widecheck{\rho}$ and $v\cong \widecheck{\tau}$.) 
Given $\textbf{E}$-modules $(\textbf{m},u)$ and $(\textbf{n},v)$, the \emph{Cauchy product} $\textbf{m}\bigcdot \textbf{n}$ of $\textbf{E}$-modules is the $\textbf{E}$-module given by 
\[
\textbf{m}\bigcdot \textbf{n} \to  (\textbf{m}\bigcdot \textbf{n})^{\textbf{E}}
,\qquad 
\mathtt{v}\otimes \mathtt{w} \mapsto \sum_{Y_1 \sqcup Y_2 =(-)}  u_{Y_1}(\mathtt{v}) \otimes  v_{Y_2}(\mathtt{w})
=
\sum_{r=0}^\infty \sum_{r_1 + r_2 =r} \dfrac{r!}{r_1 ! \,  r_2 !}   u_{r_1}(\mathtt{v}) \otimes  v_{r_2}(\mathtt{w})
.\]
This is the $\textbf{E}$-module with generating commutative up operator given by
\begin{equation} \label{eq:gen}
\begin{tikzcd}[column sep=huge,row sep=large] 
\textbf{m}\bigcdot \textbf{n}  \arrow[r,   "u\, \bigcdot\, \text{id}\, \oplus\, \text{id}\, \bigcdot\, v    "]   
& 
(\textbf{m}\bigcdot \textbf{n})'.
\end{tikzcd}
\end{equation}
This defines the symmetric monoidal category $\textsf{Rep}(\textbf{E})$ of $\textbf{E}$-modules. The natural isomorphism $\textbf{1}\xrightarrow{\sim} \textbf{1}^{\textbf{E}}$ is the unit object. 

Continuing the $(-)^{\textbf{E}}$-coalgebra point of view, an $\textbf{E}$-algebra $\textbf{a}$ is an $\textbf{E}$-module $(\textbf{a},u)$ together with a pair of morphisms of $(-)^{\textbf{E}}$-coalgebras of the form
\begin{center}
\begin{tikzcd}[column sep=110pt,row sep=huge] 
\textbf{a}\bigcdot \textbf{a} 	\arrow[d, "\mu"']     \arrow[r,   "\sum_{Y_1 \sqcup Y_2 =(-)} u_{Y_1}(-) \otimes  u_{Y_2}(-)"]   & 
(\textbf{a}\bigcdot \textbf{a})^{\textbf{E}}  	\arrow[d,   "\mu^{\textbf{E}}"]     
\\
\textbf{a}	\arrow[r,   "u_{(-)}(-)"']     
& 
\textbf{a}^{\textbf{E}}
\end{tikzcd} \qquad $\text{and}$ \qquad 
\begin{tikzcd}[column sep=huge,row sep=huge]  
\textbf{1} 	\arrow[d, "\iota"']     \arrow[r,   "\cong"]   
& 
\textbf{1}^{\textbf{E}}  	\arrow[d,   "\iota^{\textbf{E}}"]     
\\
\textbf{a}	\arrow[r,   "u_{(-)}(-)"']     
& 
\textbf{a}^{\textbf{E}}
\end{tikzcd}
\end{center}
satisfying associativity and unitality. (The observation that $\textbf{a}\bigcdot \textbf{a}\to (\textbf{a}\bigcdot \textbf{a})^{\textbf{E}}$ is the composition of $\widecheck{\rho} \bigcdot \widecheck{\rho}$ with the multiplication transformation of $(-)^{\textbf{E}}$ gives a second proof of \autoref{prop:is a homo}.)

%Analogous objects in the graded vector space setting are differential graded algebras/cochain algebras. 

By drawing the relevant diagrams, one sees that associativity and unitality are satisfied if and only if $(\textbf{a},\mu,\iota)$ is an algebra in species. Thus, an $\textbf{E}$-algebra $(\textbf{a},u)$ is equivalently an algebra in species $(\textbf{a}, \mu, \iota)$ equipped with a commutative up operator $u$ such that these two diagrams commute. By \textcolor{blue}{(\refeqq{eq:commute})} and \textcolor{blue}{(\refeqq{eq:gen})}, the first diagram commutes if and only if $u$ is an up derivation of $\textbf{a}$, that is
\begin{equation}\label{eq:upder}
u\big(\mu_{S,T}(\mathtt{a} \otimes \mathtt{b})\big) 
=\mu_{\ast S,T}\big (u(\mathtt{a}) \otimes    \mathtt{b}\big)+\mu_{S,\ast T}\big(\mathtt{a} \otimes u(\mathtt{b})\big)
.
\end{equation}
The second diagram, which requires that $u(\mathtt{1}_\textbf{a})=0$, is then automatically satisfied since
\[
u(\mathtt{1}_\textbf{a})=
u\big(\mu_{\emptyset,\emptyset}(\mathtt{1}_\textbf{a} \otimes \mathtt{1}_\textbf{a} )\big) 
=
\mu_{\ast, \emptyset}\big (u(\mathtt{1}_\textbf{a}) \otimes    \mathtt{1}_\textbf{a}\big)+\mu_{\emptyset,\ast  }\big(\mathtt{1}_\textbf{a} \otimes u(\mathtt{1}_\textbf{a} )\big)
=u(\mathtt{1}_\textbf{a})+u(\mathtt{1}_\textbf{a})
.\] 
Notice that indeed, such an up operator $u$ will generate an $\textbf{E}$-action which commutes with the structure maps of $\textbf{a}$, and this recovers the uncurried action point of view on $\textbf{E}$-algebras.  
%%%%%%%%%%%%%%%%%%%%%%%%%%%%%%%%%%%%%%%%%%%%%%%%%%%%%%%%%%%%%%%%%%%%%%%%
\subsection{Hopf $\textbf{E}$-Algebras}
%%%%%%%%%%%%%%%%%%%%%%%%%%%%%%%%%%%%%%%%%%%%%%%%%%%%%%%%%%%%%%%%%%%%%%%%
We now make Hopf $\textbf{E}$-algebras explicit. An $\textbf{E}$-bialgebra $\textbf{h}=(\textbf{h},u)$ consists of four morphisms of $(-)^{\textbf{E}}$-coalgebras of the form
\begin{center}
\begin{tikzcd}[column sep=110pt,row sep=huge] 
\textbf{h}\bigcdot \textbf{h} 	\arrow[d, "\mu"']     \arrow[r,   "\sum_{Y_1 \sqcup Y_2 =(-)} u_{Y_1}(-) \otimes  u_{Y_2}(-)"]   & (\textbf{h}\bigcdot \textbf{h})^{\textbf{E}}  	\arrow[d,   "\mu^{\textbf{E}}"]     \\
\textbf{h}	\arrow[r,   "u_{(-)}(-)"']     & \textbf{h}^{\textbf{E}}
\end{tikzcd} \qquad 
\begin{tikzcd}[column sep=huge,row sep=huge] 
\textbf{1} 	\arrow[d, "\iota"']     \arrow[r,   "\cong"]   & \textbf{1}^{\textbf{E}}  	\arrow[d,   "\iota^{\textbf{E}}"]     \\
\textbf{h}	\arrow[r,   "u_{(-)}(-)"']     & \textbf{h}^{\textbf{E}}
\end{tikzcd} \\[15pt] 
\begin{tikzcd}[column sep=110pt,row sep=huge] 
\textbf{h} 	\arrow[d, "\Delta"']     \arrow[r,   "u_{(-)}(-)"]   &  \textbf{h}^{\textbf{E}}  	\arrow[d,   "\Delta^{\textbf{E}}"]     \\
\textbf{h}\bigcdot \textbf{h}	\arrow[r,   "\sum_{Y_1 \sqcup Y_2 =(-)} u_{Y_1}(-) \otimes  u_{Y_2}(-)"']     & (\textbf{h}\bigcdot \textbf{h})^{\textbf{E}}
\end{tikzcd} \qquad 
\begin{tikzcd}[column sep=huge,row sep=huge] 
\textbf{h} 	\arrow[d, "\epsilon"']     \arrow[r,   "u_{(-)}(-)"]   & \textbf{h}^{\textbf{E}}  	\arrow[d,   "\epsilon^{\textbf{E}}"]     \\
\textbf{1}	\arrow[r,   "\cong"']     & \textbf{1}^{\textbf{E}}
\end{tikzcd}
\end{center} \medskip
satisfying (co)associativity, (co)unitality, and the usual four compatibility conditions for bimonoids. This happens if and only if $(\textbf{h},\mu,,\Delta,\iota,\epsilon)$ is a bialgebra in species. Thus, an $\textbf{E}$-bialgebra is equivalently a bialgebra in species $(\textbf{h},\mu,,\Delta,\iota,\epsilon)$ equipped with a commutative up operator $u$ such that these four diagrams commute. The top two diagrams commute if and only if $u$ is an up derivation of $\textbf{a}$, and the lower left diagram commutes if and only if $u$ is an \emph{up coderivation} of $\textbf{a}$, meaning that
\begin{equation}\label{eq:upcoder}
\big(  
u\otimes \text{id} + \text{id} \otimes u
\big)
\circ  
\Delta_{S,T}(\mathtt{a})
=
\Delta_{\ast S,T} \big(u(\mathtt{a})\big)
+
\Delta_{S,\ast T} \big(u(\mathtt{a})\big)
.
\end{equation}
The lower right diagram is tautologous. Thus, if an \emph{up biderivation} is an up operator which is both an up derivation and an up coderivation, then an $\textbf{E}$-bialgebra $(\textbf{h},u)$ is equivalently a bialgebra in species $\textbf{h}$ equipped with a commutative up biderivation $u$ (which then generates an $\textbf{E}$-action commuting with the structure maps of $\textbf{h}$). 

A Hopf $\textbf{E}$-algebra $\textbf{h}=(\textbf{h},u)$ is an $\textbf{E}$-bialgebra such that there exists an endomorphism of the form
\begin{center}
\begin{tikzcd}[column sep=huge,row sep=huge] 
\textbf{h}	\arrow[d, "\text{s}"']     \arrow[r,   "u_{(-)}(-)"]   & \textbf{h}^{\textbf{E}}  	\arrow[d,   "\text{s}^{\textbf{E}}"]     \\
\textbf{h}	\arrow[r,   "u_{(-)}(-)"']     & \textbf{h}^{\textbf{E}}
\end{tikzcd}
\end{center}
satisfying the usual properties of an antipode.

\begin{prop}\label{prop:hopfalgebra}
Let $(\textbf{h},u)$ be an $\textbf{E}$-bialgebra. Then $(\textbf{h},u)$ is a Hopf $\textbf{E}$-algebra if and only if the underlying bialgebra $\textbf{h}$ is a Hopf algebra.
\end{prop}
\begin{proof}
If $(\textbf{h},u)$ is a Hopf $\textbf{E}$-algebra, then the underlying morphism $\text{s}:\textbf{h}\to \textbf{h}$ of the antipode of $(\textbf{h},u)$ will be an antipode for $\textbf{h}$ (this follows by drawing all the relevant diagrams as before). Conversely, if $\text{s}:\textbf{h}\to \textbf{h}$ is an antipode for $\textbf{h}$, then we require the diagram
\begin{center}
\begin{tikzcd}[column sep=huge,row sep=huge] 
\textbf{h}	\arrow[d, "\text{s}"']     \arrow[r,   "u_{(-)}(-)"]   & \textbf{h}^{\textbf{E}}  	\arrow[d,   "\text{s}^{\textbf{E}}"]     \\
\textbf{h}	\arrow[r,   "u_{(-)}(-)"']     & \textbf{h}^{\textbf{E}}
\end{tikzcd}
\end{center}
to commute. If $\textbf{h}$ is connected, then the antipode $\text{s}$ is given by Takeuchi's antipode formula \cite[Proposition 8.13]{aguiar2010monoidal},
\[
\text{s}_I=\sum_{F\in \Sigma[I]} (-1)^k \mu_F \circ \Delta_F
.\]
The sum is over all compositions $F$ of $I$ (defined in \autoref{comp}). Then
\[
u\big (\text{s}_I (\mathtt{v})\big )= \underbrace{u \Big( \sum_{F\in \Sigma[I]} (-1)^k \mu_F \circ \Delta_F (\mathtt{v}) \Big)= \sum_{F\in \Sigma[\ast I]} (-1)^k \mu_F \circ \Delta_F \big ( u( \mathtt{v}) \big )}_{\text{since $u$ is a biderivation}}
.\]
The generic (not necessarily connected) case follows in a similar way using the antipode formula \cite[Proposition 8.10 (ii)]{aguiar2010monoidal}. 
%since $\text{s}$ may be given by
%\[       
%\text{s}_I = \sum_{F}  (-1)^{ l(F) } \mu_{  \emptyset, S_1, \emptyset, \dots, \emptyset, S_k, \emptyset    } \big (\text{s}_0  \otimes \text{id}_{S_1} \otimes \text{s}_0 \otimes \cdots \otimes \text{s}_0 \otimes \text{id}_{S_k} \otimes %\text{s}_0 \big )  \Delta_{  \emptyset, S_1, \emptyset, \dots, \emptyset, S_k, \emptyset    } 
%.\]
%
%for $I$ nonempty \cite[Proposition 8.10]{aguiar2010monoidal}.
%\begin{align*}
%\text{s}_{\ast I} \circ u &=  
%\sum_{F}  (-1)^{ l(F) } \mu_{  \emptyset, S_1, \emptyset, \dots, \emptyset, S_k, \emptyset    } \big (\text{s}_0  \otimes \text{id}_{S_1} \otimes \text{s}_0 \otimes \cdots \otimes \text{s}_0 \otimes \text{id}_{S_k} \otimes \text{s}_0 \big )  %\Delta_{  \emptyset, S_1, \emptyset, \dots, \emptyset, S_k, \emptyset}  
%\\
%&=543534
%\end{align*}
\end{proof}
%%%%%%%%%%%%%%%%%%%%%%%%%%%%%%%%%%%%%%%%%%%%%%%%%%%%%%%%%%%%%%%%%%%%%%%%
\subsection{Lie $\textbf{E}$-Algebras}
%%%%%%%%%%%%%%%%%%%%%%%%%%%%%%%%%%%%%%%%%%%%%%%%%%%%%%%%%%%%%%%%%%%%%%%%
A Lie $\textbf{E}$-algebra $\textbf{g}=(\textbf{g},u)$ is a morphism of $(-)^{\textbf{E}}$-coalgebras of the form
\begin{center}
\begin{tikzcd}[column sep=110pt,row sep=huge] 
\textbf{g}\bigcdot \textbf{g} 	\arrow[d, "\partial^\ast"']     \arrow[r,   "\sum_{Y_1 \sqcup Y_2 =(-)} u_{Y_1}(-) \otimes  u_{Y_2}(-)"]   & (\textbf{g}\bigcdot \textbf{g})^{\textbf{E}}  	\arrow[d,   "{\partial^\ast}^{\textbf{E}}"]     \\
\textbf{g}	\arrow[r,   "u_{(-)}(-)"']     & \textbf{g}^{\textbf{E}}
\end{tikzcd}
\end{center}
satisfying antisymmetry and the Jacobi identity. This happens if and only if $(\textbf{g},\partial^\ast)$ is a Lie algebra in species. Thus, a Lie $\textbf{E}$-algebra is equivalently a Lie algebra in species $(\textbf{g},\partial^\ast)$ equipped with a commutative up operator $u$ such that this diagram commutes. This diagram commutes if and only $u$ is a \emph{up derivation} for the Lie bracket $\partial^\ast=[-,-]$, meaning that
\[
u\big([\mathtt{a} , \mathtt{b}]\big) 
=\big[u(\mathtt{a}) ,\mathtt{b}\big]
+\big[\mathtt{a} , u(\mathtt{b})\big]
\, .\] 

\begin{prop}\label{combracketup}
Let $(\textbf{a},u)$ be an $\textbf{E}$-algebra. Then the commutator bracket gives $\textbf{a}$ the structure of a Lie $\textbf{E}$-algebra. 
\end{prop}
\begin{proof}
We check that $u$ defines a derivation with respect to the commutator bracket. We have
\begin{align*}
u \big([\mathtt{a} , \mathtt{b}]\big) 
&=
u \big( \mu_{S,T}(\mathtt{a}\otimes \mathtt{b})
-
\mu_{T,S}(\mathtt{b}\otimes \mathtt{a})     \big)\\
&=
\mu_{\ast S, T}\big(u(\mathtt{a}) \otimes \mathtt{b}\big) + \mu_{S,\ast T}\big (\mathtt{a} \otimes    u(\mathtt{b})\big)
-
\mu_{\ast T, S}\big(u(\mathtt{b}) \otimes \mathtt{a}\big) - \mu_{T,\ast S}\big (\mathtt{b} \otimes    u(\mathtt{a})\big)\\
&=
\big [u(\mathtt{a}) , \mathtt{b}\big ]
+
\big[\mathtt{a} , u(\mathtt{b})\big ]. \qedhere
\end{align*}
\end{proof}

\begin{prop}\label{prop:primelalsoalg}
Let $(\textbf{h},u)$ be a Hopf $\textbf{E}$-algebra. Then the commutator bracket gives the primitive part Lie algebra $\mathcal{P}(\textbf{h})$ the structure of a Lie $\textbf{E}$-algebra. 
\end{prop}
\begin{proof}
Recall that
\[
\mathcal{P}(\textbf{h})[I]
=
\big\{ 
\mathtt{a}\in \textbf{h}[I] : 
\Delta_I(\mathtt{a})=  \mathtt{a}\otimes \mathtt{1}_\textbf{h} + \mathtt{1}_\textbf{h}\otimes \mathtt{a} 
\big\}
.\]
Since $u$ is an up coderivation, we have
\[
\Delta_{\ast S,T} \big(u(\mathtt{a}) \big)
=\big(  u\otimes \text{id} \big )  \circ  
\Delta_{S,T}(\mathtt{a})
.\]
If $\mathtt{a}$ is primitive, then this will certainly be equal to zero unless $S=\emptyset$ or $T=\emptyset$. For the case $S=\emptyset$, we still get zero since $u$ is a derivation,
\[ 
\Delta_{\ast ,I}\big(u(\mathtt{a}) \big)    
=\big( u \otimes \text{id} \big )  \circ  
\big (\Delta_{\emptyset,I}(\mathtt{a}) \big)
=   u(\mathtt{1}_\textbf{h}) \otimes \mathtt{a} 
  =0 
.\]
For the case $T=\emptyset$, we have
\[ 
\Delta_{\ast I,\emptyset}\big(u(\mathtt{a}) \big)    
=\big(  u\otimes \text{id} \big )  \circ  
\big (\Delta_{I,\emptyset}(\mathtt{a}) \big)
=     u(\mathtt{a}) \otimes   \mathtt{1}_\textbf{h}
.\]
We can argue similarly for $\Delta_{S , \ast T} \big(u(\mathtt{a}) \big)$. We obtain
\[
\Delta_{\ast I} \big (u(\mathtt{a}) \big ) 
= u(\mathtt{a})\otimes \mathtt{1}_\textbf{h} + \mathtt{1}_\textbf{h}\otimes u(\mathtt{a}) 
.\]
Therefore, if $\mathtt{a}\in \mathcal{P}(\textbf{h})[I]$, then $u(\mathtt{a})\in  \mathcal{P}(\textbf{h})[\ast I]$. The result then follows from \autoref{combracketup}.
\end{proof}

%Analogous objects are differential graded Lie algebras.  
%%%%%%%%%%%%%%%%%%%%%%%%%%%%%%%%%%%%%%%%%%%%%%%%%%%%%%%%%%%%%%%%%%%%%%%%
%%%%%%%%%%%%%%%%%%%%%%%%%%%%%%%%%%%%%%%%%%%%%%%%%%%%%%%%%%%%%%%%%%%%%%%%
%%%%%%%%%%%%%%%%%%%%%%%%%%%%%%%%%%%%%%%%%%%%%%%%%%%%%%%%%%%%%%%%%%%%%%%%
%%%%%%%%%%%%%%%%%%%%%%%%%%%%%%%%%%%%%%%%%%%%%%%%%%%%%%%%%%%%%%%%%%%%%%%%
\section{Products, Decorated Series, and Perturbation} 
%%%%%%%%%%%%%%%%%%%%%%%%%%%%%%%%%%%%%%%%%%%%%%%%%%%%%%%%%%%%%%%%%%%%%%%%
%%%%%%%%%%%%%%%%%%%%%%%%%%%%%%%%%%%%%%%%%%%%%%%%%%%%%%%%%%%%%%%%%%%%%%%%
%%%%%%%%%%%%%%%%%%%%%%%%%%%%%%%%%%%%%%%%%%%%%%%%%%%%%%%%%%%%%%%%%%%%%%%%
%%%%%%%%%%%%%%%%%%%%%%%%%%%%%%%%%%%%%%%%%%%%%%%%%%%%%%%%%%%%%%%%%%%%%%%%
\hyperlink{foo}{We} introduce the notion of a system of products for a species, which generalizes the realization of the abstract symmetric algebra $\cS(V)$ on a vector space $V$ as the algebra of polynomial functions on the dual vector space $V^\ast$. In the presence of an algebraic system of products, we define a homomorphism $\mathcal{S}_{(-)}$ which realizes series of species as formal power series valued functions on $V$. As special cases, this construction produces the classical exponential function of differential calculus (\autoref{ex:exp}), and the S-matrix of pAQFT (\autoref{sec:Time-Ordered Products}). Finally, we describe the perturbation of systems of algebraic products for $\textbf{E}$-algebras.  
%%%%%%%%%%%%%%%%%%%%%%%%%%%%%%%%%%%%%%%%%%%%%%%%%%%%%%%%%%%%%%%%%%%%%%%%
\subsection{Decorations}\label{sec:Dec}
%%%%%%%%%%%%%%%%%%%%%%%%%%%%%%%%%%%%%%%%%%%%%%%%%%%%%%%%%%%%%%%%%%%%%%%%
%The following construction appeared in \cite{barratt1978twisted}. See also \cite[Example 8.18]{aguiar2010monoidal}. 

Let $V$ be a vector space over $\Bbbk$. We usually denote vectors of $V$ by $\ssA$. Given a finite set $I\in \sfS$, we can form the tensor product $V^{\otimes I}$. The simple tensors of $V^{\otimes I}$ may be identified with functions $I\to V$, which we denote by
\[      
\ssA_I:I\to V, 
\qquad 
i \mapsto \ssA_i    
.\]
Letting $I=\{i_1,\dots, i_n\}$, the identification is then
\[
\ssA_I=\ssA_{i_1}\otimes \cdots \otimes \ssA_{i_n} \in V^{\otimes I}
.\] 
If $\ssA_i=\ssA$ for all $i\in I$, then we write    
\begin{equation}\label{eq:simpletensors} 
\ssA^{I}:= \underbrace{\ssA\otimes \cdots \otimes \ssA}_{\text{``$I$ times''}}\in V^{\otimes I}
\qquad \text{and} \qquad 
\ssA^{ n }:=\ssA^{[n]}\in V^{\otimes [n]}
.\footnote{\ recall $[n]:=\{1,\dots,n\}$} %^{\hyperlink{notation}{1}}
\end{equation}
Define the species $\textbf{E}_V$ by
\[         
\textbf{E}_V[I]:=  V^{\otimes I} 
\qquad \text{and} \qquad    
\textbf{E}_V[\sigma](\ssA_I)=\ssA_I\circ \sigma
.\]
In this context, $V$ is called the space of \emph{decorations}, following \cite[Chapter 19]{aguiar2010monoidal}. 

\begin{remark}
Species of the form $\textbf{E}_V$ are exactly the monoidal functors $\textsf{S}^{\text{op}}\to \textsf{Vec}$ (up to equivalence). 
\end{remark}

Given $S\sqcup T=I$, $\ssA_S:S\to V$ and $\ssA_T:T\to V$, define
\[
\ssA_S\ssA_T :I\to V, 
\qquad   
\ssA_S\ssA_T(i)
:=
\begin{cases}
\ssA_S(i) &\quad \text{if $i\in S$}\\
\ssA_T(i) &\quad \text{if $i\in T$.}
\end{cases}
\]   
Given $S\subseteq I$ and $\ssA_I:I\to V$, define
\[      
\ssA_{I|_S}:S\to V, 
\qquad   
\ssA_{I|_S}(i):= \ssA_I (i)
.\]
When there is no ambiguity, we often denote $\ssA_{I|_S}$ by $\ssA_S$. Then $\textbf{E}_V$ is a connected bialgebra in species, with multiplication and comultiplication given respectively by
\begin{equation}\label{eq:e_Valg}
\mu_{S,T} (  \ssA_S \otimes \ssA_T ):
=\ssA_S \ssA_T 
\qquad\text{and}\qquad 
\Delta_{S,T}(\ssA_I):
=\ssA_{I|_S} \otimes \ssA_{I|_T}            
.
\end{equation} 
Notice that $\mu$ and $\Delta$ are mutually inverse. The unit and counit are the canonical isomorphisms between the empty tensor product and $\Bbbk$. The antipode is given by 
\[\text{s}_I(\ssA_I)=(-1)^n \ssA_I.\] 
If we let our decorations be $V=\Bbbk$, then we recover the connected bialgebra $\textbf{E}_\Bbbk=\textbf{E}$. The bosonic Fock space $\overline{\cK}(\textbf{E}_V)$ of $\textbf{E}_V$ is the symmetric algebra $\cS(V)$ on $V$ \cite[Part III]{aguiar2010monoidal}.

We have the $V$-decorations endofunctor
\[
(-)\otimes\textbf{E}_{V}:[\sfS^{\op},\textsf{Vec}]\to [\sfS^{\op},\textsf{Vec}] 
,\qquad
\textbf{p}\mapsto \textbf{p}\otimes \textbf{E}_V  
.\]
This functor is braided bilax monoidal \cite[Section 8.13.4]{aguiar2010monoidal}. Regarding the lax monoidal structure, if $\textbf{a}$ is an algebra, then $\textbf{a} \otimes \textbf{E}_V$ is also an algebra with
\begin{equation}\label{eq:E_Valg}
\mu_{S,T}\big ( (\mathtt{a}\otimes  \ssA_S )\otimes (\mathtt{b}\otimes \ssA_T)\big )
:= 
\mu_{S,T}(\mathtt{a}\otimes \mathtt{b})\otimes \ssA_S \ssA_T
\qquad \text{and} \qquad
\mathtt{1}_{\textbf{a} \otimes \textbf{E}_V}:=\mathtt{1}_\textbf{a}\otimes 1_\Bbbk
.
\end{equation}
Given a choice of decorations vector $\ssS\in V$, we have the associated commutative up operator for $\textbf{E}_V$ given by
\begin{equation}\label{upop}
u:\textbf{E}_V \to \textbf{E}'_V
,\qquad
\ssA_I\mapsto u(\ssA_I):=\ssS\otimes \ssA_I= 
\underbrace{\ssS\otimes \ssA_{i_1}\otimes \cdots \otimes \ssA_{i_n}}_{\in \textbf{E}_V[\ast I]}
.
\end{equation}
This generates the $\textbf{E}$-module given by
\[
\textbf{E}\bigcdot \textbf{E}_V
\to 
\textbf{E}_V
,\qquad \tH_Y \otimes \ssA_I \mapsto \ssS^{\, Y} \otimes  \ssA_I
.\] 
Let us abbreviate
\begin{equation}\label{eq:Sintabb}
\ssS \ssA_I:=\ssS\otimes \ssA_I
.
\end{equation}
This up operator $u$ is a coderivation (defined in \textcolor{blue}{(\refeqq{eq:upcoder})}), since we have
\[
u(\ssA_S) \otimes \ssA_T +\ssA_S \otimes u(\ssA_T) 
=
\ssS \ssA_S \otimes \ssA_T +\ssA_S \otimes   \ssS \ssA_T
= 
\Delta_{\ast S, T}\big (u(\ssA_I)\big) + \Delta_{S,\ast T}\big (u(\ssA_I)\big)
.\]
Thus, for each choice of $\ssS\in V$, $\textbf{E}_V$ is an $\textbf{E}$-coalgebra. We check that for $\ssS\neq 0$, this is not also a derivation (defined in \textcolor{blue}{(\refeqq{eq:upder})}),
\[
u( \ssA_S \ssA_T ) = \ssS   \ssA_S \ssA_T
\neq
\ssS   \ssA_S \ssA_T + \ssS   \ssA_S \ssA_T
=
u(\ssA_S) \ssA_T + \ssA_S u(\ssA_T)
.\]
In fact, as we shall see, to expect this to be a derivation would be to expect
\[
e^{x+\delta}e^{y+\delta} \overset{\mathrm{!}}{=}  e^{x+y+\delta} 
, \qquad
x,y,\delta \in \bR, \ \delta\neq 0
.\]

%\begin{remark}
%In applications to QFT, it would seem that a more sophisticated approach would define $\textbf{E}_V$ as a cospecies, that is as copresheaf on $\sfS$, and have the decorations functor land in $\sfVec$. However, we leave this to feature work.
%\end{remark}
%%%%%%%%%%%%%%%%%%%%%%%%%%%%%%%%%%%%%%%%%%%%%%%%%%%%%%%%%%%%%%%%%%%%%%%%
\subsection{Systems of Products} \label{sec:Systems of Products}
%%%%%%%%%%%%%%%%%%%%%%%%%%%%%%%%%%%%%%%%%%%%%%%%%%%%%%%%%%%%%%%%%%%%%%%%
Let $V$ be a vector space over $\Bbbk$, let $\cA$ be a $\Bbbk$-algebra with multiplication denoted by $\star$, and let $\textbf{p}$ be a species. Let a \emph{system of products} for $\textbf{p}$ be a morphism of species of the form
\[   
\eta:\textbf{p} \otimes \textbf{E}_V  \to \textbf{U}_\cA
,\qquad    
\mathtt{v}\otimes \ssA_I
\mapsto  
\eta_I(\mathtt{v}\otimes\ssA_I)   
.
\]
We view this as a $\textbf{p}$-indexed multiplication of vectors of $V$ which lands in $\cA$, see e.g. \autoref{ex.e.g}. We can curry a system of products, to give the morphism of species
\[
\wh{\eta}:\textbf{p}\to \cH(\textbf{E}_V, \textbf{U}_\cA  )
,\qquad 
\mathtt{v}\mapsto\big (\ssA_I \mapsto \eta_I(\mathtt{v}\otimes\ssA_I)\big)  
.\]
The \emph{product} corresponding to $\mathtt{v}\in \textbf{p}[I]$ is the linear map 
\[
\wh{\eta}_I(\mathtt{v}): \textbf{E}_V[I]\to \cA
,\qquad
\ssA_I\mapsto \eta_I(\mathtt{v}\otimes\ssA_I)
.\] 
%In application to QFT, $V$ is a space of bump functions, and so $\wh{\eta}$ is a realization of the species $\textbf{p}$ as $\cA$-valued distributions. 
If $\textbf{p}=\textbf{a}$ is an algebra, then $\textbf{a} \otimes \textbf{E}_V$ is also an algebra by \textcolor{blue}{(\refeqq{eq:E_Valg})}. Let an \emph{algebraic system of products} be a homomorphism of algebras of the form 
\[   
\varphi:\textbf{a} \otimes \textbf{E}_V  \to \textbf{U}_\cA
,\qquad
\mathtt{v}\otimes \ssA_I
\mapsto  
\varphi_I(\mathtt{v}\otimes \ssA_I) 
.\]
%(We could also let $\textbf{p}$ be a Lie algebra in species, and require $\eta$ to be a homomorphism of Lie algebras.) 

\begin{ex} \label{ex.e.g}
If $\cA$ is a commutative $\Bbbk$-algebra and $V\subseteq \cA$ is a vector subspace, then we have the algebraic system of products given by
\[
\textbf{E}\otimes \textbf{E}_{V} \to \textbf{U}_{\cA}
,\qquad
\tH_I \otimes a_{i_1} \otimes \dots \otimes a_{i_n} \mapsto  a_{i_1} \star \dots \star  a_{i_n} 
.\]
Similarly, if $\cA$ is a noncommutative $\Bbbk$-algebra and $V\subseteq \cA$ is a vector subspace, then we have the algebraic system of products given by
\[
\textbf{L}\otimes \textbf{E}_{V} \to \textbf{U}_{\cA}
,\qquad
\tH_{\ell} \otimes a_{i_1} \otimes \dots \otimes a_{i_n} \mapsto  a_{\sigma(i_1)} \star \dots \star  a_{\sigma(i_n)} 
\]
where $\sigma:I\to I$ is the permutation given by $(\sigma(i_1), \dots, \sigma(i_n))=\ell$. 
\end{ex}

In applications to pAQFT, we are interested in cases where $\cA$ consists of functionals on the space of sections of a vector bundle $E\to \cX$, $V$ is a space of test functions which can be integrated against sections, and elements of $\textbf{p}$ eat test functions to produce functionals on sections in various interesting ways.

\begin{ex}\label{ex:polyfun}
A primordial, and very basic, example of such a system of commutative algebraic products is the representation of the abstract symmetric algebra $\cS(V)$ on a vector space $V$ as the algebra of polynomial functions on the dual vector space $V^\ast$. Let $\Bbbk=\bR$, let $\cX\in \textsf{Set}$ be a finite set of $k$ points (so integration will just be summation), and let $\text{Fib}\in \textsf{Vec}$ be a one-dimensional real vector space. Let $\text{Fib}^\ast$ denote its dual space, with the pairing denoted by 
\[
\la - , - \ra :\text{Fib}^\ast \otimes \text{Fib}\to \bR
.\] 
Let $E:\cX\times \text{Fib} \to \cX$ be the corresponding trivial line bundle, with dual bundle $E^\ast:\cX\times \text{Fib}^\ast \to \cX$. We have the vector spaces of sections
\[
\Gamma(E)
:=\big \{  \text{functions }  \Phi:\cX\to  \text{Fib}   \big  \} 
\qquad \text{and} \qquad   
\Gamma(E^\ast)=\Gamma'(E^\ast)
:=\big \{  \text{functions }  \ga:\cX\to  \text{Fib}^\ast   \big  \} 
.\] 
%(Since $\cX$ is finite, these coincide with the spaces of (compactly supported) distributional sections/generalized functions $\Gamma'_{\cX}(E)$ and $\Gamma'_{\cX}(E^\ast)$.) 
Then, we have the product for $\textbf{X}$ which realizes $\Gamma(E^\ast)$ as linear functionals on $\Gamma(E)$,
\[   
\eta:  \textbf{X} \otimes \textbf{E}_{\Gamma(E^\ast)}  \to
 \textbf{U}_{C^\infty(\Gamma(E))}
,\qquad    
\eta_{i}\big ( \tH_i \otimes \ga\big )(\Phi)  
:=
\sum_{x\in \cX} \big \la  \ga(x), \Phi(x)  \big \ra  
.\]
Let $V=\Gamma(E^\ast)$ and $V^\ast=\Gamma(E)$. Then, acting with $\textbf{E}\boldsymbol{\circ}(-)$ as described below, we obtain the pointwise product of polynomial functions, which is the algebraic system of products given by
\[
\varphi:\textbf{E}\otimes\textbf{E}_{V}\to  \textbf{U}_{C^\infty(V^\ast)}
,\qquad
\varphi_I( \tH_I\otimes \ga_I )(\Phi):= \sum_{(x_{i_j})\in \cX^I } \big \la  \ga_{i_1}(x_{i_1}), \Phi(x_{i_1})  \big \ra\cdots \big \la  \ga_{i_n}(x_{i_n}), \Phi(x_{i_n})  \big \ra
.\]
There are various ways to recover the $\bR$-algebra $\cS(V)$ from this, see \cite[Example 19.4]{aguiar2010monoidal}. For the associated construction of the exponential function, see \autoref{ex:exp}.
\end{ex}

Recall that $\cH(\textbf{E}_V, \textbf{U}_\cA  )$ is naturally an algebra in species \textcolor{blue}{(\refeqq{eq:curlyHalg})}, with multiplication and unit given by
\[                                   
\mu_{S,T}(f \otimes g)\,  (\ssA_I)=f(\ssA_{S})\star g(\ssA_{T})
\qquad
\text{and}
\qquad
\mathtt{1}_{ \cH(\textbf{E}_V, \textbf{U}_\cA  ) }=1_\cA 
.\] 

\begin{prop}\label{prop:curry}
Let $\textbf{a}$ be an algebra in species, and suppose that we have a system of products
\[   
\eta:\textbf{a} \otimes \textbf{E}_V  \to \textbf{U}_\cA  
.\]
Then $\eta$ is an algebraic system of products if and only if its currying 
\[
\wh{\eta}:\textbf{a}\to \cH(\textbf{E}_V, \textbf{U}_\cA  )
.\]
is a homomorphism.
\end{prop}
\begin{proof}
Directly from the definition of $\wh{\eta}$, we see that  
$\eta_\emptyset(\mathtt{1}_\textbf{a}\otimes 1_\Bbbk)=1_\cA$ 
if and only if 
$\wh{\eta}_\emptyset(\mathtt{1}_\textbf{a})(1_\Bbbk)= 1_\cA$. Therefore units are preserved. The morphism $\eta$ preserves the multiplication if
\[    
\eta_I\big(\mu_{S,T}(\mathtt{a}\otimes \mathtt{b}) \otimes \ssA_S \ssA_T \big)  
=\eta_S(\mathtt{a}\otimes \ssA_S) \star \eta_T(\mathtt{b}\otimes \ssA_T)
.\]
The currying $\wh{\eta}$ preserves the multiplication if
\[         
\wh{\eta}_I\big ( \mu_{S,T} (\mathtt{a}\otimes \mathtt{b} ) \big  )(\ssA_I)= \wh{\eta}_S(\mathtt{a})(\ssA_S)   \star   \wh{\eta}_T(\mathtt{b})(\ssA_T) 
.\]
Directly from the definition of $\wh{\eta}$, we see that the two left-hand sides agree, and the two right-hand sides agree. 
\end{proof}

Let $\textbf{p}_+$ be a positive species. The free algebra on $\textbf{p}_+$ is naturally the plethystic product $\textbf{L}\boldsymbol{\circ} \textbf{p}_+$, see \cite[Section 11.2]{aguiar2010monoidal}. Therefore, a system of products
\[
\textbf{p}_+\otimes \textbf{E}_V \to \textbf{U}_\cA
\]
uniquely extends to an algebraic system of products
\[
\textbf{L}\boldsymbol{\circ} (\textbf{p}_+\otimes \textbf{E}_V) \to \textbf{U}_\cA
.\]
We have
\begin{align*}
\textbf{L}\boldsymbol{\circ} (\textbf{p}_+ \otimes \textbf{E}_V)[I]
=&\ 
\bigoplus_{F\in \Sigma[I]} 
(\textbf{p}_+[S_1] \otimes \textbf{E}_V[S_1])
\otimes \dots \otimes
(\textbf{p}_+[S_k] \otimes \textbf{E}_V[S_k])\\
\cong&\ 
\bigoplus_{F\in \Sigma[I]} 
(\textbf{p}_+[S_1]
\otimes \dots \otimes
\textbf{p}_+[S_k]) \otimes \textbf{E}_V[I]\\
=&\ 
(\textbf{L}\boldsymbol{\circ} \textbf{p}_+) \otimes \textbf{E}_V[I]
.
\end{align*}
Therefore we can uniquely extend a system of products for a positive species $\textbf{p}_+$ to an algebraic system of products for the free algebra on $\textbf{p}_+$. If $\cA$ is commutative, then this will factor through abelianization $\textbf{L}\boldsymbol{\circ} \textbf{p}_+ \twoheadrightarrow \textbf{E}\boldsymbol{\circ} \textbf{p}_+$ to give an algebraic system of products for the free commutative algebra on $\textbf{p}_+$, which is naturally $\textbf{E}\boldsymbol{\circ} \textbf{p}_+$.

%%%%%%%%%%%%%%%%%%%%%%%%%%%%%%%%%%%%%%%%%%%%%%%%%%%%%%%%%%%%%%%%%%%%%%%%
\subsection{Decorated Series}\label{sec:decseries}
%%%%%%%%%%%%%%%%%%%%%%%%%%%%%%%%%%%%%%%%%%%%%%%%%%%%%%%%%%%%%%%%%%%%%%%%
Given a system of products for $\textbf{p}$, we now define a linear map $\mathcal{S}_{(-)}$ which realizes series of $\textbf{p}$ as $\cA[[\formj]]$-valued functions on $V$. We will see that if the system of products is algebraic, then $\mathcal{S}_{(-)}$ is a homomorphism. 

Given a system of products $\eta:\textbf{p} \otimes \textbf{E}_V  \to \textbf{U}_\cA$, we curry to give the morphism of species
\[
\wh{\eta}:\textbf{p}\to \cH(\textbf{E}_V, \textbf{U}_\cA  ), \qquad  \mathtt{v}\mapsto \big (\ssA_I \mapsto \eta_I(\mathtt{v}\otimes\ssA_I)\big )  
.\]
The image of this morphism under the series functor $\mathscr{S}$ is the linear map given by
\[     
\mathscr{S}(\wh{\eta}):\mathscr{S}(\textbf{p})\to \Hom(\textbf{E}_V,\textbf{U}_\cA)
,\qquad  
\mathtt{s} \mapsto 
\big(  
\ssA_I \mapsto  \eta_I(\mathtt{s}_I \otimes\ssA_I) 
\big)
.\]
Given vector spaces $V$ and $W$, let $\text{Func}(V,W)$ denote the vector space of all functions $V\to W$. Then, composing $\mathscr{S}(\wh{\eta})$ with the linear map
\[
\Hom(\textbf{E}_V, \textbf{U}_\cA) 
\to
\text{Func}\big (V,  \mathscr{S}( \textbf{U}_\cA)      \big )
\xrightarrow{\sim}
\text{Func}\big (V,   \cA[[\formj]]    \big )
,\qquad
\xi \mapsto  \bigg ( \ssA \mapsto  \sum_{n=0}^\infty \dfrac{\formj^n}{n!}\xi_{n}(\ssA^{n}) \bigg)
\footnote{\ recall our notation \textcolor{blue}{(\refeqq{eq:simpletensors})} \hypertarget{notation}{}}
\]
we obtain the linear map
\[
\mathcal{S}=\mathcal{S}_{(-)}: \mathscr{S}( \textbf{p})\to \text{Func}\big (V,\cA[[\formj]]\big )
,\qquad
\mathtt{s}\mapsto \mathcal{S}_\mathtt{s}
\]
where
\[                
\mathcal{S}_\mathtt{s}: V\to \cA[[\formj]], 
\qquad 
\ssA \mapsto \mathcal{S}_\mathtt{s}(\formj\! \ssA)
=
\sum^\infty_{n=0}  \dfrac{1}{n!} \eta_n \big(\mathtt{s}_n\otimes (\formj\! \ssA)^{n}\big)
:=
\sum^\infty_{n=0}  \dfrac{\formj^n}{n!} \eta_n (\mathtt{s}_n\otimes \ssA^n )
.\]
The vector space $\text{Func}\big (V,\cA[[\formj]]\big )$ is a $\Bbbk$-algebra under the pointwise multiplication of functions,
\[  
\text{Func}\big (V,\cA[[\formj]]\big )\otimes \text{Func}\big (V,\cA[[\formj]]\big ) \to \text{Func}\big (V,\cA[[\formj]]\big ), 
\qquad        
(\mathcal{U}\star \mathcal{V})(\formj\! \ssA):
= \mathcal{U}(\formj\! \ssA)  \star \mathcal{V}(\formj\! \ssA)           
.\]
The unit of $\text{Func}\big (V,\cA[[\formj]]\big )$ is the constant function with value $1_\cA$.

\begin{prop} \label{homo}
If $\varphi: \textbf{a}\otimes \textbf{E}_V \to \textbf{U}_{\cA}$ is an algebraic system of products, then the linear map
\[                        
\mathcal{S}: \mathscr{S}(\textbf{a})\to \text{Func}\big (V,\cA[[\formj]]\big ),
\qquad
\mathtt{s}\mapsto \mathcal{S}_{\mathtt{s}} 
\]
is a homomorphism of $\Bbbk$-algebras, that is
\[
\mathcal{S}_\mathtt{\mathtt{s}\ast \mathtt{t}}(\formj\! \ssA)     
=\mathcal{S}_\mathtt{s}(\formj\! \ssA)\star \mathcal{S}_\mathtt{t}(\formj\! \ssA )
\qquad \text{and} \qquad
\mathcal{S}_{1_{\mathscr{S}(\textbf{a})}}(\formj\! \ssA) = 1_\cA
.\] 
\end{prop}
\begin{proof}
To see that $\mathcal{S}$ preserves units, recalling the definition \textcolor{blue}{(\refeqq{eq:unitseries})} of $1_{\mathscr{S}(\textbf{a})}$, we have
\[ 
\mathcal{S}_{1_{\mathscr{S}(\textbf{a})}}(\formj\! \ssA)
=
\sum^\infty_{n=0}  \dfrac{\formj^n}{n!} \varphi_n (\mathtt{x}_n\otimes \ssA^n )
=
\dfrac{\formj^0}{0!} \varphi_{\emptyset}(\mathtt{x}_{\emptyset}\otimes \ssA^0)
=
\! \! \! \! \! \! 
\underbrace{\varphi_{\emptyset}(\mathtt{1}_\textbf{a})
=1_\cA}_{\text{$\varphi$ is a homomorphism}}
\! \! \! \! \! \!  
.\]
The currying $\wh{\varphi}:\textbf{a}\to\cH(\textbf{E}_V,\textbf{U}_{\cA})$ is a homomorphism by \autoref{prop:curry}. Then, since $\mathscr{S}$ is lax monoidal, the linear map
\[     
\mathscr{S}(\wh{\varphi}):\mathscr{S}( \textbf{a})\to \Hom(\textbf{E}_V, \textbf{U}_\cA)
,\qquad  
\mathtt{s} \mapsto \big (  \ssA_I \mapsto  \varphi_I(\mathtt{s}_I \otimes\ssA_I) \big)
\]
preserves the multiplication (it is also straightforward to explicitly check this). %Explicitly checking this, we have
%\[
%\varphi_{I}(\mathtt{s} \ast \mathtt{t}_{I}\otimes \ssA_I ) 
%= 
%\varphi_{I}\bigg(\sum_{S\sqcup T=I} \mu_{S,T}(\mathtt{s}_S\otimes  \mathtt{t}_T)  \otimes \ssA_I\bigg) 
%=
%\sum_{S\sqcup T=I} \varphi_{S}(\mathtt{s}_S\otimes \ssA_{I|_S})
%\star
%\varphi_{T}(\mathtt{s}_T\otimes \ssA_{I|_T})
%.\]
Therefore $\mathcal{S}$ preserves the multiplication if the linear map
\[
\Hom(\textbf{E}_V, \textbf{U}_\cA) 
\to
\text{Func}\big (V,   \cA[[\formj]]    \big )
,\qquad
\xi \mapsto  \bigg ( \ssA \mapsto \sum_{n=0}^\infty \dfrac{\formj^n}{n!} \xi_{n}(\ssA^{n}) \bigg   )
\]
preserves the multiplication. Indeed, for $\xi, \nu\in \Hom(\textbf{E}_V, \textbf{U}_\cA) $, the convolution product $\xi\ast \nu$ is given by \textcolor{blue}{(\refeqq{eq:convolprod})},
\[     
(\xi\ast \nu)_I (\ssA_I)
=\bigg(\sum_{S\sqcup T=I}
\mu_{S,T}  \circ  (  \xi_S\otimes \nu_T )  \circ    \Delta_{S,T} \bigg) (\ssA_I)
=
\sum_{S\sqcup T=I}
\xi_S(\ssA_S) \star \nu_T(\ssA_T)
.\]
Then
\begin{align*}  
\sum_{n=0}^\infty \dfrac{1}{n!} (\xi\ast \nu)_n (\ssA^n)
&= 
\sum_{n=0}^\infty  \dfrac{1}{n!}  \sum_{S\sqcup T=[n]}
\xi_S(\ssA^S) \star \nu_T(\ssA^T)\\[6pt]
&=\sum_{n=0}^\infty  \dfrac{1}{n!}  \sum_{s+t=n}  \dfrac{n!}{s! t!}
\xi_s(\ssA^s) \star \nu_t(\ssA^t)\\[6pt]
&=
\sum_{n=0}^\infty \dfrac{1}{n!} \xi_{n}(\ssA^{n}) \star \sum_{n=0}^\infty \dfrac{1}{n!} \nu_{n}(\ssA^{n})
. \qedhere
\end{align*}
\end{proof}

If $\textbf{a}$ and $\cA$ are commutative (or at least partly commutative), then we can restrict $\mathcal{S}$ to the set of exponential series $\mathscr{E}(\textbf{a})$ of $\textbf{a}$,
\[                        
\mathcal{S}: \mathscr{E}(\textbf{a})\to \text{Func}\big(V,\cA[[\formj]]\big)
,\qquad
\mathtt{e}\mapsto \mathcal{S}_{\mathtt{e}} 
.\]

\begin{prop}\label{prop:exp}
Let $\varphi: \textbf{a}\otimes \textbf{E}_V \to \textbf{U}_{\cA}$ be an algebraic system of products, and let $\mathtt{e}:\textbf{E}\to \textbf{a}$ be an exponential series of $\textbf{a}$. Then 
\[       
\mathcal{S}_\mathtt{e}(\formj\! \ssA +\formj\! \textsf{\emph{B}} )
=
\mathcal{S}_\mathtt{e}(\formj\! \ssA) \star \mathcal{S}_\mathtt{e}(\formj\! \textsf{\emph{B}} )
.\]
\end{prop} 
\begin{proof}
In this case, we have
\begin{equation} \label{eq:1}
\varphi_I( \mathtt{e}_I \otimes \ssA_I)
= 
\varphi_I\Big ( 
\mu_{S,T}
\big ( (\mathtt{e}_S \otimes \ssA_S )   \otimes (\mathtt{e}_T \otimes \ssA_T)\big)    
\Big)
=
\varphi_S( \mathtt{e}_S \otimes \ssA_S) \star \varphi_T( \mathtt{e}_T \otimes \ssA_T). 
\end{equation}
The first equality follows from the fact that $\mathtt{e}$ is exponential, and the second equality follows from the fact that $\varphi$ is a homomorphism. Then
\begin{align*} 
\mathcal{S}_\mathtt{e}(\formj\! \ssA +\formj\! \textsf{\emph{B}} )
=& 
\sum_{n=0}^\infty \dfrac{1}{n!}
\varphi_n\big( \mathtt{e}_n\otimes (\formj\! \ssA +\formj\! \textsf{\emph{B}})^{n}\big) \\[6pt]
=&
\sum_{n=0}^\infty \dfrac{1}{n!}  
\varphi_n\bigg( \mathtt{e}_n\otimes \sum_{s+t=n}\dfrac{n!}{s! t!} (\formj\! \ssA)^s \formj\! \textsf{\emph{B}})^{t} \bigg) \\[6pt]
=&
\sum_{n=0}^\infty \sum_{s+t=n}
\varphi_n\bigg( \mathtt{e}_n\otimes \dfrac{1}{s!} (\formj\! \ssA)^s \dfrac{1}{t!} (\formj\! \textsf{\emph{B}})^{t}      \bigg)\\[6pt]
=&
\sum_{n=0}^\infty \sum_{s+t=n}\dfrac{1}{s!}
\varphi_s\big( \mathtt{e}_s\otimes  (\formj\! \ssA)^s \big)   
\dfrac{1}{t!}
\varphi_t\big(\mathtt{e}_t\otimes (\formj\! \textsf{\emph{B}})^{t}\big) \\[6pt]
=&
\bigg(\sum_{n=0}^\infty \dfrac{1}{n!}
\varphi_n\big( \mathtt{e}_n\otimes (\formj\! \ssA)^{n}\big) \bigg) 
\star
\bigg(\sum_{n=0}^\infty \dfrac{1}{n!}
\varphi_n\big( \mathtt{e}_n\otimes (\formj\! \textsf{\emph{B}})^{n}\big)\bigg) \\[6pt]
=&
\, \mathcal{S}_\mathtt{e}(\formj\! \ssA) 
\star 
\mathcal{S}_\mathtt{e}(\formj\! \textsf{\emph{B}}).
\end{align*}
From the third line to the fourth line, we used \textcolor{blue}{(\refeqq{eq:1})}.
\end{proof}

If $\textbf{p}=\textbf{c}$ is a coalgebra, then we can restrict $\mathcal{S}$ to the set of group-like series $\mathscr{G}(\textbf{c})$ of $\textbf{c}$, giving
\[                        
\mathcal{S}: \mathscr{G}(\textbf{c})\to \text{Func}\big (V,\cA[[\formj]]\big ),
\qquad
\mathtt{g}\mapsto \mathcal{S}_{\mathtt{g}} 
.\] 
Let us recall some results of \cite[Section 12.3]{aguiar2013hopf}, adapted for the setting with decorations. If $\textbf{h}$ is a bialgebra in species, then $\mathscr{G}(\textbf{h})$ is closed under the convolution of series (this is a consequence of the compatibility condition between multiplication and comultiplication), and so we obtain a homomorphism of monoids,
\[                        
\mathcal{S}: \mathscr{G}(\textbf{h})\to \text{Func}\big (V,\cA[[\formj]]\big ),
\qquad
\mathtt{g}\mapsto \mathcal{S}_{\mathtt{g}} 
.\]
If $\textbf{h}$ is a Hopf algebra in species, then we obtain a homomorphism of groups
\[                        
\mathcal{S}: \mathscr{G}(\textbf{h})\to \text{Func}\big (V,\cA[[\formj]]\big )^\times,
\qquad
\mathtt{g}\mapsto \mathcal{S}_{\mathtt{g}} 
.\]
In particular (combining \textcolor{blue}{(\refeqq{interse})} and \autoref{homo}), we have
\begin{equation}\label{inverse2}
\mathcal{S}_\mathtt{g}(\formj\! \ssA)\star \mathcal{S}_{\text{s}\circ\mathtt{g}}(\formj\! \ssA)
=
\mathcal{S}_{\text{s}\circ\mathtt{g}}(\formj\! \ssA) \star \mathcal{S}_\mathtt{g}(\formj\! \ssA)
=
1_\cA
\end{equation}
for all $\ssA\in V$. 

\begin{ex}\label{ex:exp}
Recall the algebraic system of products from \autoref{ex:polyfun}, which is the pointwise product of polynomial functions,
\[
\varphi:\textbf{E}\otimes\textbf{E}_{V}
\to
\textbf{U}_{C^\infty(V^\ast)}
,\qquad
\tH_{I}\otimes \ssA_I \mapsto \varphi_I(\tH_{I}\otimes \ssA_I)
.\]
Then, exponential series (=group-like series) of $\textbf{E}$ are all of the form
\[
\mathtt{e}(c): \textbf{E}\to \textbf{E}, 
\qquad
\tH_I\mapsto \mathtt{e}(c)_I:=  c^n\, \tH_I
\qquad \quad \text{for}\quad 
c\in \bR
.\]
The associated formal power series valued function $\mathcal{S}_{\mathtt{e}(c)}$ is
\[
V\to  C^\infty(V^\ast)[[\formj]]
,\qquad
\ssA \mapsto \mathcal{S}_{\mathtt{e}(c)}(\formj\! \ssA)
=
\sum_{n=0}^\infty \dfrac{\formj^n c^n}{n!} \varphi_n (\tH_{[n]} \otimes \ssA^n)
.\]
We can set $\formj=1$, to obtain the classical exponential function on a real vector space
\[
\mathcal{S}_{\mathtt{e}(c)}(\ssA):V^\ast\to \bR
,\qquad  
\Phi\mapsto  \mathcal{S}_{\mathtt{e}(c)}(\ssA)(\Phi)=  e^{c\la  \ssA, \Phi \ra}
.\]
\end{ex}
%%%%%%%%%%%%%%%%%%%%%%%%%%%%%%%%%%%%%%%%%%%%%%%%%%%%%%%%%%%%%%%%%%%%%%%%
\subsection{Perturbation of Products}\label{sec:Perturbation of Products and Series}
%%%%%%%%%%%%%%%%%%%%%%%%%%%%%%%%%%%%%%%%%%%%%%%%%%%%%%%%%%%%%%%%%%%%%%%%
We now consider the special case of algebraic products for \hbox{$\textbf{E}$-algebras}, which, after a choice of decorations vector $\ssS\in V$, may then be `perturbed' using the $\textbf{E}$-action.

Let $(\textbf{a},u)$ be an $\textbf{E}$-algebra, equivalently an algebra in species $\textbf{a}$ equipped with a commutative up derivation $u:\textbf{a}\to \textbf{a}'$. Recall that we denote the $\textbf{E}$-action by
\[
\rho: \textbf{E}\bigcdot \textbf{a}\to \textbf{a}
,\qquad
\tH_Y\otimes \mathtt{a} \mapsto u_Y(\mathtt{a})
.\]
Suppose we have a system of products for $\textbf{a}$,
\[   
\eta:  \textbf{a}  \otimes \textbf{E}_V  
\to 
\textbf{U}_\cA
.\]
For a choice of decorations vector $\ssS\in V$, we have the $(-)^\textbf{E}$-coalgebra $\widecheck{\rho}_{u,\ssS}$ given by
\[        
\widecheck{\rho}_{u,\ssS}:
\textbf{a}\otimes \textbf{E}_V   \to  (\textbf{a}\otimes  \textbf{E}_V)^{\textbf{E}}
,\qquad
\mathtt{v}\otimes \ssA_I \mapsto 
\underbrace{u_{(-)}(\mathtt{v}) \otimes \ssS^{\, (-)} \ssA_I
=
\sum_{r=0}^\infty  u_r(\mathtt{v}) \otimes \ssS^{\, r} \ssA_I}_{\text{via the isomorphism \textcolor{blue}{(\refeqq{eq:isoseriesr})}}}
.\]
Then we obtain a new system of products $\eta_{u, \ssS}:=\eta^\textbf{E}\circ \widecheck{\rho}_{u,\ssS}$, now with target algebra $\cA[[\formg ]]$, 
\[  
\eta_{u, \ssS}: 
\textbf{a} \otimes \textbf{E}_V  
\xrightarrow{\widecheck{\rho}_{u,\ssS}}(\textbf{a} \otimes \textbf{E}_V)^{\textbf{E}} 
\xrightarrow{\eta^{\textbf{E}}} \underbrace{\textbf{U}^{\textbf{E}}_\cA \cong  \textbf{U}_{\cA[[\formg]]}}_{\text{\autoref{prop:g}}}
.\]
Explicitly, the image of $\mathtt{v}\otimes \ssA_I$ in $\cA[[\formg]]$ under $\eta_{u, \ssS}$ is 
\[
\sum^{\infty}_{r=0}  \eta^r_{I} \big (u_r (\mathtt{v}) \otimes \ssS^{\, r}  \ssA_I\big )
=
\eta_{I}(\mathtt{v}\otimes\ssA_I)
+\underbrace{\formg\,  \eta^1_{I}\big (u(\mathtt{v})\otimes \ssS \ssA_I\big) 
+\dfrac{\formg^2}{2!} \eta^2_{I}\big (u_2(\mathtt{v})\otimes \ssS^{\, 2} \ssA_I\big)+\,  \cdots}_{\text{perturbation}} 
\ .\]
We think of $\eta_{u, \ssS}$ as being a perturbation of $\eta$ by $\ssS$. If we put $\ssS=0$ (or set $\formg=0$) then $\eta_{u, \ssS}=\eta$. 

\begin{prop} \label{prop:isEalg}
Let $(\textbf{a},u)$ be an $\textbf{E}$-algebra. Given $\ssS\in V$, then the $(-)^\textbf{E}$-coalgebra
\[        
\widecheck{\rho}_{u,\ssS}:
\textbf{a}\otimes \textbf{E}_V   \to  (\textbf{a}\otimes  \textbf{E}_V)^{\textbf{E}}
,\qquad
\mathtt{a}\otimes \ssA_I \mapsto 
u_{(-)}(\mathtt{a}) \otimes \ssS^{\, (-)} \ssA_I
=
\sum_{r=0}^\infty  u_r(\mathtt{a}) \otimes \ssS^{\, r} \ssA_I
\]
is a homomorphism of algebras. 
\end{prop}
\begin{proof}
For preservation of units, we have
\[
\mathtt{1}_\textbf{a}\otimes 1_\Bbbk \mapsto  \sum_{r=0}^\infty  u_r(\mathtt{1}_\textbf{a}) \otimes \ssS^{\, r}=  u_0(\mathtt{1}_\textbf{a}) \otimes 1_\Bbbk =\mathtt{1}_\textbf{a} \otimes 1_\Bbbk 
.\] 
For preservation of multiplication, recall that the multiplication of $\textbf{a}\otimes \textbf{E}_V$ is given by
\[
\mu_{S,T}
\big(
(\mathtt{a}\otimes \ssA_S) \otimes (\mathtt{b}\otimes \ssA_T)
\big)
=
\mu_{S,T}(\mathtt{a}\otimes \mathtt{b}) \otimes \ssA_S \ssA_T
.\]
Then $\widecheck{\rho}_{u,\ssS}$ maps this to the series
\[
u_{(-)}\big(\mu_{S,T}(\mathtt{a}\otimes \mathtt{b})\big) \otimes \ssS^{\, (-)} \ssA_S \ssA_T \\[6pt]
=
\sum_{Y_1 \sqcup Y_2 =(-)} 
\mu_{Y_1\sqcup S,Y_2\sqcup T}\big(u_{Y_1}(\mathtt{a})\otimes u_{Y_2}(\mathtt{b})\big) 
\otimes
\ssS^{\, (-)} \ssA_S \ssA_T
.\]
The equality here, which is an instance of \textcolor{blue}{(\refeqq{eq:actionohh})}, follows from the fact that $u$ gives $\textbf{a}$ the structure of an $\textbf{E}$-algebra. On the other hand, the multiplication of $(\textbf{a}\otimes  \textbf{E}_V)^{\textbf{E}}$, which was defined in \textcolor{blue}{(\refeqq{eq:multofNderiv})}, is given by
\[
\mu_{S,T} \Big( \big ( 
u_{(-)}(\mathtt{a}) \otimes \ssS^{\, (-)} \ssA_S\big) \otimes \big(  
u_{(-)}(\mathtt{b}) \otimes \ssS^{\, (-)} \ssA_T\big)
\Big)
=
\sum_{Y_1 \sqcup Y_2 =(-)} 
\mu_{Y_1\sqcup S,Y_2\sqcup T}\big(u_{Y_1}(\mathtt{a})\otimes u_{Y_2}(\mathtt{b})\big) 
\otimes
\ssS^{\, (-)} \ssA_S \ssA_T 
.\]
The right-hand sides match, thus multiplication is also preserved. 
\end{proof}

\begin{cor}
Let $(\textbf{a},u)$ be an $\textbf{E}$-algebra. Given $\ssS\in V$, then $\textbf{a}\otimes \textbf{E}_V$ is an $\textbf{E}$-algebra where the $\textbf{E}$-action is given by
\[
\rho_{u,\ssS}:\textbf{E}\bigcdot (\textbf{a}\otimes \textbf{E}_V)\to \textbf{a}\otimes \textbf{E}_V
,\qquad
\tH_Y \otimes ( \mathtt{a}\otimes \ssA_I ) \mapsto u_{Y}(\mathtt{a}) \otimes \ssS^{\, Y} \ssA_I
.\]
\end{cor}
\begin{proof}
This is the result of combining \autoref{prop:isEalg} with \autoref{prop:is a homo}.
\end{proof}

\begin{cor}\label{prop:perturbedishomo}
Let $(\textbf{a},u)$ be an $\textbf{E}$-algebra, and let $\varphi:\textbf{a}\otimes \textbf{E}_V\to \textbf{U}_{\cA}$ be an algebraic system of products for $\textbf{a}$. Then, for $\ssS\in V$, the perturbed system of products
\[
\varphi_{u,\ssS}:
\textbf{a}\otimes \textbf{E}_V \to \textbf{U}_{\cA[[\formg]]}, \qquad \varphi_{u,\ssS}:=\varphi^\textbf{E}\circ \widecheck{\rho}_{u,\ssS}
\]
is again an algebraic system of products. 
\end{cor}
\begin{proof}
We have that $\varphi^{\textbf{E}}$ is a homomorphism because  $(-)^{\textbf{E}}$ is lax monoidal (it is also straightforward to explicitly check this). The result then follows from \autoref{prop:isEalg}, and the fact that a composition of homomorphisms is a homomorphism.
\end{proof}

For series $\mathtt{s}\in \mathscr{S}(\textbf{p})$, let $\mathcal{S}_{\mathtt{s},u, \formg\ssS}$ denote the corresponding formal power series valued function for the perturbed system of products $\eta_{u, \ssS}$, thus
\[  
\mathcal{S}_{\mathtt{s},u, \formg\ssS}:V\to \cA[[\formg,\! \formj]]
,\qquad
\ssA\mapsto \,
\mathcal{S}_{\mathtt{s},u, \formg\ssS}(\formj\! \ssA):
=
\sum^\infty_{n=0} \dfrac{\formj^n}{n!} {\eta_{u, \ssS}}_{n} (\mathtt{s}_n\otimes \ssA^n) 
.\]
Explicitly, this is
\[
\mathcal{S}_{\mathtt{s},u, \formg\ssS}(\formj\! \ssA)
= 
\sum^\infty_{n=0} \sum^\infty_{r=0} 
\dfrac{1}{r!\, n!} 
\eta^{r}_{n} 
\big(
u_r(\mathtt{s}_n)\otimes (\formg\ssS)^{r} (\formj\! \ssA)^n 
\big) 
:=
\sum^\infty_{n=0} \sum^\infty_{r=0} 
\dfrac{\formg^{r} \formj^n}{r!\, n!} 
\eta^{r}_{n} 
\big(
u_r(\mathtt{s}_n)\otimes \ssS^{\, r} \ssA^n 
\big)         
.\]
This defines the linear map
\[                        
\mathcal{S}_{(-),u,\formg\ssS}: 
\mathscr{S}(\textbf{p})\to \text{Func}\big (V,\cA[[\formg,\! \formj]]\big )
,\qquad 
\mathtt{s}\mapsto\,  \mathcal{S}_{\mathtt{s},u, \formg\ssS}
.\]

\begin{cor}
If $\eta=\varphi$ is an algebraic system of products for $\textbf{a}$ and $u$ is a commutative up derivation of $\textbf{a}$, then $\mathcal{S}_{(-),u,\formg\ssS}$ is a homomorphism of $\Bbbk$-algebras.
\end{cor}
\begin{proof}
This is the result of combining \autoref{prop:perturbedishomo} with \autoref{homo}.
\end{proof}
%%%%%%%%%%%%%%%%%%%%%%%%%%%%%%%%%%%%%%%%%%%%%%%%%%%%%%%%%%%%%%%%%%%%%%%%
%%%%%%%%%%%%%%%%%%%%%%%%%%%%%%%%%%%%%%%%%%%%%%%%%%%%%%%%%%%%%%%%%%%%%%%%
%%%%%%%%%%%%%%%%%%%%%%%%%%%%%%%%%%%%%%%%%%%%%%%%%%%%%%%%%%%%%%%%%%%%%%%%
%%%%%%%%%%%%%%%%%%%%%%%%%%%%%%%%%%%%%%%%%%%%%%%%%%%%%%%%%%%%%%%%%%%%%%%%
%%%%%%%%%%%%%%%%%%%%%%%%%%%%%%%%%%%%%%%%%%%%%%%%%%%%%%%%%%%%%%%%%%%%%%%%
\part{The Cocommutative Hopf Monoid of Set Compositions} 
%%%%%%%%%%%%%%%%%%%%%%%%%%%%%%%%%%%%%%%%%%%%%%%%%%%%%%%%%%%%%%%%%%%%%%%%
%%%%%%%%%%%%%%%%%%%%%%%%%%%%%%%%%%%%%%%%%%%%%%%%%%%%%%%%%%%%%%%%%%%%%%%%
%%%%%%%%%%%%%%%%%%%%%%%%%%%%%%%%%%%%%%%%%%%%%%%%%%%%%%%%%%%%%%%%%%%%%%%%
%%%%%%%%%%%%%%%%%%%%%%%%%%%%%%%%%%%%%%%%%%%%%%%%%%%%%%%%%%%%%%%%%%%%%%%%
%%%%%%%%%%%%%%%%%%%%%%%%%%%%%%%%%%%%%%%%%%%%%%%%%%%%%%%%%%%%%%%%%%%%%%%%
We now focus on the Hopf algebra of compositions $\Sig$, together with its Lie algebra of primitive elements $\Zie$. Our main task shall be the description of $\Sig$ and $\Zie$ as $\textbf{E}$-gebras. This is a \hbox{species-theoretic} formalization of structure discovered by Steinmann \cite{steinmann1960} and \hbox{Epstein-Glaser-Stora} \cite{epstein1976general}. This, combined with the construction of \autoref{sec:Perturbation of Products and Series}, will recover the perturbative construction of interacting fields in pAQFT, as in \cite[Section 8.1]{ep73roleofloc}, \cite[Section 6.2]{dutfred00}, going back to Bogoliubov \cite[Chapter 4]{Bogoliubov59}.

%%%%%%%%%%%%%%%%%%%%%%%%%%%%%%%%%%%%%%%%%%%%%%%%%%%%%%%%%%%%%%%%%%%%%%%%
\section{The Algebras} 
%%%%%%%%%%%%%%%%%%%%%%%%%%%%%%%%%%%%%%%%%%%%%%%%%%%%%%%%%%%%%%%%%%%%%%%%
  
%%%%%%%%%%%%%%%%%%%%%%%%%%%%%%%%%%%%%%%%%%%%%%%%%%%%%%%%%%%%%%%%%%%%%%%%
\subsection{Decompositions and Compositions}  \label{comp}
%%%%%%%%%%%%%%%%%%%%%%%%%%%%%%%%%%%%%%%%%%%%%%%%%%%%%%%%%%%%%%%%%%%%%%%%
A (set) \emph{decomposition} $F$ of $I$ is a function of the form
\[  
F:I \to      \{  1,\dots,k\}
,\qquad
\text{for some}\quad k\in \bN.
\] 
The ordering $1>\dots > k$ is understood. Let $\widehat{\Sigma}[I]$ denote the set of decompositions of $I$ ($\widehat{\Sigma}=\widehat{\Sigma}[-]$ naturally extends to a presheaf on $\sfS$). We encode decompositions $F\in \widehat{\Sigma}[I]$ in formal expressions
\[  
F=(S_1, \dots, S_k) 
,\qquad 
\text{where}\quad 
S_j:= F^{-1}(j)
.\]
The $S_j$ are called the \emph{lumps} of $F$. We let $l(F)=k$ denote the number of lumps of $F$. The \emph{opposite} of $F$ is defined by 
\[
\bar{F}:=(S_k,\dots, S_1), \qquad  \text{i.e.} \quad \bar{F}^{-1}(j)=F^{-1}(k+1-j)
.\] 
For $S\sqcup T=I$, $F\in \widehat{\Sigma}[S]$ and $G\in \widehat{\Sigma}[T]$, the \emph{concatenation} $FG\in \widehat{\Sigma}[I]$ of $F$ with $G$ is the decomposition obtained by concatenating formal expressions,
\[
\text{if}\quad F=(S_1, \dots, S_{l(F)})\quad \text{and} \quad G=(T_1, \dots, T_{l(G)}) \qquad
\text{then}\qquad 
FG:= (S_1, \dots, S_{l(F)} ,T_1, \dots, T_{l(G)})
.\]
For $S\subseteq I$ and $F:I\to \{ 1,\dots, k \}$, the \emph{restriction} $F|_S:S\to \{1,\dots, k \}$ of $F$ to $S$ is the decomposition obtained by precomposing $F$ with $S \hookrightarrow I$, thus
\[
\text{if}\quad F=(S_1, \dots, S_k)
\qquad \text{then}\qquad 
F|_S:= (S_1\cap S, \dots, S_k \cap S)
.\]
A (set) \emph{composition} $F$ of $I$ is a surjective decomposition of $I$, i.e. all the lumps are nonempty. Let $\Sigma[I]$ denote the set of compositions of $I$ ($\Sigma=\Sigma[-]$ also naturally extends to a presheaf on $\sfS$). %A \emph{linear order} $\ell$ of $I$ is a bijective decomposition of $I$, i.e. all the lumps are singletons. Let $\text{L}[I]$ denote the set of linear orders of $I$. 

Given a decomposition $F$, let $F_+$ denote the composition whose formal expression is obtained from the formal expression of $F$ by deleting all copies of the empty set. This defines retractions
\[     
(-)_+:  \widehat{\Sigma}[I] \twoheadrightarrow \Sigma[I]
,\qquad 
F\mapsto F_+
.\]
Note that compositions are not closed under restriction. If $F$ is a composition, we now make the convention that 
\[
F|_S:=(F|_S)_+
\, .\]
For compositions $F,G\in \Sigma[I]$, we write 
\[G\leq F\]
if $G$ is obtained from $F$ by merging contiguous lumps of $F$. Given compositions $G\leq F$ with $G=(T_1, \dots, T_k)$, we let
\[
l(F/G):=\prod^{k}_{j=1} l( F|_{T_j} )     
\qquad \text{and} \qquad  
(F/G)!:=\prod^{k}_{j=1} l( F|_{T_j} )!\,    
.\]
%%%%%%%%%%%%%%%%%%%%%%%%%%%%%%%%%%%%%%%%%%%%%%%%%%%%%%%%%%%%%%%%%%%%%%%%
\subsection{The Cocommutative Hopf Monoid of Compositions} \label{hopfofsetcomp}
%%%%%%%%%%%%%%%%%%%%%%%%%%%%%%%%%%%%%%%%%%%%%%%%%%%%%%%%%%%%%%%%%%%%%%%%
Let $\Sig$ be the (vector) species which is the composition of the presheaf $\Sigma$ (a `set species') with the free vector space functor $\textsf{Set}\to \sfVec$, thus 
\[   
\Sig[I] 
:=\big\{\text{formal $\Bbbk$-linear combinations of compositions of $I$}\big\}
.\]
For $F$ a composition of $I$, let $\tH_F\in \Sig[I]$ denote the basis element corresponding to $F$. The sets $\{\tH_F: F\in \Sigma[I]\}$ form the \emph{$\tH$-basis} of $\Sig$. Following \cite[Section 11]{aguiar2013hopf}, $\Sig$ is a connected bialgebra, with multiplication and comultiplication given in terms of the $\tH$-basis by
\[
\mu_{S,T}(\tH_F\otimes \tH_G):=\tH_{FG} \qquad \text{and} \qquad   \Delta_{S,T}  (\tH_F) :=  \tH_{F|_S} \otimes    \tH_{F|_T}
.\]
We sometimes abbreviate $\tH_F \tH_G:= \mu_{S,T}(\tH_F \otimes\tH_G)$. The unit and counit are given by 
\[
\mathtt{1}_{\Sig}=\tH_\emptyset\qquad  \text{and} \qquad \epsilon_\emptyset(\tH_\emptyset)=1_\Bbbk
.\] 
Let
\begin{equation}\label{antipode}
\overline{\tH}_F:= \sum_{G\geq \bar{F}}  (-1)^{l(G)}\, \tH_G 
.
\end{equation}
Then \cite[Theorem 11.38]{aguiar2010monoidal} (in the case $\textbf{q}=\textbf{E}^\ast_+$ and $q=1$) shows that
\begin{equation}\label{eq:inversion relation for reverse time-ordered products}
\sum_{S\sqcup T=I} \tH_{F|_S} \overline{\tH}_{F|_T}
=0
\qquad \text{and} \qquad 
\sum_{S\sqcup T=I}\overline{\tH}_{F|_S}  \tH_{F|_T}
=0
.
\end{equation}
Therefore the antipode of $\Sig$ is given by 
\[
\text{s}_I(\tH_F)=\overline{\tH}_F
.\] 
The Hopf algebra $\Sig$ is the free cocommutative Hopf algebra on the positive coalgebra $\textbf{E}^\ast_+$ \cite[Section 11.2.5]{aguiar2010monoidal}, and so $\Sig\cong \textbf{L}\boldsymbol{\circ} \textbf{E}^\ast_+$.

There is a second important basis of $\Sig$, called the \emph{$\tQ$-basis}. The $\tQ$-basis is also indexed by compositions, and is given by  
%\begin{equation}\label{eq:1}
\[ 
\tQ_F:= \sum_{G\geq F}   (-1)^{ l(G)-l(F) }  \dfrac{1}{    l(G/F) }   \tH_G\qquad \text{or equivalently} \qquad  \mathtt{H}_F=: \sum_{G\geq F} \dfrac{1}{( G/F )!}  \tQ_G 
.\]
%\end{equation}
For $S\subseteq I$ and $F\in \Sigma[I]$, we have \emph{deshuffling}
\[F\res_S\, :=  
\begin{cases}
F|_S &\quad  \text{if $S$ is a union of lumps of $F$}\footnote{\ }\\
0\in \Sig[S] &\quad \text{otherwise.}
\end{cases}
\]
\footnotetext{ not necessarily contiguous}The multiplication and comultiplication of $\Sig$ is given in terms of the $\tQ$-basis by  
\[         
\mu_{S,T} ( \tQ_F\otimes \tQ_G ) =  \tQ_{FG} 
\qquad \text{and} \qquad  
\Delta_{S,T}  (\tQ_F) =  \tQ_{F\res_S} \otimes \,   \tQ_{F\res_T}
.\]

%The \emph{adjoint representation} of $\Sig$ is the $\Sig$-module given by
%\[
%\text{Ad}:\Sig \bigcdot \Sig \to \Sig 
%,\qquad
%\text{Ad}_{Y,I}(\tH_G \otimes \tH_{F})=   \sum_{Y_1\sqcup Y_2=Y} \tH_{G|_{Y_1}} \tH_F \,  \overline{\tH}_{G|_{Y_2}}
%.\]
%One can check that this gives $\Sig$ the structure of a Hopf $\Sig$-algebra (everything goes through in the same way as the classical case).
%%%%%%%%%%%%%%%%%%%%%%%%%%%%%%%%%%%%%%%%%%%%%%%%%%%%%%%%%%%%%%%%%%%%%%%%
\subsection{Decorations}
%%%%%%%%%%%%%%%%%%%%%%%%%%%%%%%%%%%%%%%%%%%%%%%%%%%%%%%%%%%%%%%%%%%%%%%%
Given a vector space $V$, the species of $V$-decorated compositions $\Sig\otimes \textbf{E}_V$ is a connected bialgebra via the bilax structure of the endofunctor $(-)\otimes \textbf{E}_V$, with multiplication given by
\[
\mu_{S,T}\big((\tH_F\otimes\ssA_S) \otimes (\tH_G \otimes \ssA_T)\big)
:=
\tH_F \tH_G \otimes \ssA_S \ssA_T 
\]
and comultiplication given by 
\[
\Delta_{S,T}(\tH_F\otimes \ssA_I)
:=
(\tH_{F|_S}\otimes\ssA_{I|_S})\otimes(\tH_{F|_T} \otimes \ssA_{I|_T}) 
.\]
The unit and counit are given by 
\[
\mathtt{1}_{\Sig\otimes \textbf{E}_V}=\tH_\emptyset\otimes 1_\Bbbk\qquad  \text{and}\qquad \epsilon_\emptyset(\tH_\emptyset\otimes 1_\Bbbk)=1_\Bbbk
.\] 
For $\tH_F\otimes \ssA_I\in \Sig\otimes \textbf{E}_V[I]$, we have
\[
\sum_{S\sqcup T=I} 
\mu_{S,T}\big ((\tH_{F|_S}\otimes \ssA_{I|_S})\otimes (\overline{\tH}_{F|_T}\otimes \ssA_{I|_T})\big) 
=
\underbrace{\sum_{S\sqcup T=I} \tH_{F|_S}  \overline{\tH}_{F|_T}}_{\text{$=0$ by \textcolor{blue}{(\refeqq{eq:inversion relation for reverse time-ordered products})}}} \otimes \ssA_{I}
=0
\] 
and
\[
\sum_{S\sqcup T=I} 
\mu_{S,T}\big ((\overline{\tH}_{F|_S}\otimes \ssA_{I|_S})\otimes (\tH_{F|_T}\otimes \ssA_{I|_T})\big) 
=
\underbrace{\sum_{S\sqcup T=I} \overline{\tH}_{F|_S} \tH_{F|_T}}_{\text{$=0$ by \textcolor{blue}{(\refeqq{eq:inversion relation for reverse time-ordered products})}}}\otimes \ssA_{I}
=0
.\]
Therefore the antipode of $\Sig\otimes \textbf{E}_V$ is given by
\begin{equation}\label{eq:antipode}
\text{s}_I(\tH_F\otimes \ssA_I)= \overline{\tH}_F \otimes \ssA_I.  
\end{equation}
%%%%%%%%%%%%%%%%%%%%%%%%%%%%%%%%%%%%%%%%%%%%%%%%%%%%%%%%%%%%%%%%%%%%%%%%
\subsection{The Steinmann Algebra}\label{sec:Steain}
%%%%%%%%%%%%%%%%%%%%%%%%%%%%%%%%%%%%%%%%%%%%%%%%%%%%%%%%%%%%%%%%%%%%%%%% 
The Hopf algebra $\Sig$ is connected and cocommutative, and so \autoref{CMM} (CMM) applies. We now describe the positive Lie algebra of primitive elements $\mathcal{P}(\Sig)\subset \Sig$. 
%We may take two geometric perspectives on this Lie algebra; the first is associated with the braid arrangement and type $A$ Coxeter complex, and the second is associated with the adjoint braid arrangement and Steinmann sphere.

For $I\in \sfS$ a finite set, let a \emph{tree} $\mathcal{T}$ over $I$ be a planar\footnote{\ i.e. a choice of left and right child is made at every node} full binary tree whose leaves are labeled bijectively with the blocks of a partition of $I$ (a \emph{partition} $P$ of $I$ is a set of disjoint nonempty subsets of $I$, called \emph{blocks}, whose union is $I$). The blocks of this partition, called the \emph{lumps} of $\mathcal{T}$, form a composition called the \emph{debracketing} $F_\mathcal{T}$ of $\mathcal{T}$, by listing them in order of appearance from left to right. We denote trees by nested products $[\, \cdot\, ,\, \cdot\, ]$ of subsets or trees, see \autoref{fig:tree}. We make the convention that no trees exist over the empty set $\emptyset$. 

\begin{figure}[H]
\centering
\includegraphics[scale=0.6]{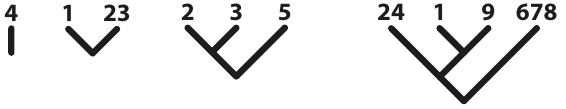}	\caption{Let $I=\{1,2,3,4,5,6,7,8,9\}$. The trees $[4]$, $[1,23]$ ($\neq[23,1]$), $[[2,3],5]$, $[[24,[1,9]],678]$. The debracketing of $[[24,[1,9]],678]$ is the composition $(24,1,9,678)$. If we put $\mathcal{T}_1=[24,[1,9]]$ and $\mathcal{T}_2=[678]$, then $[\mathcal{T}_1, \mathcal{T}_2]$ would also denote this tree.}
\label{fig:tree}
\end{figure} 

We define the positive species $\textbf{Zie}$ by letting $\textbf{Zie}[I]$ denote the vector space of formal $\Bbbk$-linear combinations of trees over $I$, modulo the relations of antisymmetry and the Jacobi identity as interpreted on trees in the usual way. Explicitly, 
\begin{enumerate}
\item(antisymmetry) for all compositions of $I$ with two lumps $(S,T)$, and all trees $\mathcal{T}_1\in \Zie[S]$ and $\mathcal{T}_2\in \Zie[T]$, we have
\[
[\mathcal{T}_1, \mathcal{T}_2 ]
+[ \mathcal{T}_2, \mathcal{T}_1 ]=
0
\]
\item(Jacobi Identity) for all compositions of $I$ with three lumps $(S,T,U)$, and all trees $\mathcal{T}_1\in \Zie[S]$, $\mathcal{T}_2\in \Zie[T]$ and $\mathcal{T}_3\in \Zie[U]$, we have
\[
[[\mathcal{T}_1,\mathcal{T}_2],\mathcal{T}_3]+
[[\mathcal{T}_3,\mathcal{T}_1],\mathcal{T}_2]+
[[\mathcal{T}_2,\mathcal{T}_3],\mathcal{T}_1]=
0
.\]
\end{enumerate}
Then $\Zie$ is a positive Lie algebra in species, with Lie bracket given by
\[     
\partial^\ast :  \Zie\bigcdot \Zie\to \Zie, 
\qquad
\partial_{S,T}^\ast(\mathcal{T}_1\otimes \mathcal{T}_2):=[\mathcal{T}_1,\mathcal{T}_2] 
.\]
Note that $\Zie$ is freely generated by the series consisting of `stick trees',
\[
\textbf{E}_+^\ast \to \Zie , \qquad   \tH_{(I)} \mapsto [I]
.\]
Indeed, $\Zie$ is the free Lie algebra on the positive exponential species $\textbf{E}^\ast_+$, and so the species $\Zie$ is also given by
\[   
\Zie[I]
=   
\Lie \boldsymbol{\circ} \bE^\ast_+[I]= \bigoplus_{P}   \Lie[P]
\]
where $\textbf{Lie}$ is the species of the Lie operad, and the direct sum is over all partitions $P$ of $I$. 

The Lie algebra in species $\Zie$ is closely related to the Steinmann algebra from the physics literature. Precisely, the Steinmann algebra is an ordinary graded Lie algebra based on the structure map for the adjoint braid arrangement realization of $\Zie$, see \cite[Section III.1]{bros}, \cite[Section 6]{Ruelle}. The adjoint braid arrangement realization of $\Zie$ is the topic of \cite{lno2019}, and the fact that the Lie algebra there is indeed $\Zie$ was shown in \cite{norledge2019hopf}. 

Via the commutator bracket, $\Sig$ is a Lie algebra in species, given by
\[
\Sig \bigcdot \Sig \to \Sig,
\qquad
[\tH_F,\tH_G] =\tH_F \tH_G -\tH_G \tH_F
.\]
Recall the Lie subalgebra of primitive elements $\mathcal{P}(\Sig)\subset \Sig$, defined in \textcolor{blue}{(\refeqq{eq:prim})}. Since $\Sig$ is connected, this is a positive Lie algebra, given on nonempty $I$ by
\[
\mathcal{P}(\Sig)[I]
=\bigcap_{(S,T)}  \text{ker} 
\big(    
\Delta_{S,T} :\Sig[I]\to \Sig[S]\otimes \Sig[T]    
\big)
.\]
In particular, $\tQ_{(I)}\in \mathcal{P}(\Sig)[I]$ for $I$ nonempty. Since $\Zie$ is freely generated by stick trees, we can define a homomorphism of Lie algebras by 
\[\Zie\to \mathcal{P}(\Sig), \qquad [I]\mapsto \tQ_{(I)}.\] 
To describe this explicitly, given a tree $\mathcal{T}$, let $\text{antisym}(\mathcal{T})$ denote the set of $2^{l(F_{\mathcal{T}})-1}$ many trees which are obtained by switching left and right branches at nodes of $\mathcal{T}$. For $\mathcal{T}' \in \text{antisym}(\mathcal{T})$, let $(\mathcal{T}, \mathcal{T}')\in \bZ/2\bZ$ denote the parity of the number of node switches required to bring $\mathcal{T}$ to $\mathcal{T}'$. Then the homomorphism is given in full by
\[      
\textbf{Zie} \to \mathcal{P}(\Sig), \qquad \mathcal{T} \mapsto   \tQ_\mathcal{T}  := \sum_{\mathcal{T}' \in \text{antisym}(\mathcal{T})}  (-1)^{ (\mathcal{T},\mathcal{T}') }  \tQ_{F_{\mathcal{T}'}}
.\]
By \cite[Corollary 11.46]{aguiar2010monoidal}, this is an isomorphism. From now on, we make the identification
\[
\Zie= \mathcal{P}(\Sig)
\]
and retire the notation $\mathcal{P}(\Sig)$.

%\begin{ex}
%If $\mathcal{T}=[[1,23],4]$, then the trees $\mathcal{T}'\in \text{antisym}(\mathcal{T})$, written with prefactor $(-1)^{(\mathcal{T}, \mathcal{T}')}$, are 
%\[[[1,23],4],\qquad  -[[23,1],4],\qquad  -[4,[1,23]],\qquad [4,[23,1]].\]
%\end{ex}

%Recall that $\tE_I=\tQ_{(I)}$ for $I$ nonempty, and that the multiplication of the $\tQ$-basis of $\Sig$ is given by the concatenation of compositions. 

%%%%%%%%%%%%%%%%%%%%%%%%%%%%%%%%%%%%%%%%%%%%%%%%%%%%%%%%%%%%%%%%%%%%%%%%
\subsection{Dynkin Elements}\label{adjoint}
%%%%%%%%%%%%%%%%%%%%%%%%%%%%%%%%%%%%%%%%%%%%%%%%%%%%%%%%%%%%%%%%%%%%%%%% 

\begin{figure}[t]
	\centering
	\includegraphics[scale=0.7]{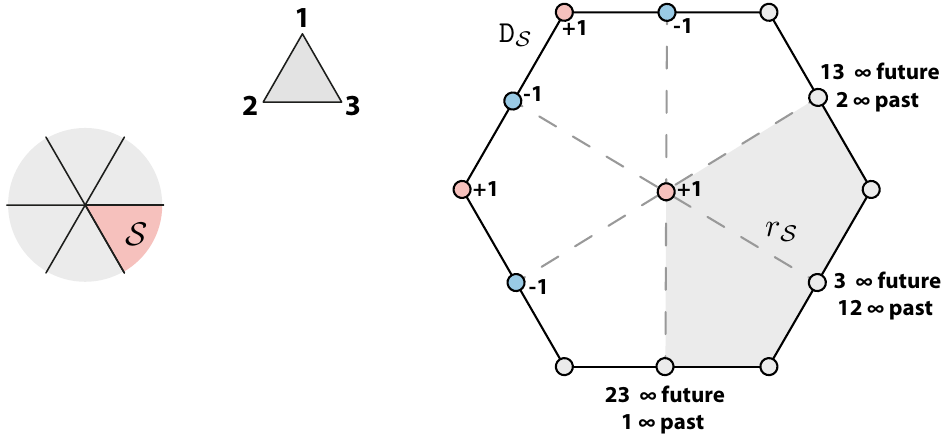}
	\caption{A cell $\cS$ over $\{1,2,3\}$ (on the adjoint braid arrangement) and its Dynkin element $\mathtt{D}_\cS$ (on the tropical geometric realization of $\boldsymbol{\Sigma}$, where the multiplication embeds facets and the comultiplication projects onto facets, see \cite[Introduction]{norledge2019hopf})). In the presence of causal factorization, the time component of the corresponding generalized retarded function $r_\cS$ is a $\bC[[\hbar, \formg]]$-valued generalized function on the braid arrangement with support the gray cone. The Dynkin element shown is $\mathtt{D}_\cS=\mathtt{D}_3=\mathtt{R}_{(12;3)}$. Its support consists of those configurations such that the event labeled by $3$ can be causally influenced by the events labeled by $1$ and $2$.}
	\label{fig:supp}
\end{figure}

Let 
\[
[I; \text{2}]:= \big\{    \text{surjective functions}\ I\to \{1,2 \}   \big\} 
\]
denote the set of compositions of $I$ with two lumps. Recall that the set of minuscule weights of (the root datum of) $\text{SL}_I(\bC)$ is in natural bijection with $[I; \text{2}]$. We denote the minuscule weight corresponding to $(S,T)$ by $\lambda_{ST}$. See \cite[Section 3.1]{norledge2019hopf} for more details.

A \emph{cell}\footnote{\ also known as maximal unbalanced families \cite{billera2012maximal} and positive sum systems \cite{MR3467341}} \cite[Definition 6]{epstein1976general} over $I$ is (equivalent to) a subset $\cS\subseteq [I; \text{2}]$ such that for all $(S,T)\in [I; \text{2}]$, exactly one of 
\[
(S,T)\in \cS \qquad \text{and} \qquad (T,S)\in \cS
\] 
is true, and whose corresponding set of minuscule weights is closed under conical combinations, that is
\[
\lambda_{UV}\in \text{coni}\big \la  \lambda_{ST} : (S,T)\in \cS \big \ra 
\quad \implies \quad 
(U,V)\in \cS
.\]
By dualizing conical spaces generated by minuscule weights, cells are in natural bijection with chambers of the adjoint of the braid arrangement, see \cite[Section 3.3]{norledge2019hopf}, \cite[Definition 2.5]{epstein2016}. Their number is sequence \href{https://oeis.org/A034997}{A034997} in the OEIS. We denote the species of formal $\Bbbk$-linear combinations of cells by $\textbf{L}^\vee$.

Associated to each composition $F$ of $I$ is the subset $\cF_F\subseteq [I; \text{2}]$ consisting of those compositions $(S,T)$ which are obtained by merging contiguous lumps of $F$, 
\[
\cF_F:=\big \{  (S,T)\in [I; \text{2}] : (S,T) \leq F\big \} 
.\]
More geometrically, $\cF_F$ is the subset corresponding to the set of minuscule weights which are contained in the closed braid arrangement face of $F$. Let us write $F\subseteq \cS$ to indicate that $\cF_F \subseteq \cS$.  

Consider the morphism of species given by
\begin{equation}     \label{eq:hbasisexp}
\textbf{L}^\vee\to \Sig
,\qquad 
\cS\mapsto  \mathtt{D}_\cS
:= 
-\sum_{\bar{F}\subseteq \cS} (-1)^{l(F)} \tH_{F}
.
\end{equation} 
The element $\mathtt{D}_\cS$ is called the \emph{Dynkin element} associated to the cell $\cS$. These special elements were defined by Epstein-Glaser-Stora in \cite[Equation 1, p.26]{epstein1976general}, and the name is due to \hbox{Aguiar-Mahajan} \cite[Equation 14.1]{aguiar2017topics} (see \autoref{Rem:dny}). In fact, $\mathtt{D}_\cS$ is a primitive element \cite[Proposition 14.1]{aguiar2017topics}, and so we actually have a morphism $\textbf{L}^\vee\to \Zie$.

For $i\in I$, let $\cS_i$ denote the cell given by
\[     \cS_i:=\big \{ (S,T)\in [I, \text{2}]: i\in S     \big \}  .    \]
This is the cell corresponding to the adjoint braid arrangement chamber which contains the projection of the basis element $e_i\in \bR I$ onto the sum-zero hyperplane. Let the \emph{total retarded} Dynkin element $\mathtt{D}_i$ associated to $i$ be given by
\[   
\mathtt{D}_i:= \mathtt{D}_{\cS_i}  =-\sum_{\substack{F\in \Sigma[I]\\ i\in S_k}} (-1)^{l(F)} \tH_F   
.\]
These Dynkin elements are considered in \cite[Section 14.5]{aguiar2013hopf}. For $i\in I$, let
\[     
\bar{\cS}_i
:=
\big \{ 
(S,T)\in [I, \text{2}]: i\in T     
\big \}  
.\]
This is the cell corresponding to the adjoint braid arrangement chamber which is opposite to the chamber of $\cS_i$. Let the \emph{total advanced} Dynkin element $\mathtt{D}_{\bar{i}}$ associated to $i$ be given by
\[   
\mathtt{D}_{\bar{i}}:= \mathtt{D}_{\bar{\cS}_i}=-\sum_{\substack{F\in \Sigma[I]\\ i\in S_1}} (-1)^{l(F)} \tH_F   
.\]

\begin{remark}\label{Rem:dny}
More generally, Dynkin elements are certain Zie elements of generic real hyperplane arrangements, which are indexed by chambers of the corresponding adjoint arrangement. They were introduced by Aguiar-Mahajan in \cite[Equation 14.1]{aguiar2017topics}. Specializing to the braid arrangement, one recovers the type $A$ Dynkin elements $\mathtt{D}_\cS$. 
\end{remark}

In \cite{norledge2019hopf},  the following perspective on the Dynkin elements is given. The Hopf algebra $\Sig^\ast$ which is dual to $\Sig$ is realized as an algebra $\hat{\Sig}^\ast$ of piecewise-constant functions on the braid arrangement. Then its dual, in the sense of polyhedral algebras \cite[Theorem 2.7]{MR1731815}, is an algebra $\check{\Sig}^\ast$ of certain functionals of piecewise-constant functions on the adjoint braid arrangement, i.e. those coming from evaluating on permutohedral cones. We have the morphism of species
\[   
\check{\Sig}^\ast\to  \cH(\textbf{L}^\vee, \textbf{E})      
\]
defined by sending functionals to their restrictions to piecewise-constant functions on the complement of the hyperplanes. Since the multiplication of $\check{\Sig}^\ast$ corresponds to embedding hyperplanes, this morphism is the indecomposable quotient of $\check{\Sig}^\ast$ \cite[Theorem 4.5]{norledge2019hopf}. Then, in \cite[Proposition 5.1]{norledge2019hopf}, we see that taking the linear dual of this morphism recovers the Dynkin elements map,
\[        
 \textbf{L}^\vee\to \Sig, \qquad \cS\mapsto  \mathtt{D}_\cS
.\] 
(Here we have identified $\Sig^\ast=  \check{\Sig}^\ast$.) Therefore we obtain the following.

\begin{thm}[$\! \! ${\cite{norledge2019hopf}}]
The morphism of species $\textbf{L}^\vee\to \Zie$ is surjective. Therefore the Dynkin elements $\{\mathtt{D}_\cS: \cS\ \text{is a cell over $I$} \}$ span $\Zie$.
\end{thm} 
%%%%%%%%%%%%%%%%%%%%%%%%%%%%%%%%%%%%%%%%%%%%%%%%%%%%%%%%%%%%%%%%%%%%%%%%
\subsection{The Steinmann Relations}\label{stein}
%%%%%%%%%%%%%%%%%%%%%%%%%%%%%%%%%%%%%%%%%%%%%%%%%%%%%%%%%%%%%%%%%%%%%%%% 
The Dynkin elements are not linearly independent. The relations which are satisfied by the Dynkin elements are generated by relations known in physics as the Steinmann relations, introduced in \cite{steinmann1960zusammenhang}, \cite{steinmann1960}. 
%We now abstract out the essential feature of the above example to describe the Steinmann relations in general. 

Let a pair of \emph{overlapping channels} over $I$ be a pair $(S,T),(U,V)\in [I; \text{2}]$ of two-lump compositions of $I$ such that 
\[     
S\cap U\neq \emptyset \qquad \text{and}  \qquad      T\cap U \neq \emptyset   
.\]
Let $\cS_1$, $\cS_2$, $\cS_3$, $\cS_4$ be four cells over $I$ with $(S,T),(U,V)\in \cS_1$, and such that $\cS_2$, $\cS_3$, $\cS_4$ are obtained from $\cS_1$ by replacing, respectively, 
\[        (S,T), (U,V)   \mapsto  (T,S), (U,V)         \]
\[        (S,T), (U,V)   \mapsto  (T,S), (V,U)         \]
\[        (S,T), (U,V)   \mapsto  (S,T), (V,U).        \]
Then, by inspecting the definition of the Dynkin elements \textcolor{blue}{(\refeqq{eq:hbasisexp})}, we see that\footnote{\ we go through the argument for the basic $4$-point case in \autoref{ex:stein}, which is sufficient to exhibit the general phenomenon}
\[    
\underbrace{\mathtt{D}_{\cS_1} -\mathtt{D}_{\cS_2}    +\mathtt{D}_{\cS_3}    -\mathtt{D}_{\cS_4}    =0.}_{\text{Steinmann relation}}         
\]
In general, a \emph{Steinmann relation} is any relation between Dynkin elements obtained in this way, i.e. an alternating sum of four Dynkin elements which are obtained from each other by switching overlapping channels only. This definition of the Steinmann relations can be found in \cite[Seciton 4.3]{epstein1976general} (it is given slightly more generally there for paracells). 

An alternative characterization of the Steinmann relations in terms of the Lie cobracket of the dual Lie coalgebra $\Zie^\ast$ is \cite[Definition 4.2]{lno2019}. Here, the Steinmann relations appear in the same way one can arrive at generalized permutohedra, i.e. by insisting on type $A$ `factorization' in the sense of species-theoretic coalgebra structure. See \cite[Theorem 4.2 and Remark 4.2]{norledge2019hopf}.

%The Steinmann algebra $\Zie$ appears most explicitly in the work of Ruelle \cite[Section 6]{Ruelle}. However the approach there uses certain `cycles', and the connection is explained in \cite[Section 4.4]{epstein1976general}.

So, the Dynkin elements satisfy the Steinmann relations. Moreover, we have the following.

\begin{thm}
The relations which are satisfied by the Dynkin elements are generated by the Steinmann relations. That is, if 
\[
\textbf{Stein}[I]
:=
\big \la
\mathtt{D}_{\cS_1}-\mathtt{D}_{\cS_2}+\mathtt{D}_{\cS_3}-\mathtt{D}_{\cS_4}
:
\mathtt{D}_{\cS_1}-\mathtt{D}_{\cS_2}+\mathtt{D}_{\cS_3}-\mathtt{D}_{\cS_4}=0 \text{ is a Steinmann relation}
\big \ra\footnote{\ angled brackets denote $\Bbbk$-linear span} 
\] 
then
\[
\Zie\cong \bigslant{\textbf{L}^\vee}{\textbf{Stein}} 
.\] 
\end{thm}
\begin{proof}
This follows by combining \cite[Theorem 4.3]{lno2019} with \cite[Theorems 4.2 and 4.5]{norledge2019hopf}. 
\end{proof}

\begin{ex}\label{ex:stein}
Let us give the basic $4$-point example $I=\{1,2,3,4\}$, which takes place on a square facet of the type $A$ coroot solid \cite[Figure 1]{lno2019}. Consider the following four cells over $I$ (we have marked where they differ, the names `$s$-channel' and `$u$-channel' are from physics and refer to Mandelstam variables),
\[   
\cS_1=\big\{\underbrace{(23,14)}_{u\text{-channel}}, (12,34), (1,234), (13,24), (13,24), (134,2), (3,124)  \}    
\]
\[ 
\cS_2=\big\{(23,14), \underbrace{(34,12)}_{s\text{-channel}}, (1,234), (13,24), (13,24), (134,2), (3,124)  \big \} 
\] 
\[ 
\cS_3=\big\{\underbrace{(14,23)}_{u\text{-channel}}, (34,12), (1,234), (13,24), (13,24), (134,2), (3,124)  \big \} 
\]
\[ 
\cS_4=\big\{(14,23), \underbrace{(12,34)}_{s\text{-channel}}, (1,234), (13,24), (13,24), (134,2), (3,124)  \big \}
.\] 
The $s$-channel and the $u$-channel overlap, and so we should now have
\[    
\mathtt{D}_{\cS_1}-\mathtt{D}_{\cS_2}+\mathtt{D}_{\cS_3}-\mathtt{D}_{\cS_4}=0 
.\]    
To see this, let us assume throughout that $\tH_{F}$ appears in the $\tH$-basis expansion \textcolor{blue}{(\refeqq{eq:hbasisexp})} of $\mathtt{D}_{\cS_1}$, i.e. $\bar{F} \subseteq \cS_1$. Then we have
\begin{equation*}
\bar{F} \subseteq \cS_1 \setminus \{  (12,34), (23,14)  \}   \quad \implies \quad      \bar{F}\subseteq\cS_1, \   \cS_2,\ \cS_3,\ \cS_4.  \tag{$\spadesuit$} \label{1}
\end{equation*}
If $\bar{F} \nsubseteq \cS_1 \setminus \{  (12,34), (23,14)  \}$, then either $(12,34)\in \bar{F}$ or $(23,14)\in \bar{F}$ but not both, since the channels overlap. We then have
\begin{equation*}
(12,34)\in \bar{F}   \implies    \bar{F}\subseteq\cS_1, \ \bar{F}\nsubseteq\cS_2, \ \bar{F}\nsubseteq\cS_3, \   \bar{F}\subseteq\cS_4.    \tag{$\varheart$}  \label{2}
\end{equation*}
We also have
\begin{equation*}
(23,14)\in \bar{F} \implies      \bar{F}\subseteq\cS_1, \ \bar{F}\subseteq\cS_2, \ \bar{F}\nsubseteq\cS_3, \   \bar{F}\nsubseteq\cS_4   .              \tag{$\vardiamond$}  \label{3}
\end{equation*}
Notice that in all three cases  \textcolor{blue}{(\refeqq{1})}, \textcolor{blue}{(\refeqq{2})}, \textcolor{blue}{(\refeqq{3})}, the prefactors of $\tH_{F}$ sum to zero in the four term alternating sum of the Steinmann relation. 
\end{ex}

\begin{remark}
Ocneanu \cite{oc17} and Early \cite{early2019planar} have studied an affine version of the Steinmann condition, in the context of higher structures and matroid subdivisions. Here, one observes that the (translated) hyperplanes of the adjoint braid arrangement for the Mandelstam variables give three subdivisions of the hypersimplex $\Delta(2,4)$ (octahedron). 
\begin{figure}[H]
\centering
\includegraphics[scale=0.5]{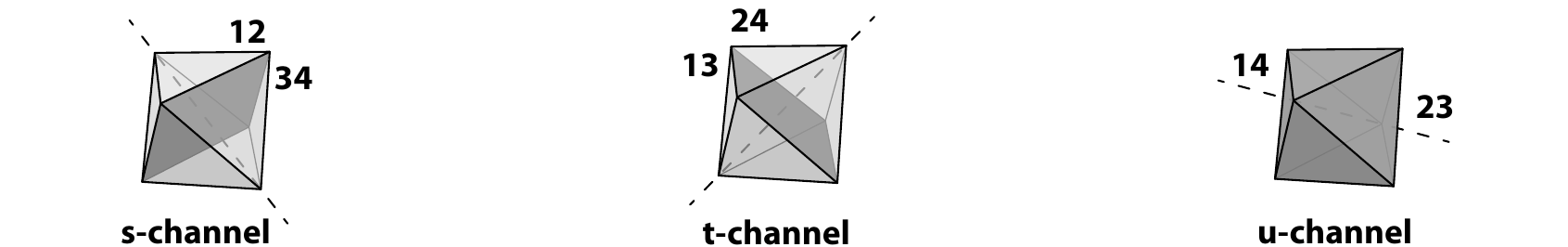}	
\label{fig:affine}
\end{figure}
\noindent See \cite{borges2019generalized}, \cite{cachazo2019planar} for the closely related study of generalized Feynman diagrams in generalized biadjoint $\Phi^3$-theory.
\end{remark}

%%%%%%%%%%%%%%%%%%%%%%%%%%%%%%%%%%%%%%%%%%%%%%%%%%%%%%%%%%%%%%%%%%%%%%%%
\subsection{The Tits Algebra}\label{sec:Tits}
%%%%%%%%%%%%%%%%%%%%%%%%%%%%%%%%%%%%%%%%%%%%%%%%%%%%%%%%%%%%%%%%%%%%%%%% 
The Hopf algebra $\Sig$ is also a monoid for the Hadamard product, called the Tits algebra (also called the (opposite of the) algebra of proper sequences in \cite{epstein1976general}), given by
\[     
\Sig\otimes\Sig \to \Sig
,\qquad   
\tH_F\otimes \tH_G \mapsto  \tH_F\triangleright \tH_G :=   \mu_F \big (  \Delta_F(  \tH_G  ) \big )   
.\]
Explicitly, given compositions $F=(S_1, \dots , S_{k_F})$ and $G=(T_1,\dots, T_{k_G})$ of $I$, we have
\[  
\mu_F \big (  \Delta_F(  \tH_G  ) \big )=\tH_{(  T_1\cap S_1, \dots, T_{k_G}\cap S_1, \dots \dots, T_{1}\cap S_{k_F}, \dots ,   T_{k_G}\cap S_{k_F}        )_+}     
.\]
The product $\tH_F\triangleright \tH_G$ is called the \emph{Tits product}, going back to Tits \cite{Tits74}. Its unit is $\tH_{(I)}$. See \cite[Section 13]{aguiar2013hopf} for more on the structure of the Tits product. See also \cite[Section 1.4.6]{brown08} for the context of other Coxeter systems and Dynkin types. 

%\begin{remark}
%The motivation for considering the Tits product in causal perturbation theory comes from considering causal factorization. The Tits product may be interpreted as the projection of the tropical permutohedron onto its facets \cite[Introduction]{norledge2019hopf}, which corresponds to the factorization of scattering amplitudes on the facets of the permutohedron, see \autoref{fig:respects} and \autoref{prop:scattering}. 
%\end{remark}

Let us recall \cite[Theorem 82 and Proposition 88]{aguiar2013hopf}, which describe the special role of $\Sig$. If $\textbf{h}$ is a connected cocommutative Hopf algebra, then we have the (Tits algebra) $\Sig$-module $\textbf{h}=(\textbf{h},\Psi)$ where the action is taking Hopf powers,
\[
\Psi:\Sig \otimes \textbf{h} \to \textbf{h}
,\qquad
\tH_F \otimes \mathtt{h}\mapsto   \tH_F \triangleright \mathtt{h}:=   \mu_F(  \Delta_F(\mathtt{h}))
.\] 
The curried action
\[
\wh{\Psi}:\Sig \to  \CMcal{E}(\textbf{h})
\]
is a homomorphism of monoids for the Hadamard product. It is also a homomorphism of algebras (i.e. monoids for the Cauchy product), where $\CMcal{E}(\textbf{h})$ is equipped with the convolution algebra structure \textcolor{blue}{(\refeqq{eq:curlyHalg})}. In fact, if $\textbf{h}$ is finite-dimensional, then $\CMcal{E}(\textbf{h})$ is naturally a Hopf algebra, and $\wh{\Psi}$ is a homomorphism of Hopf algebras. 

\begin{prop}[$\! \! ${\cite[Lemma 8]{epstein1976general}}]
Given a cell $\cS=\big \{  (S_1,T_1),\dots, (S_k,T_k)  \big \}$ over $I$, we have  
\[      
\mathtt{D}_{\cS}=   \underbrace{(\tH_{(  I )}-   \tH_{(  T_1,S_1 )}) \triangleright\,   \cdots\,   \triangleright     (\tH_{(  I )}-   \tH_{(  T_k,S_k )} )}_{\text{these elements commute in the Tits product}}
.\]
\end{prop}

\begin{prop}[$\! \! ${\cite[Theorem 7]{epstein1976general}}] \label{actondynkinwithtits}
Let $\cS$ be a cell over $I$. Then for all $(S,T)\in \cS$, we have
\[      
\mathtt{D}_{\cS}\triangleright \tH_{(T,S)}      =0       
.\]
\end{prop}
%%%%%%%%%%%%%%%%%%%%%%%%%%%%%%%%%%%%%%%%%%%%%%%%%%%%%%%%%%%%%%%%%%%%%%%%
\subsection{Ruelle's Identity and the GLZ Relation} \label{sec:Ruelle's Identity and the GLZ Relation}
%%%%%%%%%%%%%%%%%%%%%%%%%%%%%%%%%%%%%%%%%%%%%%%%%%%%%%%%%%%%%%%%%%%%%%%% 
Since the Dynkin elements span $\Zie$, we can ask what is the description of the Lie bracket of $\Zie$ in terms of the Dynkin elements. The answer is known in the physics literature as Ruelle's identity.

In order to state Ruelle's identity, we need to notice the following. For $S\sqcup T=I$, if $\cS_1$ is a cell over $S$ and $\cS_2$ is a cell over $T$, then $\cS_1 \sqcup \cS_2$ describes a collection of codimension one faces of the adjoint braid arrangement which are supported by the hyperplane orthogonal to $\lambda_{ST}$ (in \cite{lno2019}, such faces were called \emph{Steinmann equivalent}). A cell $\cS^{[S,T]}$ over $I$ which satisfies
\[
\cS^{[S,T]} \supseteq \cS_1\sqcup\cS_2
\qquad \text{and} \qquad
(S,T)\in \cS^{[S,T]}
\]
corresponds to a chamber arrived at by moving (by an arbitrarily small amount) from an interior point of a face of $\cS_1 \sqcup \cS_2$ in the $\lambda_{ST}$ direction. In particular, such cells always exist, but they are not unique (the Steinmann relations exactly quotient out this ambiguity). The chamber obtained by moving in the opposite direction corresponds to the cell obtained by replacing $(S,T)$ with $(T,S)$ in $\cS^{[S,T]}$.

\begin{prop}[Ruelle's Identity {\cite[Equation 6.6]{Ruelle}}] \label{prop:ruelle}
For $S\sqcup T=I$, let $\cS_1$ be a cell over $S$ and let $\cS_2$ be a cell over $T$. Let $\cS^{[S,T]}$ be a cell over $I$ which satisfies 
\[
\cS^{[S,T]} \supseteq \cS_1\sqcup\cS_2
\qquad \text{and} \qquad
(S,T)\in \cS^{[S,T]}
.\] 
Let $\cS^{[T,S]}$ denote the cell obtained by replacing $(S,T)$ with $(T,S)$ in $\cS^{[S,T]}$. Then the Lie bracket of $\Zie$ is given by
\begin{equation}\label{eq:ruelleiden}
[\mathtt{D}_{\cS_1},\mathtt{D}_{\cS_2}]   
=   
\mathtt{D}_{\cS^{[S,T]}}-\mathtt{D}_{\cS^{[T,S]}}
.
\end{equation} 
\end{prop}
\begin{proof}
%Notice that $\cS$ is the adjoint family for a codimension one shard, and $\cS^{[S,T]}$ and $\cS^{[T,S]}$ are its two adjacent maximal shards.
This result is clear from \cite[Section 5.2]{lno2019}; the Lie bracket which was given to the adjoint braid arrangement realization of $\textbf{Z}\textbf{ie}$ (denoted there by $\Gam$) coincides with \textcolor{blue}{(\refeqq{eq:ruelleiden})}. Alternatively, we can just explicitly check, as in \cite[Section 4.3]{epstein1976general}. 
\end{proof}

If a system of interacting generalized time-ordered products is constructed as in \cite{steinbook71}, \cite{dutfredretard04}, \cite{dutsch2019perturbative} via total retarded products, then one includes the GLZ relation \cite[Equation 11]{GLZ1957}, \cite[Proposition 1.10.1]{dutsch2019perturbative}, which is a consequence of the following relation satisfied by the Dynkin elements $\mathtt{D}_i=\mathtt{R}_{(X;i)}$,
\[     
\mathtt{D}_{i_1}- \mathtt{D}_{i_2}=   \sum_{\substack{(S,T)\in [\text{2};I]\\[1pt] i_1\in S,\ i_2\in T }}[\mathtt{D}_{i_1}  , \mathtt{D}_{i_2}  ]_{S,T} 
.\]

%%%%%%%%%%%%%%%%%%%%%%%%%%%%%%%%%%%%%%%%%%%%%%%%%%%%%%%%%%%%%%%%%%%%%%%%
\section{$\Sig$ as a Hopf $\textbf{E}$-Algebra} 
%%%%%%%%%%%%%%%%%%%%%%%%%%%%%%%%%%%%%%%%%%%%%%%%%%%%%%%%%%%%%%%%%%%%%%%%
%Bogoliubov-Shirkov approach to perturbative QFT

We now recall the Steinmann arrows, which are (or we interpret as) actions of $\textbf{E}$ on $\Sig$. We show that they give $\Sig$ the structure of a Hopf $\textbf{E}$-algebra (=Hopf monoid internal to $\textbf{E}$-modules) in two ways, and thus the primitive part $\Zie=\mathcal{P}(\Sig)$ the structure of a Lie $\textbf{E}$-algebra in two ways. %As we shall see, the resulting perturbation of products on $\Sig$ recovers Bogoliubov's formula \cite[Chapter 4]{Bogoliubov59}. 
%%%%%%%%%%%%%%%%%%%%%%%%%%%%%%%%%%%%%%%%%%%%%%%%%%%%%%%%%%%%%%%%%%%%%%%%
\subsection{Derivations and Coderivations of $\Sig$} 
%%%%%%%%%%%%%%%%%%%%%%%%%%%%%%%%%%%%%%%%%%%%%%%%%%%%%%%%%%%%%%%%%%%%%%%%

Recall the definition of (not necessarily commutative) up (co)derivations from \autoref{sec:Modules in Species}, or see \cite[Section 8.12.4]{aguiar2010monoidal}. In particular, an up derivation $u$ of $\Sig$ is a morphism of species
\[
u:\Sig\to \Sig',
\qquad 
\tH_F\mapsto u(\tH_F) 
\qquad \quad \text{such that} \qquad 
u(\tH_F \tH_G )=
u(\tH_F)\tH_G+\tH_F u(\tH_G)
.\] 
An up derivation of $\Sig$ is determined by its values on the elements $\tH_{(I)}$, $I\in \sfS$, since then
\[
u(\tH_F)= u(\tH_{(S_1)})\tH_{(S_2)}\dots \tH_{(S_k)} 
+\ \,  \cdots\ \,  +
\tH_{(S_1)} \dots \tH_{(S_{k-1})}   u(\tH_{(S_k)}) 
.\] 
An up derivation must have $u(\tH_{\emptyset})=0$, since $\mathtt{1}_{\Sig}= \tH_{\emptyset}$. An up coderivation $u$ of $\Sig$ is a morphism of species
\[
u:\Sig\to \Sig', \qquad \tH_F\mapsto u(\tH_F) 
\qquad \quad \text{such that} \qquad
\Delta_{\ast S,T} \big(  u( \tH_F )\big)=
u(\tH_{F|_S}) \otimes \tH_{F|_T}
.\]
In particular, an up coderivation must have
\[
\Delta_{\ast S,T}\big(u( \tH_{(I)} )\big)=
u(\tH_{(S)}) \otimes \tH_{(T)}  
.\]
Therefore, a biderivation $u$ of $\Sig$ must have
\[
u(\tH_{(i)})=  a_1 \tH_{(\ast,i)}+a_2 \tH_{(\ast i)}+a_3 \tH_{(i, \ast)} 
\qquad \quad \text{where} \qquad
a_1+a_2+a_3=0\in \Bbbk
.\]
Motivated by this, given $a,b\in \Bbbk$, we define an up derivation $u_{a,b}$ of $\Sig$ by
\begin{equation}\label{eq:defbider}
u_{a,b}:\Sig\to \Sig'
,\qquad 
u_{a,b}(\tH_{(I)}):
= -a \tH_{(\ast,I)}+(a+b) \tH_{(\ast I)}-b \tH_{(I, \ast)}
.
\end{equation}
Towards an explicit description, consider the following example for $I=\{1,2,3\}$,
\begin{align*}
u_{a,b}(\tH_{(12,3)})
&
=\, u_{a,b}(\tH_{(12)})\tH_{(3)}+ \tH_{(12)}u_{a,b}(\tH_{(3)})\\
&
=(-a\tH_{(\ast, 12)}+(a+b)\tH_{(\ast 12)}-b\tH_{(12,\ast)})\tH_{(3)}+
\tH_{(12)}(-a\tH_{(\ast, 3)}+(a+b)\tH_{(\ast 3)}-b\tH_{(3,\ast)})\\
&
=-a\tH_{(\ast, 12,3)}+(a+b)\tH_{(\ast 12,3)}-b\tH_{(12,\ast,3)})
-a\tH_{(12,\ast, 3)}+(a+b)\tH_{(12,\ast 3)}-b\tH_{(12,3,\ast)}
.
\end{align*}
From this, we see that in general
\[
u_{a,b}(\tH_F)
= 
\sum_{1\leq m\leq k}
-a\mathtt{H}_{(S_1,\dots   ,\ast, S_m,\dots,S_k)}
+(a+b)\mathtt{H}_{(S_1,\dots, \ast S_m,\dots,S_k)}
-b\mathtt{H}_{(S_1,\dots, S_m,\ast,\dots,S_k)}
.\]

\begin{thm}\label{steinmannarrowaredercoder}
Given $a,b\in \Bbbk$, the up operator 
\[
\Sig\to \Sig', 
\qquad 
\tH_F \mapsto u_{a,b}(\tH_{F})
\]
is a biderivation of $\Sig$ (and so gives $\Sig$ the structure of a Hopf $\textbf{L}$-algebra).
\end{thm}
\begin{proof}
In the following, for $F=(S_1,\dots,S_k)$ a composition of $I$ and $S\subseteq I$, we write
\[ 
(U_1,\dots,U_k):=(S_1\cap S, \dots ,S_k\cap S)
.\]
In general, $(U_1,\dots,U_k)$ is a decomposition of $I$. 

First, $u_{a,b}$ defines a derivation of $\Sig$ by construction. To see that $u_{a,b}$ also defines a coderivation, we have
\begin{align*}
\Delta_{\ast S,T} \big(  u_{a,b}( \tH_F )\big)\ =\ 
&
\ \ \ \   \Delta_{\ast S,T} \Bigg( \sum_{1\leq m\leq k}
-a\mathtt{H}_{(S_1,\dots,\ast, S_m,\dots,S_k)}
+(a+b)\mathtt{H}_{(S_1,\dots, \ast S_m,\dots,S_k)}
-b\mathtt{H}_{(S_1,\dots, S_m,\ast,\dots,S_k)}\Bigg )\\[7pt]
=\
&
\ \ \ \    
\Bigg(\sum_{1\leq m\leq k}-a\mathtt{H}_{(U_1,\dots,\ast, U_m,\dots,U_k)_+}+
(a+b)\mathtt{H}_{(U_1,\dots, \ast U_m,\dots,U_k)_+}
-b\tH_{(U_1,\dots, U_m,\ast,\dots,U_k)_+}\Bigg)\otimes \tH_{F|_T}\\[7pt]
=\
&
\ \ \ \     
\Bigg(\sum_{\substack{1\leq m\leq k \\[2pt] U_m \neq \emptyset}}-a\mathtt{H}_{(U_1,\dots,\ast, U_m,\dots,U_k)_+}+
(a+b)\mathtt{H}_{(U_1,\dots, \ast U_m,\dots,U_k)_+}
-b\tH_{(U_1,\dots, U_m,\ast,\dots,U_k)_+}\Bigg)\otimes \tH_{F|_T}\\
&+\ \underbrace{ 
\Bigg(\sum_{\substack{1\leq m\leq k\\[2pt]U_m=\emptyset}}\big (-a+(a+b)-b\big)\, \mathtt{H}_{(U_1,\dots,U_{m-1},\ast,U_{m+1},\dots,U_k)_+}\Bigg)}_{=0} \otimes\, \tH_{F|_T}\\[7pt]
=\  
&
\ \ \ \ u(\tH_{F|_S})\otimes \tH_{F|_T}.
\end{align*}
Therefore $u_{a,b}$ is a biderivation of $\Sig$. %Then \autoref{prop:hopfalgebra}, which is for Hopf $\textbf{E}$-algebras, goes through in the same way for  Hopf $\textbf{L}$-algebras, and so $u_{a,b}$ equips $\Sig$ with the structure of a Hopf $\textbf{L}$-algebra.
\end{proof}

%\begin{conj}
%Whilst \textcolor{blue}{(\refeqq{eq:defbider})} may still seem ad hoc, there are indications that this definition is forced, and also $ab=0$ is forced, if one additionally requires commutativity of the biderivation. 
%\end{conj}
%%%%%%%%%%%%%%%%%%%%%%%%%%%%%%%%%%%%%%%%%%%%%%%%%%%%%%%%%%%%%%%%%%%%%%%%
\subsection{The Steinmann Arrows} \label{sec:The Steinmann Arrows}
%%%%%%%%%%%%%%%%%%%%%%%%%%%%%%%%%%%%%%%%%%%%%%%%%%%%%%%%%%%%%%%%%%%%%%%%
We now recall the Steinmann arrows for $\Sig$, whose precise definition is due to Epstein-Glaser-Stora \cite[p.82-83]{epstein1976general}. The Steinmann arrows were first considered by Steinmann in settings where $\Sig$ is represented as operator-valued distributions \cite[Section 3]{steinmann1960}.

Let the \emph{retarded Steinmann arrow} be the up biderivation of $\Sig$ given by
\begin{equation}\label{steindown}
\ast\downarrow(-) :\Sig\to \Sig', 
\qquad    
\ast\downarrow \tH_F:=  u_{1,0}(\tH_F)=
\sum_{1\leq m\leq k}
-\mathtt{H}_{(S_1,\dots,\ast, S_m,\dots,S_k)}
+\mathtt{H}_{(S_1,\dots, \ast S_m,\dots, S_k)} 
.
\end{equation} 
Let the \emph{advanced Steinmann arrow} be the up biderivation of $\Sig$ given by
\begin{equation}\label{steinup}
\ast\uparrow(-) : \Sig\to \Sig' , 
\qquad    
\ast\uparrow \tH_F:=  u_{0,1}(\tH_F)=
\sum_{1\leq m\leq k}
\mathtt{H}_{(S_1,\dots, \ast S_m,\dots, S_k)}
-\mathtt{H}_{(S_1, \dots, S_m,\ast, \dots,S_k)}
.
\end{equation} 
In particular
\[
\ast \downarrow \tH_{(I)}=-\tH_{(\ast,I)} + \tH_{(\ast I)}
\qquad \text{and} \qquad
\ast \uparrow \tH_{(I)}=\tH_{(\ast I)} -\tH_{(I,\ast )}
.\]
We have
\[
\ast \uparrow \tH_{F}\, -\, \ast \downarrow \tH_{F}  =  u_{-1,1}(\tH_F) = [ \tH_{(\ast)} , \tH_{F}]
.\]
This identity appears often in the physics literature for operator-valued distributions, e.g. \cite[Equation 13]{steinmann1960}, \cite[Equation 83]{ep73roleofloc}. The biderivation $u_{-1,1}$ gives $\Sig$ the structure of a Hopf $\textbf{L}$-algebra. This $\textbf{L}$-action is the restriction of the adjoint representation of $\Sig$.

Notice that the Steinmann arrows are commutative up operators. Therefore, by \autoref{prop:primelalsoalg}, we can restrict them to obtain commutative up derivations of $\Zie$,
\[
\ast\downarrow(-): \Zie\to \Zie', 
\qquad 
\mathtt{D}_\cS \mapsto  \ast  \downarrow \mathtt{D}_\cS
\qquad \text{and} \qquad
\ast\uparrow(-):\Zie\to \Zie', 
\qquad 
\mathtt{D}_\cS \mapsto  \ast  \uparrow \mathtt{D}_\cS
.\]
%Thus, we have
%\[
%\ast \downarrow [ \mathtt{D}_{\cS_1}, \mathtt{D}_{\cS_2} ]  =   [   \ast \downarrow  \mathtt{D}_{\cS_1},   \mathtt{D}_{\cS_2} ]+  [     \mathtt{D}_{\cS_1},   \ast \downarrow \mathtt{D}_{\cS_2} ]
%\qquad \text{and} \qquad 
%\ast \uparrow [ \mathtt{D}_{\cS_1}, \mathtt{D}_{\cS_2} ]  =   [   \ast \uparrow  \mathtt{D}_{\cS_1},   \mathtt{D}_{\cS_2} ]+  [     \mathtt{D}_{\cS_1},   \ast \uparrow\mathtt{D}_{\cS_2} ]
%.\]
By \autoref{steinmannarrowaredercoder}, the Steinmann arrows equip $\Sig$ with the structure of a Hopf $\textbf{E}$-algebra (and $\Zie$ with the structure of a Lie $\textbf{E}$-algebra) in two ways. We denote the corresponding $\textbf{E}$-modules, as defined in \textcolor{blue}{(\refeqq{eq:inducedE})}, by
\[
\textbf{E}\bigcdot \Sig \to \Sig,
\qquad
\tH_Y\otimes \mathtt{a}\mapsto
Y\! \downarrow \mathtt{a}=
\underbrace{y_r \downarrow 
\circ\cdots \circ     
y_1 \downarrow}_{\text{invariant of the order}}(\mathtt{a}) 
\]
and 
\[
\textbf{E}\bigcdot \Sig \to \Sig,
\qquad
\tH_Y\otimes \mathtt{a}\mapsto
Y\! \uparrow \mathtt{a}=
\underbrace{y_r \uparrow 
\circ\cdots \circ     
y_1 \uparrow}_{\text{invariant of the order}}(\mathtt{a}) 
\]
where $Y=\{y_1,\dots,y_r\}$. For $Y=[r]:=\{\ast_1 , \dots ,\ast_r \}$, we abbreviate 
\[
\downarrow(-) :=\ast\downarrow(-), \qquad  \downarrow\downarrow(-) :=\{  \ast_1, \ast_2 \} \downarrow(-),\qquad  \dots
\] 
and similarly for the advanced arrow. 

By inspecting the definitions, we see that
\begin{equation}\label{eq:retardadvan}
Y\! \downarrow \tH_{(I)}=
\mathtt{R}_{(Y;I)}:=\! \sum_{Y_1\sqcup Y_2=Y}\overline{\tH}_{(Y_1)}  \tH_{(Y_2\sqcup I)} 
\qquad \text{and} \qquad  
Y\! \uparrow \tH_{(I)}=
\mathtt{A}_{(Y;I)}:=\! \sum_{Y_1\sqcup Y_2=Y}   \tH_{(Y_1\sqcup I)}    \overline{\tH}_{(Y_2)}
.
\end{equation}
It follows that 
\[     
Y\! \downarrow \tH_F
=  
\sum_{Y_1 \sqcup\dots\sqcup Y_k =Y} \mathtt{R}_{(Y_1;S_1)}\dots \mathtt{R}_{(Y_{k};S_{k})} 
\qquad \text{and} \qquad 
Y\! \uparrow \tH_F
=  
\sum_{Y_1 \sqcup\dots\sqcup Y_k =Y} \mathtt{A}_{(Y_1;S_1)}\dots \mathtt{A}_{(Y_{k};S_{k})} 
.\]
The sums are over all decompositions $(Y_1,\dots, Y_k)$ of $Y$ of length $l(F)$. We call $\mathtt{R}_{(Y;I)}, \mathtt{A}_{(Y;I)}\in \Sig[Y\sqcup I]$ the \emph{retarded} and \emph{advanced} elements respectively. The \emph{total retarded} and \emph{total advanced} elements are given by
\[          
Y\! \downarrow \tH_{(i)}=
\mathtt{R}_{(Y;i)}  =\sum_{Y_1\sqcup Y_2=Y}\overline{\tH}_{(Y_1)}\,  \tH_{(Y_2 i)}
\qquad\text{and}\qquad    
Y\! \uparrow \tH_{(i)}=
\mathtt{A}_{(Y;i)}=\sum_{Y_1\sqcup Y_2=Y}   \tH_{(Y_2 i )}\,    \overline{\tH}_{(Y_1)}
\]
respectively. 
\begin{remark}\label{rem:double}
If we put $I=X\sqcup \{i\}$, then we have
\[
\mathtt{R}_{(X;i)}=
\sum_{\substack{S\sqcup T=I\\ i\in T}}\overline{\tH}_{(S)}\,  \tH_{(T)}
=-\sum_{\substack{F\in \Sigma[I]\\ i\in S_k}} (-1)^{l(F)} \tH_F
=\mathtt{D}_i 
\]
and
\[
\mathtt{A}_{(X;i)}=
\sum_{\substack{S\sqcup T=I\\ i\in T}} \tH_{(T)}\, \overline{\tH}_{(S)}=
-\sum_{\substack{F\in \Sigma[I]\\ i\in S_1}} (-1)^{l(F)} \tH_F=
\mathtt{D}_{\bar{i}}
.\]
\end{remark}

\begin{thm} \label{coalgahomo}
We have the following homomorphisms of algebras in species,
\[
\Sig\to \Sig^{\textbf{E}}, 
\qquad
\tH_F\mapsto
\sum_{Y_1 \sqcup\dots\sqcup Y_k =(-)}   
\mathtt{R}_{(Y_1;S_1)}  \dots \mathtt{R}_{(Y_{k};S_{k})}
=
\sum_{r=0}^\infty\,  \sum_{Y_1 \sqcup\dots\sqcup Y_k =[r]}   
\mathtt{R}_{(Y_1;S_1)}  \dots \mathtt{R}_{(Y_{k};S_{k})}
\]
and
\[
\Sig\to \Sig^{\textbf{E}}, 
\qquad
\tH_F\mapsto 
\sum_{Y_1 \sqcup\dots\sqcup Y_k =(-)}   
\mathtt{A}_{(Y_1;S_1)}  \dots \mathtt{A}_{(Y_{k};S_{k})}
=
\sum_{r=0}^\infty\,  \sum_{Y_1 \sqcup\dots\sqcup Y_k =[r]}   
\mathtt{A}_{(Y_1;S_1)}  \dots \mathtt{A}_{(Y_{k};S_{k})}
.\]
\end{thm}
\begin{proof}
The Steinmann arrows are commutative up biderivations of $\Sig$, and so give $\Sig$ the structure of a Hopf $\textbf{E}$-algebra. The maps here are the associated $(-)^{\textbf{E}}$-coalgebras obtained by currying the actions of $\textbf{E}$, which are homomorphisms by \autoref{prop:is a homo}.
\end{proof}

The homomorphisms of \autoref{coalgahomo} are the unique extensions of the maps
\[
\tH_{(I)}\mapsto \sum_{r=0}^\infty \mathtt{R}_{(r;I)}
\qquad \text{and} \qquad
\tH_{(I)}\mapsto \sum_{r=0}^\infty \mathtt{A}_{(r;I)}
\]
to homomorphisms.
%%%%%%%%%%%%%%%%%%%%%%%%%%%%%%%%%%%%%%%%%%%%%%%%%%%%%%%%%%%%%%%%%%%%%%%%
\subsection{The Steinmann Arrows and Dynkin Elements} \label{sec:The Steinmann Arrows and Dynkin Elements}
%%%%%%%%%%%%%%%%%%%%%%%%%%%%%%%%%%%%%%%%%%%%%%%%%%%%%%%%%%%%%%%%%%%%%%%%

\begin{figure}[t]
	\centering
	\includegraphics[scale=0.6]{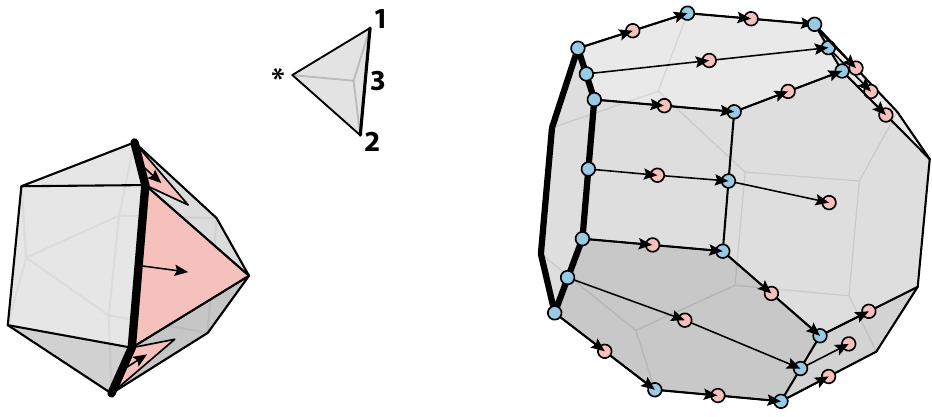}
	\caption{Schematic for the action of the retarded Steinmann arrow $\ast \downarrow$ for $I=\{1,2,3\}$ on the Steinmann sphere (left) and the tropical geometric realization of $\boldsymbol{\Sigma}$ (right, see \cite[Introduction]{norledge2019hopf}).}
	\label{fig:steinarrow}
\end{figure}

We now show that the restriction of the Steinmann arrows to $\Zie$, which are derivations for its Lie bracket, have an interesting description in terms of cells, i.e. chambers of the adjoint braid arrangement.

Following \cite[Section 2]{epstein2016}, we define the commutative up operators
\[
\ast \downarrow(-):\textbf{L}^\vee\to {\textbf{L}^\vee}',
\qquad
\ast \downarrow \cS:=\big \{(\ast S,T),(S,\ast T),(I,\ast):(S,T)\in \cS\big\}
\]
and
\[
\ast\uparrow(-):\textbf{L}^\vee\to {\textbf{L}^\vee}',
\qquad
\ast\uparrow \cS:=\big \{(\ast S,T),(S,\ast T),(\ast,I):(S,T)\in \cS\big\}
.\]
These are indeed well-defined; $\ast \downarrow \cS$ corresponds to the adjoint braid arrangement chamber on the $I$ side of the hyperplane $\lambda_{\ast,I}=0$ which has the face of $\cS$ as a facet, and $\ast\uparrow \cS$ corresponds to the chamber on the $\ast$ side of the hyperplane $\lambda_{\ast,I}=0$ which has the face of $\cS$ as a facet. See around \cite[Remark 2.2]{lno2019} for more details. Thus, it follows from \autoref{prop:ruelle} (Ruelle's identity) that
\[     
[ \tH_{(\ast)}, \mathtt{D}_\cS]= \mathtt{D}_{\ast \uparrow \cS} - \mathtt{D}_{\ast \downarrow \cS} 
.\]
We denote the corresponding $\textbf{E}$-modules by
\[
\textbf{E}\bigcdot \textbf{L}^\vee \to \textbf{L}^\vee
,\qquad 
\tH_Y\otimes \cS \mapsto Y\downarrow \cS=
\big\{ 
( Y_1\sqcup S, Y_2\sqcup T )\in [Y \sqcup I;\text{2}]  : (S,T)\in \cS\ \text{or}\ S=I   
\big\}   
\]
and
\[
\textbf{E}\bigcdot \textbf{L}^\vee \to \textbf{L}^\vee
,\qquad 
\tH_Y\otimes \cS \mapsto Y\uparrow \cS=
\big\{ 
( Y_1\sqcup S, Y_2\sqcup T )\in [Y \sqcup I;\text{2}]  : (S,T)\in \cS\ \text{or}\ T=I   
\big\}      
.\]

\begin{prop}\label{adjointinterp}
Given a cell $\cS$ over $I$, we have
\[        
Y \downarrow\mathtt{D}_\cS=  \mathtt{D}_{Y \downarrow \cS} 
\qquad \text{and} \qquad
Y \uparrow\mathtt{D}_\cS=  \mathtt{D}_{Y \uparrow \cS}
.\]
\end{prop}
\begin{proof}
We consider the retarded case $Y\downarrow\mathtt{D}_\cS=  \mathtt{D}_{Y \downarrow \cS}$ only, since the advanced case then follows similarly. It is sufficient to consider the case $Y=\{\ast\}$. We have
\[   
\downarrow\mathtt{D}_\cS
=  
-\sum_{ \bar{F} \subseteq \cS } (-1)^{l(F)}  \downarrow\tH_{F} 
\qquad \text{and} \qquad 
\mathtt{D}_{\downarrow \cS}= -\sum_{\bar{F} \subseteq \,   \downarrow\cS } (-1)^{l(F)}\,  \tH_{ F }
.\]
So, the result follows if we have the following equality
\[ 
\sum_{ \bar{F} \subseteq \cS } (-1)^{l(F)} \sum_{1\leq m\leq k}
- \mathtt{H}_{(S_1, \dots,\ast, S_m, \dots,S_k)} 
+\mathtt{H}_{(S_1,\dots, \ast S_m,\dots, S_k)}   
\overset{\mathrm{?}}{=} 
\sum_{\bar{G} \subseteq \,   \downarrow\cS } (-1)^{l(G)}\,  \tH_{G}    
.\]
Indeed, notice that the $\tH$-basis elements $\tH_G\in \Sig[\ast I]$ which appear on the LHS are exactly those such that 
\[\bar{G}\subseteq  \downarrow\cS.\] 
Notice also that each $\tH_G$ appears with total sign $(-1)^{l(G)}$, since when $\ast$ is inserted as a singleton lump, thus increasing $l(G)$ by one, it appears also with a negative sign. 
\end{proof} 

\begin{remark}
This interpretation of the $\textbf{E}$-module structure of $\Sig$ restricted to the primitive part $\Zie=\mathcal{P}(\Sig)$ in terms of the adjoint braid arrangement suggests obvious generalizations of the Steinmann arrows in the direction of \cite{aguiar2017topics}, \cite{aguiar2020bimonoids}, since the generalization of Hopf monoids there is via hyperplane arrangements.
\end{remark}

\begin{cor}
We have the following homomorphisms of Lie algebras in species,
\[
\Zie\to \Zie^{\textbf{E}}, 
\qquad
\mathtt{D}_\cS\mapsto \mathtt{D}_{ (-) \downarrow \cS }
=
\sum_{r=0}^\infty \mathtt{D}_{ [r] \downarrow \cS }
= \mathtt{D}_\cS+ \mathtt{D}_{\downarrow \cS}+ \mathtt{D}_{  \downarrow\downarrow \cS}+ \cdots 
\]
and
\[
\Zie\to \Zie^{\textbf{E}}, 
\qquad
\mathtt{D}_\cS\mapsto \mathtt{D}_{ (-) \uparrow \cS }
=
\sum_{r=0}^\infty \mathtt{D}_{ [r]\uparrow \cS }
= \mathtt{D}_\cS+ \mathtt{D}_{\uparrow  \cS}+ \mathtt{D}_{\uparrow \uparrow \cS}+ \cdots 
.\]
\end{cor}
\begin{proof}
The Steinmann arrows are commutative up biderivations of $\Zie$, and so give $\Zie$ the structure of a Lie $\textbf{E}$-algebra. The maps given here are the associated $(-)^{\textbf{E}}$-coalgebras obtained by currying the actions of $\textbf{E}$, which are homomorphisms by \autoref{prop:is a homo}.
\end{proof}
%%%%%%%%%%%%%%%%%%%%%%%%%%%%%%%%%%%%%%%%%%%%%%%%%%%%%%%%%%%%%%%%%%%%%%%%
\section{Products and Series}
%%%%%%%%%%%%%%%%%%%%%%%%%%%%%%%%%%%%%%%%%%%%%%%%%%%%%%%%%%%%%%%%%%%%%%%%
We now recall several basic constructions of casual perturbation theory in the current, clean, abstract setting. We do this without yet imposing causal factorization/causal additivity. We say e.g. `$\text{T}$-product' and `$\text{R}$-product' for now, and then change to `time-ordered product' and `retarded product' in the presence of causal factorization.

%Whether systems of products for $\Sig$ are meaningful without some additional condition like causal factorization seems to be an interesting question.
%%%%%%%%%%%%%%%%%%%%%%%%%%%%%%%%%%%%%%%%%%%%%%%%%%%%%%%%%%%%%%%%%%%%%%%%
\subsection{T-Products, Generalized T-Products, and Generalized R-Products} \label{sec:T-Products, Generalized T-Products, and Generalized R-Products}
%%%%%%%%%%%%%%%%%%%%%%%%%%%%%%%%%%%%%%%%%%%%%%%%%%%%%%%%%%%%%%%%%%%%%%%%
Let $V$ be a vector space over $\Bbbk$, and let $\cA$ be a $\Bbbk$-algebra with multiplication denoted by $\star$. Let a \emph{system of} $\text{T}$\emph{-products} $\text{T}$ be a system of products (as defined in \autoref{sec:Systems of Products}) for the positive exponential species $\textbf{E}^\ast_+$,
\[   
\text{T}:  \textbf{E}^\ast_+  \otimes \textbf{E}_V  \to \textbf{U}_\cA
,\qquad 
\tH_{(I)} \otimes \ssA_I \mapsto 
\text{T}_I(\tH_{(I)} \otimes \ssA_I)
.\]
We abbreviate
\[
\text{T}_I(\ssA_I):=\text{T}_I(\tH_{(I)} \otimes \ssA_I) 
.\]
The currying of $\text{T}$ is denoted by
\[          
\textbf{E}^\ast_+\to \cH(\textbf{E}_V, \textbf{U}_\cA), \qquad \tH_{(I)}\mapsto \text{T}(I)
\]
where the linear maps
\[
\text{T}(I): \textbf{E}_V[I] \to \cA
,\qquad
\ssA_I \mapsto \text{T}_I(\ssA_I)
\] 
are called $\text{T}$\emph{-products}. %If $V$ is a space of bump functions, then $\text{T}$-products $\text{T}(I)$ are `$\cA$-valued distributions'. 
Notice that $\text{T}$-products are commutative in the sense that 
\[
\text{T}_I\big( \textbf{E}_V[\sigma](\ssA_I)\big)=\text{T}_I(\ssA_I) \qquad \quad \text{for all bijections}\quad \sigma:I\to I
.\]
This property holds because the system $\text{T}$ is a morphism of species, and bijections act trivially for $\textbf{U}_\cA$. This commutativity exists despite the fact that the algebra $\cA$ is noncommutative in general. 

\begin{remark}
In applications to QFT, we shall also have a causal structure on $V$. Then $\text{T}$ is meant to first order the vectors of $\ssA_I$ according to the causal structure, and then multiply in $\cA$, giving rise to this commutativity. 
\end{remark}

%The space of system of $\text{T}$-products forms the convolution algebra $\Hom(\textbf{E}^\ast_+\otimes \textbf{E}_V)$. 
%We say system of $\text{T}$-products beccause in QFT, this commutativity corresponds to time-ordering. 

Let the \emph{system of generalized} $\text{T}$\emph{-products} associated to a system of $\text{T}$-products be the unique extension to an algebraic system of products for $\Sig=\textbf{L}\boldsymbol{\circ}\textbf{E}_+^\ast$, thus
\[  
\text{T}:  \Sig\otimes \textbf{E}_V  \to \textbf{U}_\cA
,\qquad 
\tH_F\otimes \ssA_I\mapsto \text{T}_I(\tH_F\otimes \ssA_I):=
\text{T}_{S_1}(\ssA_{S_1}) \star \dots \star \text{T}_{S_k}(\ssA_{S_k})
.\]
The currying of $\text{T}$ is denoted by
\[          
\Sig\to \cH(\textbf{E}_V, \textbf{U}_\cA), \qquad \tH_{F}\mapsto \text{T}(S_1)\dots \text{T}(S_k)
.\]
The linear maps 
\[
\text{T}(S_1)\dots \text{T}(S_k): \textbf{E}_V[I] \to \cA
,\qquad
\ssA_I\mapsto \text{T}_I(\tH_F\otimes \ssA_I)
\] 
are called \emph{generalized} $\text{T}$-\emph{products}. Let the \emph{system of generalized} $\text{R}$\emph{-products} associated to a system of $\text{T}$-products be the restriction to the Lie algebra of primitive elements $\Zie$,
\[  
\text{R}: \Zie \otimes \textbf{E}_V  \to \textbf{U}_\cA, \qquad  
\mathtt{D}_\cS\otimes \ssA_I\mapsto
\text{R}_I(\mathtt{D}_\cS \otimes \ssA_I):=\text{T}_I(\mathtt{D}_\cS  \otimes \ssA_I)   
.\]
This is a morphism of Lie algebras, where $\textbf{U}_\cA$ is equipped with the commutator bracket. The currying of $\text{R}$ is denoted by
\[          
\Zie \to \cH(\textbf{E}_V, \textbf{U}_\cA), \qquad \mathtt{D}_\cS\mapsto \text{R}_\cS
.\]
The linear maps 
\[
\text{R}_\cS: \textbf{E}_V[I] \to \cA
,\qquad
\ssA_I\mapsto  \text{R}_I(\mathtt{D}_\cS \otimes \ssA_I) 
\] 
are called \emph{generalized} $\text{R}$-\emph{products}. From the expansion \textcolor{blue}{(\refeqq{eq:hbasisexp})} of Dynkin elements $\mathtt{D}_\cS$ in terms of the $\tH$-basis, we recover \cite[Equation 79]{ep73roleofloc},
\[ 
\text{R}_\cS=-\sum_{\cF_{F}\subseteq\bar{\cS}} (-1)^{k}\, \text{T}(S_1)\dots \text{T}(S_k)
.\]
Consider a system of products of the form
\[ 
\text{Z}:\textbf{E}^\ast_+\otimes {\textbf{E}_{V}} 
\to 
\textbf{U}_{V}
,\qquad 
\tH_{(I)}\otimes \ssA_I\mapsto \text{Z}_I(\ssA_I)
.\]
Then we obtain a new $\text{T}$-product $\text{T}'$, given by
\[  
\text{T}':  
\textbf{E}^\ast_+\otimes  \textbf{E}_V  \to \textbf{U}_\cA
,\qquad 
\text{T}'_I(\ssA_I)
:= 
\sum_{P}\text{T}_P     
\big(
\text{Z}_{S_1} (\ssA_{S_1})\dots  \text{Z}_{S_k}(\ssA_{S_k})
\big)    
.\]
The sum is over all partitions $P=\{S_1,\dots,S_k\}$ of $I$. This construction underlies renormalization in pAQFT \cite[Section 3.6.2]{dutsch2019perturbative}, which deals with the remaining ambiguity of $\text{T}$-products after imposing causal factorization, and perhaps other renormalization conditions.

\begin{remark}
We have the presheaf of compositions $\Sigma:\textsf{S}^{\text{op}}\to \textsf{Set}$, and we may pick a monoidal copresheaf $\textbf{E}_V:\textsf{S}\to \textsf{Vec}$. Then a system of generalized $\text{T}$-products is equivalently a homomorphism on the weighted colimit $\overline{\cK}_V(\Sig)=\Sigma\otimes_{\textsf{S}} \textbf{E}_V$ into some $\Bbbk$-algebra $\cA$. In fact, letting $\textbf{E}_V$ be a copresheaf is more natural in the application to QFT.
\end{remark} 

%In the presence of a causal structure and causal factorization, generalized $\text{R}$-products will have certain important support properties, as a consequence of \autoref{actondynkinwithtits}. 
%\begin{remark}
%The definition of a system of $\text{T}$-products does not use the algebraic structure of $\cA$. If $V$ is a space of bump functions on a spacetime $\cX$, then causal factorization is an additional property one can impose which uses the algebraic structure of $\cA$. We shall call such system of $\text{T}$-products time-ordered products.
%\end{remark}

%\begin{remark}
%A system of products of the form
%\[\Sig\otimes \textbf{E}_V  \to \textbf{U}_\cA\] 
%is a natural noncommutative generalization of the system of products for polynomial functions from \autoref{ex:polyfun},
%\[\textbf{E}\otimes \textbf{E}_V  \to \textbf{U}_{\text{PolyFunc}(V^\ast)}.\] 
%Products of the form $\Sig\otimes \textbf{E}_V  \to \textbf{U}_\cA$ appear in QFT, however their significance there relies on the fact that $V$ consists of bump functions on a spacetime, which has a causal structure. In particular, one imposes the property of causal factorization on systems of $\text{T}$-products. %So one can ask, are products $\Sig\otimes \textbf{E}_V  \to \textbf{U}_\cA$ meaningful without this causal structure and causal factorization.
%\end{remark}
  
%%%%%%%%%%%%%%%%%%%%%%%%%%%%%%%%%%%%%%%%%%%%%%%%%%%%%%%%%%%%%%%%%%%%%%%%
\subsection{Reverse T-Products}
%%%%%%%%%%%%%%%%%%%%%%%%%%%%%%%%%%%%%%%%%%%%%%%%%%%%%%%%%%%%%%%%%%%%%%%%

The system of \emph{reverse generalized} $\text{T}$\emph{-products} $\overline{\text{T}}$ of a system of generalized $\text{T}$-products is given by precomposing $\text{T}$ with the antipode \textcolor{blue}{(\refeqq{eq:antipode})} of $\Sig\otimes \textbf{E}_V$, thus
\[ 
\overline{\text{T}}:\Sig \otimes {\textbf{E}}_V\to \textbf{U}_{\cA^{\op}}, 
\qquad 
\overline{\text{T}}_I(\tH_{F}\otimes \ssA_I)
:=
{\text{T}}_I\big(\overline{\tH}_{F}\otimes \ssA_I\big)
.\]
Since the antipode is a homomorphism $\Sig\otimes \textbf{E}_V\to (\Sig\otimes \textbf{E}_V)^{\op, \text{cop}}$ \cite[Proposition 1.22 (iii)]{aguiar2010monoidal}, this is a system of generalized $\text{T}$-products into the opposite algebra $\textbf{U}_{\cA^{\op}}$. The image of $\tH_{(I)}$ under the currying of $\overline{\text{T}}$ is called the \emph{reverse} $\text{T}$\emph{-product} 
\[
\overline{\text{T}}(I): \textbf{E}_V[I] \to \cA^{\text{op}}
.\] 
From \textcolor{blue}{(\refeqq{antipode})}, we obtain
\[
\overline{\text{T}}(I) 
=
\sum_{F\in \Sigma[I]} (-1)^{k}\, \text{T}(S_1)\dots \text{T}(S_k)
.\]
Note that reverse $\text{T}$-products in \cite[Equation 11]{ep73roleofloc} are defined to be $(-1)^n\, \overline{\text{T}}(I)$. Our definition agrees with \cite[Definition 15.35]{perturbative_quantum_field_theory}.

%%%%%%%%%%%%%%%%%%%%%%%%%%%%%%%%%%%%%%%%%%%%%%%%%%%%%%%%%%%%%%%%%%%%%%%%
\subsection{T-Exponentials} \label{sec:T-Exponentials}
%%%%%%%%%%%%%%%%%%%%%%%%%%%%%%%%%%%%%%%%%%%%%%%%%%%%%%%%%%%%%%%%%%%%%%%%
The (scaled) \emph{universal series} $\mathtt{G}(c)\in \mathscr{S}(\Sig)$ is the group-like series of $\Sig$ given by
\[   
\mathtt{G}(c):  \textbf{E} \to  \Sig
,\qquad  
\tH_{I}\mapsto \mathtt{G}(c)_{I}:= c^n\,  \tH_{(I)}    
\qquad  \quad   
\text{for} \quad  
c\in \Bbbk  
.\]
The fundamental nature of this series is described in \cite[Section 13.6]{aguiar2013hopf}. The series $\text{s} \circ \mathtt{G}(c)$ which is the composition of $\mathtt{G}(c)$ with the antipode $\text{s}$ of $\Sig$ is given by
\begin{equation}\label{inverseuni}
\text{s} \circ \mathtt{G}(c): \textbf{E} \to  \Sig
,\qquad  
\tH_I\mapsto \big (\text{s} \circ \mathtt{G}(c)\big )_I= c^n\, \overline{\tH}_{(I)}
.
\end{equation}
Given a system of generalized $\text{T}$-products $\text{T}:\Sig \otimes {\textbf{E}}_V\to \textbf{U}_\cA$, let the $\text{T}$-\emph{exponential} $\mathcal{S}:=\mathcal{S}_{\mathtt{G}(c)}$ of this system be the $\cA[[\formj]]$-valued function on the vector space $V$ associated to the series $\mathtt{G}(c)$, as constructed in \autoref{sec:decseries}. Thus, we have
\begin{equation}\label{eq:tsexp}
\mathcal{S}: V\to \cA[[ \formj ]] , 
\qquad        
\ssA\mapsto  \mathcal{S}(\formj\! \ssA)
=
\sum^\infty_{n=0} \dfrac{c^n}{n!} \text{T}_n 
\underbrace{\big ( \formj\! \ssA\otimes  \cdots \otimes \formj\! \ssA \big )}_{ \text{ $n$ times } }   
=
\sum^\infty_{n=0} \dfrac{\formj^n c^n}{n!} \text{T}_n(\ssA^n)     
.
\end{equation}
%In the presence of causal factorization, $\mathcal{S}$ will be called an S-matrix scheme. Then the following is the result that reverse system of $\text{T}$-products express the inverse S-matrix \cite[Section 2]{ep73roleofloc}. 
By \textcolor{blue}{(\refeqq{inverse2})} and \textcolor{blue}{(\refeqq{inverseuni})}, the $\text{T}$-exponential for the system of reverse $\text{T}$-products is the inverse of $\mathcal{S}$ as an element of the $\Bbbk$-algebra of functions $\text{Func}(V,\cA[[\formj]])$, given by
\[
\mathcal{S}^{-1}: V\to \cA[[ \formj ]] , 
\qquad        
\ssA\mapsto  \mathcal{S}^{-1}(\formj\! \ssA) :=
\sum^\infty_{n=0} \dfrac{\formj^n c^n}{n!} \overline{\text{T}}_n(\ssA^n)
=
\sum^\infty_{n=0} \dfrac{\formj^n c^n}{n!} \text{T}_n(\overline{\tH}_{(n)}\otimes \ssA^n)
.\]
Therefore
\[
\mathcal{S}(\formj\! \ssA)\star \mathcal{S}^{-1}(\formj\! \ssA)=
\mathcal{S}^{-1}(\formj\! \ssA)\star \mathcal{S}(\formj\! \ssA) = 1_\cA
\]
for all $\ssA\in V$. This appears in e.g. \cite[Equation 2]{ep73roleofloc}.

%%%%%%%%%%%%%%%%%%%%%%%%%%%%%%%%%%%%%%%%%%%%%%%%%%%%%%%%%%%%%%%%%%%%%%%%
\section{Perturbation of T-Products}
%%%%%%%%%%%%%%%%%%%%%%%%%%%%%%%%%%%%%%%%%%%%%%%%%%%%%%%%%%%%%%%%%%%%%%%%
We now construct three perturbations of T-products by $\textbf{E}$-actions, as considered in \autoref{sec:Perturbation of Products and Series}.

%%%%%%%%%%%%%%%%%%%%%%%%%%%%%%%%%%%%%%%%%%%%%%%%%%%%%%%%%%%%%%%%%%%%%%%%
\subsection{Perturbation of T-Products by Up Coderivation}   \label{sec:Perturbation of T-Products by Up Coderivation}
%%%%%%%%%%%%%%%%%%%%%%%%%%%%%%%%%%%%%%%%%%%%%%%%%%%%%%%%%%%%%%%%%%%%%%%%
%The positive coalgebra $\textbf{E}_+^\ast$ has a commutative up coderivation, given by
%\[              
%u:\textbf{E}_+^\ast \to {\textbf{E}_+^\ast}', 
%\qquad 
%u(\tH_{(I)}):=  \tH_{(\ast I)} 
%.\] 

Suppose we have a system of \hbox{$\text{T}$-products}
\[   
\text{T}:  \textbf{E}^\ast_+  \otimes \textbf{E}_V  \to \textbf{U}_\cA
,\qquad \tH_{(I)}\otimes \ssA_I  \mapsto \text{T}_I(\ssA_I)
.\]
Following \autoref{sec:Perturbation of Products and Series}, given a choice of decorations vector $\ssS\in V$, we can use the up coderivation 
\[
u:\textbf{E}_V\to \textbf{E}'_V,\qquad u(\ssA_I)=\ssS \ssA_I
\] 
from \autoref{sec:Dec} to perturb $\text{T}$, giving
\[
\textbf{E}_+^\ast \otimes \textbf{E}_V\to  \textbf{U}_{\cA[[\formg]]}
,\quad
\tH_{(I)}\otimes \ssA_I 
\mapsto  
\sum_{r=0}^\infty \dfrac{\formg^{\, r\, }}{r!}\text{T}^r_{I}(\ssS^{\, r}\ssA_I)
=
\text{T}_{I}(\ssA_I) 
+
\underbrace{\formg\,  \text{T}^1_{I} (\ssS \ssA_I)
+
\dfrac{\formg^2}{2!} \text{T}^2_{I}(\ssS \ssS \ssA_I)
+ \cdots}_{\text{perturbation}}\ 
.\]
Let $\mathcal{S}_{\formg\ssS}$ denote the $\text{T}$-exponential for this new perturbed system of $\text{T}$-products (as defined in \textcolor{blue}{(\refeqq{eq:tsexp})}), thus
\[  
\mathcal{S}_{\formg\ssS} : V\to \cA[[\formg,\! \formj]]
,\qquad   
\ssA \mapsto \mathcal{S}_{\formg\ssS}(\formj\! \ssA)
:=
\sum^\infty_{n=0}
\sum_{r=0}^\infty 
\dfrac{\formg^{r}\formj^n\, c^{r+n}}{r!\,  n!} 
\text{T}^r_{n}(\ssS^{\, r} \ssA^n)  
\]
for $c \in \Bbbk$.

\begin{prop}
Given vectors $\ssS, \ssA\in V$, we have
\[
\mathcal{S}_{\formg \ssS}(\formj\! \ssA)= \mathcal{S}(\formg \ssS+ \formj\! \ssA)         
.\]
\end{prop}
\begin{proof}
We have
\[
\mathcal{S}_{\formg \ssS}(\formj\! \ssA)
=\sum^\infty_{n=0}\sum^\infty_{r=0}\dfrac{\formg^{r} \formj^n\, c^{r+n}}{r!\,  n!}{\text{T}^r_{n}}(\ssS^{\, r} \ssA^n)
=\sum^\infty_{k=0} \dfrac{c^k}{k!} \text{T}_k
\underbrace{(\formg\ssS+\formj\! \ssA\otimes   \cdots \otimes  \formg\ssS+\formj\! \ssA)}_{k \text{ times}}
.\]
The second equality is obtained by putting $k=n+r$.
\end{proof}
%%%%%%%%%%%%%%%%%%%%%%%%%%%%%%%%%%%%%%%%%%%%%%%%%%%%%%%%%%%%%%%%%%%%%%%%
\subsection{Perturbation of T-Products by Steinmann Arrows}  \label{sec:Perturbation of T-Products by Steinmann Arrows}
%%%%%%%%%%%%%%%%%%%%%%%%%%%%%%%%%%%%%%%%%%%%%%%%%%%%%%%%%%%%%%%%%%%%%%%%
Suppose we have a system of generalized \hbox{$\text{T}$-products}
\[   
\text{T}:  \Sig \otimes \textbf{E}_V  \to \textbf{U}_\cA
,\qquad  
\tH_F\otimes \ssA_I \mapsto \text{T}_I(\tH_F\otimes \ssA_I)
.\]
Following \autoref{sec:Perturbation of Products and Series}, given a choice of decorations vector $\ssS\in V$, we can use the Steinmann arrows \textcolor{blue}{(\refeqq{steindown})} and \textcolor{blue}{(\refeqq{steinup})} to perturb $\text{T}$, giving
\[ 
{\text{T}_{\downarrow,\ssS}}:\Sig\otimes\textbf{E}_V \to \textbf{U}_{\cA[[\formg]]}, 
\qquad 
{\text{T}_{\downarrow,\ssS}}:
=\text{T}^{\textbf{E}} \circ {\widecheck{\rho}_{\downarrow,\ssS}}
\]
for the retarded arrow, and
\[ 
{\text{T}_{\uparrow,\ssS}}:\Sig\otimes\textbf{E}_V \to \textbf{U}_{\cA[[\formg]]}, 
\qquad 
{\text{T}_{\uparrow,\ssS}}:
=\text{T}^{\textbf{E}}\circ \widecheck{\rho}_{\uparrow,\ssS}
\]
for the advanced arrow. We sometimes denote $\widetilde{\text{T}}:=\, {\text{T}_{\downarrow,\ssS}}$. Towards an explicit description of these new perturbed systems of generalized $\text{T}$-products, given
\[
\ssS_Y\ssA_I=  \ssS_{y_1}\otimes \dots \otimes \ssS_{y_r} \otimes \ssA_{i_1} \otimes \dots \otimes \ssA_{i_n} \in \textbf{E}_V^{[Y]}[I]
,\]
let
\begin{equation}\label{eq:retardprod}
\text{R}_{Y;I}(\ssS_Y;\ssA_I)
:= 
\underbrace{\text{T}^{[Y]}_{I}( \mathtt{R}_{(Y;I)} \otimes \ssS_Y \ssA_I)
=
\sum_{Y_1 \sqcup Y_2=Y}  \overline{\text{T}}^{[Y_1]}_{\emptyset}(\ssS_{Y_1}) \star \text{T}^{[Y_2]}_{I}( \ssS_{Y_2} \ssA_I)}_{\text{by } \textcolor{blue}{(\refeqq{eq:retardadvan}})} 
\end{equation}
and    
\[
\text{A}_{Y;I}(\ssS_Y;\ssA_I)
:= 
\underbrace{\text{T}^{[Y]}_{I}(  \mathtt{A}_{(Y;I)} \otimes \ssS_Y \ssA_I)
=
\sum_{Y_1 \sqcup Y_2=Y}  \text{T}^{[Y_1]}_{I}( \ssS_{Y_1} \ssA_I) \star \overline{\text{T}}^{[Y_2]}_{\emptyset}(\ssS_{Y_2})}_{\text{by } \textcolor{blue}{(\refeqq{eq:retardadvan}})} 
.\]
Then perturbation by the retarded arrow $\widetilde{\text{T}}=\, {\text{T}_{\downarrow,\ssS}}$ is given by
\[ 
\widetilde{\text{T}}:\Sig\otimes\textbf{E}_V \to \textbf{U}_{\cA[[\formg]]}
,\quad 
\tH_F\otimes \ssA_I \mapsto \sum_{r=0}^\infty\    \sum_{r_1 +\, \cdots\,  + r_k=r } \dfrac{\formg^{r}}{r!}  \text{R}_{r_1;S_1}(\ssS^{\, r_1};\ssA_{S_1})\star \cdots\star \text{R}_{r_k;S_k}(\ssS^{\, r_k};\ssA_{S_k}) 
.\footnote{\ as usual, we abbreviate $\text{R}_{r;I}(\ssS^{\, r};\ssA_I):=\text{R}_{[r];I}(\ssS^{\, [r]};\ssA_I)=\text{R}_{[r];I}(\underbrace{\ssS\otimes \dots \otimes \ssS}_{\text{$r$ times}} \, ;  \ssA_I)$}\]
In particular, the restriction to $\textbf{E}_+^\ast\otimes\textbf{E}_V$, i.e. the new $\text{T}$-product, is given by
\begin{align*}
\widetilde{\text{T}}_I(\ssA_I)&= \sum_{r=0}^\infty \dfrac{\formg^{r}}{r!}  \text{R}_{r;I}(\ssS^{\, r};\ssA_I) \\
&=\text{T}_I(\ssA_I) + \underbrace{\formg\, \text{T}^1_I(\downarrow \tH_{(I)} \otimes \ssS  \ssA_I) + \dfrac{\formg^2}{2!} \text{T}^2_I(\downarrow \downarrow \tH_{(I)} \otimes \ssS \ssS   \ssA_I)  + \cdots}_{\text{perturbation}}\ .
\end{align*}
Similarly for the advanced Steinmann arrow. Recall that the Steinmann arrows give $\Sig$ the structure of a Hopf $\textbf{E}$-algebra in two ways (\autoref{steinmannarrowaredercoder}), and so ${\text{T}_{\downarrow,\ssS}}$ and ${\text{T}_{\uparrow,\ssS}}$ are indeed new systems of generalized $\text{T}$-products by \hbox{\autoref{prop:perturbedishomo}}, i.e. they are homomorphisms of algebras.

We let $\mathcal{V}_{\! \formg\ssS}$ and $\mathcal{W}_{\! \formg\ssS}$ denote the $\text{T}$-exponentials for these new perturbed systems of generalized $\text{T}$-products (as defined in \textcolor{blue}{(\refeqq{eq:tsexp})}), thus
\[
\mathcal{V}_{\! \formg\ssS}: V\to \cA[[\formg,\! \formj]],
\qquad
\mathcal{V}_{\! \formg\ssS}(\formj\! \ssA)
:=
\sum_{n=0}^\infty 
\dfrac{\formj^n c^n}{n!}\,  {^{\downarrow, \ssS}\text{T}}_n (\ssA^n)
=
\sum_{n=0}^\infty \sum_{r=0}^\infty 
\dfrac{\formg^{r} \formj^n c^{r+n}}{r!\, n!}\,  \text{R}_{r;n} (\ssS^{\, r} ; \ssA^n)         
\]
and
\[
\mathcal{W}_{\! \formg\ssS}: V\to \cA[[\formg,\! \formj]],
\qquad
\mathcal{W}_{\! \formg\ssS}(\formj\! \ssA):=\sum_{n=0}^\infty 
\dfrac{\formj^n c^n}{n!}\,  {^{\uparrow, \ssS}\text{T}}_n (\ssA^n)
=
\sum_{n=0}^\infty \sum_{r=0}^\infty 
\dfrac{\formg^{r} \formj^n c^{r+n}}{r! \, n!}\,  \text{A}_{r;n} (\ssS^{\, r};\ssA^n)         
.\]

\begin{thm}
We have
\[
\mathcal{V}_{\formg \ssS}(\formj\! \ssA)= 
\mathcal{S}^{-1}( \formg \ssS)\star \mathcal{S}(\formg \ssS +\formj\! \ssA )
\qquad \text{and} \qquad 
\mathcal{W}_{\formg \ssS}(\formj\! \ssA)= 
\mathcal{S}(\formg \ssS +\formj\! \ssA )\star \mathcal{S}^{-1}( \formg \ssS)
.\]
\end{thm}
\begin{proof}
We have
\[         
\text{R}_{r;I}(\ssS^{\, r}; \ssA_I)=
\sum_{Y_1\sqcup Y_2=[r]} \overline{\text{T}}^{[Y_1]}_{\emptyset}(\ssS^{Y_1})   
\star 
\text{T}^{[Y_2]}_{I}( \ssS^{Y_2}\ssA_I) 
.\]
Then
\begin{align*}
\mathcal{V}_{\formg \ssS}(\formj\! \ssA)
=&\  
\sum_{n=0}^\infty \sum_{r=0}^\infty 
\dfrac{\formg^{r} \formj^n c^{r+n}}{r!\, n!}\,  \text{R}_{r;n} (\ssS^{\, r} ; \ssA^n) \\[6pt]
=&\  
\sum_{n=0}^\infty \sum_{r=0}^\infty 
\dfrac{\formg^{r} \formj^n c^{r+n}}{r!\, n!} 
\sum_{Y_1\sqcup Y_2=[r]} \overline{\text{T}}^{[Y_1]}_{\emptyset}(\ssS^{Y_1})   
\star 
\text{T}^{[Y_2]}_{n}(\ssS^{Y_2}\ssA^n) \\[6pt]
=& \ 
\sum^\infty_{r=0} \dfrac{\formg^{r} c^r}{r!} \overline{\text{T}}^r_{\emptyset}(\ssS^{\, r}) 
\star
\sum_{n=0}^\infty \sum_{r=0}^\infty \dfrac{c^n}{n!} \text{T}^r_{n} (\ssS^{\, r} \ssA^n) 
\\[6pt]
=& \ 
\mathcal{S}^{-1}( \formg \ssS)\star \mathcal{S}(\formg \ssS +\formj\! \ssA )
\end{align*}
The case for $\mathcal{W}_{\formg \ssS}(\formj\! \ssA)$ is then similar.
\end{proof}

In the following, recall formal differentiation \textcolor{blue}{(\refeqq{eq:formalder})}.

\begin{cor}[Bogoliubov's Formula {\cite[Chapter 4]{Bogoliubov59}}] 
We have
\begin{equation} \label{eq:Bog}
\widetilde{\text{T}}_i(\ssA)
=
\dfrac{1}{c}\,  \dfrac{d}{d \formj }  \Bigr|_{\formj=0} \mathcal{V}_{\formg \ssS}(\formj\! \ssA)
.
\end{equation}
\end{cor}
\begin{proof}
This is trivial. We have
\[
\dfrac{d}{d \formj }\mathcal{V}_{\formg \ssS}(\formj\! \ssA)
=
\dfrac{d}{d \formj } \sum_{n=0}^\infty 
\dfrac{\formj^n c^n}{n!}\, \widetilde{\text{T}}_n (\ssA^n)
=
\sum_{n=1}^\infty 
\dfrac{\formj^{n-1} c^n}{(n-1)!}\,  \widetilde{\text{T}}_n (\ssA^n)
.\]
Then, putting $\formj=0$, we obtain
\[
\dfrac{d}{d \formj }  \Bigr|_{\formj=0}  \mathcal{V}_{\formg \ssS}(\formj\! \ssA) 
= 
c\,  \widetilde{\text{T}}_1 (\ssA)
.\qedhere \]
\end{proof}

This formula was originally motivated by the path integral heuristic, see e.g. \cite[Remark 15.16]{perturbative_quantum_field_theory}. %which is closely related to the idea that (interacting) vacuum expectation values are obtained by summing over all possible interactions of virtual particles. %The more modern and rigorous version of this is the so-called worldline formalism, which is a perspective on pQFT that comes from perturbative string theory.

%%%%%%%%%%%%%%%%%%%%%%%%%%%%%%%%%%%%%%%%%%%%%%%%%%%%%%%%%%%%%%%%%%%%%%%%
\subsection{$\text{R}$-Products and $\text{A}$-Products} 
%%%%%%%%%%%%%%%%%%%%%%%%%%%%%%%%%%%%%%%%%%%%%%%%%%%%%%%%%%%%%%%%%%%%%%%%
The linear maps $\text{R}(Y;I)$ given by
\[
\text{R}(Y;I):\textbf{E}_V^{[Y]}[I] \to \cA
,\qquad
\ssS_Y \ssA_I \mapsto \text{R}_{Y;I}(\ssS_Y \ssA_I)
\] 
are called R\emph{-products}. In the case of singletons $I=\{i\}$, the maps $\text{R}(Y;i)$ are called \emph{total} R\emph{-products}. By \textcolor{blue}{(\refeqq{eq:retardadvan})}, $\text{R}$-products are given in terms of $\text{T}$-products and reverse $\text{T}$-products by
\[
\text{R}(Y;I)=\sum_{Y_1 \sqcup Y_2 =Y} \overline{\text{T}}(Y_1) \star  \text{T}(Y_2\sqcup I)  
.\]
Then
\[
\widetilde{\text{T}}(I)= \sum_{r=0}^\infty  \dfrac{c^r}{r!} \text{R}(r;I)
.\]
In a similar way, we can define the A-\emph{products} $\text{A}(Y;I)$, so that 
\[
\text{A}(Y;I)=\sum_{Y_1 \sqcup Y_2 =Y}  \text{T}(Y_1 \sqcup  I)   \star \overline{\text{T}}(Y_2)
.\]
The total R-products are both R-products and generalized R-products, which is due to the double description appearing in \autoref{rem:double}. A related result is \cite[Proposition 109]{aguiar2013hopf}.

\begin{remark}
In the literature, the total retarded products in our sense are sometimes called retarded products, and the retarded products in our sense are then called generalized retarded products, e.g. \cite{polk58}, \cite[Exercise 3.3.16]{dutsch2019perturbative}. 
\end{remark}

%%%%%%%%%%%%%%%%%%%%%%%%%%%%%%%%%%%%%%%%%%%%%%%%%%%%%%%%%%%%%%%%%%%%%%%%
%%%%%%%%%%%%%%%%%%%%%%%%%%%%%%%%%%%%%%%%%%%%%%%%%%%%%%%%%%%%%%%%%%%%%%%%
%%%%%%%%%%%%%%%%%%%%%%%%%%%%%%%%%%%%%%%%%%%%%%%%%%%%%%%%%%%%%%%%%%%%%%%%
%%%%%%%%%%%%%%%%%%%%%%%%%%%%%%%%%%%%%%%%%%%%%%%%%%%%%%%%%%%%%%%%%%%%%%%%
%%%%%%%%%%%%%%%%%%%%%%%%%%%%%%%%%%%%%%%%%%%%%%%%%%%%%%%%%%%%%%%%%%%%%%%%
\part{Perturbative Algebraic Quantum Field Theory}
%%%%%%%%%%%%%%%%%%%%%%%%%%%%%%%%%%%%%%%%%%%%%%%%%%%%%%%%%%%%%%%%%%%%%%%%
%%%%%%%%%%%%%%%%%%%%%%%%%%%%%%%%%%%%%%%%%%%%%%%%%%%%%%%%%%%%%%%%%%%%%%%%
%%%%%%%%%%%%%%%%%%%%%%%%%%%%%%%%%%%%%%%%%%%%%%%%%%%%%%%%%%%%%%%%%%%%%%%%
%%%%%%%%%%%%%%%%%%%%%%%%%%%%%%%%%%%%%%%%%%%%%%%%%%%%%%%%%%%%%%%%%%%%%%%%
%%%%%%%%%%%%%%%%%%%%%%%%%%%%%%%%%%%%%%%%%%%%%%%%%%%%%%%%%%%%%%%%%%%%%%%%

We now apply the general theory we have developed to the case of a real scalar quantum field on a Minkowski spacetime, as described by pAQFT (pAQFT deals more generally with perturbative Yang-Mills gauge theory on curved spacetimes). Mathematically, the important extra property is a causal structure on the vector space of decorations $V$, which allows one to impose causal factorization. 

Our references for pAQFT are \cite{dutfred00}, \cite{rejzner2016pQFT}, \cite{dutsch2019perturbative}, \cite{perturbative_quantum_field_theory}. We mainly adopt the notation and presentation of \cite{perturbative_quantum_field_theory}. Key features of pAQFT are its local, i.e. sheaf-theoretic, approach, the (closely related) use of adiabatic switching of interaction terms to avoid IR-divergences, and the interpretation of renormalization as the extension of distributions to the fat diagonal to avoid UV-divergences. It provides a rigorous mathematical framework in which to consider the Wilsonian cutoff. 

Note that connections between QFT and species have been previously studied in \cite{MR2036353}, \cite{MR2862982}, \cite{MR3753672}.

%%%%%%%%%%%%%%%%%%%%%%%%%%%%%%%%%%%%%%%%%%%%%%%%%%%%%%%%%%%%%%%%%%%%%%%%
\section{Spacetime and Field Configurations} 
%%%%%%%%%%%%%%%%%%%%%%%%%%%%%%%%%%%%%%%%%%%%%%%%%%%%%%%%%%%%%%%%%%%%%%%%

Let $\cX\cong \bR^{1,p}$ denote a $(p+1)$-dimensional Minkowski spacetime, for $p\in \bN$. Thus, $\cX$ is a real vector space equipped with a metric tensor which is a symmetric nondegenerate bilinear form $\cX\times \cX\to \bR$ with signature $(1,p)$. The bilinear form gives rise to a volume form on $\cX$, which we denote by $\text{dvol}_\cX\in \Omega^{p+1}(\cX)$. For regions of spacetime $X_1,X_2\subset \cX$, we write 
\[
X_1\! \vee\! \! \wedge X_2
\] 
if one cannot travel from $X_1$ to $X_2$ on a future-directed timelike or lightlike curve. We have the set species $\cX^{(-)}$ given by
\[
I\mapsto \cX^I:= \big \{ \text{functions}\ I\to \cX \big\}   
.\]

\begin{remark} 
If $p=0$, or we just take the time components, then the space of configurations \hbox{$\lambda:I\to \cX$} (we do not require injectivity) modulo translations of $\cX$ is the (essentialized) braid arrangement over $I$. This configuration space is a tropical algebraic torus, so we may take the toric compactification with respect to the braid arrangement fan, denoted $\bT \Sigma$, giving the tropical version of the Losev-Manin moduli space \cite{losevmanin}. The boundary at infinity consists of would-be limiting configurations where points are separated by infinite times. Compositions may be identified with limiting configurations as follows, e.g. the composition $(12,3)$ is the configuration where $1$ and $2$ coincide, with an infinite time separating them from $3$. Then a natural Hopf monoidal structure\footnote{\ in the sense of \cite[Section 4.3]{aguiar2013hopf}} on $\bT\Sigma$ induces the structure of $\Sig$, see \cite[Introduction]{norledge2019hopf}. 
\end{remark}

%References for this section are \cite[Sections 3 and 7]{perturbative_quantum_field_theory}, \cite{bar15}. Given a one-dimensional real vector space $\text{Fib}\in \textsf{Vec}$,\footnote{\ pAQFT does deal more generally with perturbative Yang-Mills theory} we denote its dual space by $\text{Fib}^\ast$, and the pairing by
%\[
%\la - , - \ra :\text{Fib}^\ast \otimes \text{Fib}\to \bR
%.\] 

For simplicity, we restrict ourselves to the Klein-Gordan real scalar field on $\cX$. Therefore, let $E\to \cX$ be a smooth real vector bundle over $\cX$ with one-dimensional fibers. An (off-shell) \emph{field configuration} $\Phi$ is a smooth section of the bundle $E\to \cX$,
\[
\Phi:\cX\hookrightarrow E
,\qquad
x\mapsto \Phi(x)
.\]
The space of all field configurations, denoted $\Gamma(E)$, has the structure of a Fr\'echet topological (real) vector space. We can always pick an isomorphism $(E\to \cX) \cong (\cX\times \bR\to \cX)$, which induces an isomorphism $\Gamma(E)\cong C^\infty(\cX,\bR)$, so that field configurations are modeled as smooth functions $\cX\to \bR$. 

Let $E^\ast\to \cX$ denote the dual vector bundle of $E$, and let the canonical pairing be denoted by
\[
\la -,- \ra : E^\ast \otimes E \to \bR
.\]
Let a \emph{compactly supported distributional section} $\ga$ be a distribution of field configurations \hbox{$\ga:\Gamma(E)\to \bR$}, i.e. an element of the topological dual vector space of $\Gamma(E)$, which is modeled as a sequence $(\ga_j)_{j\in \bN}$ of smooth compactly supported sections of the dual bundle $E^\ast\to \cX$,
\[
\ga_j:\cX \hookrightarrow E^\ast, \qquad j\in \bN
,\]
where the modeled distribution is recovered as the following limit of integrals,
\[
\Gamma(E)\to \bR, \qquad \Phi \mapsto \int_{x\in \cX}  \big \la \ga(x), \Phi(x)  \big \ra   \text{dvol}_\cX := \lim_{j\to \infty} \int_{x\in \cX} \la \ga_j(x) , \Phi(x) \ra    \text{dvol}_\cX 
.\]
The space of all compactly supported distributional sections is denoted $\Gamma_{\text{cp}}'(E^\ast)$. By e.g. \cite[Lemma 2.15]{bar15}, all distributions $\Gamma(E)\to \bR$ may be obtained as compactly supported distributional sections.

%We formalize this realization of the sequence of smooth sections as distributions (called linear observables) via a (generalized) system of products on $\textbf{X}$ below.

We can pullback the vector bundle $E^\ast$ to $\cX^I$ along each canonical projection 
\[
\cX^I\to \cX^{\{i\}}\cong \cX 
,\qquad 
i\in I
.\] 
The tensor product of these $n$ many pullback bundles is the exterior tensor product bundle $(E^\ast)^{ \boxtimes I }$. This defines a presheaf of smooth vector bundles on $\sfS$,  
\[
\sfS^{\op}\to \textsf{Diff}_{/ \cX}
, \qquad 
I\mapsto (E^\ast)^{ \boxtimes I }
.\]
By taking compexified compactly supported distributional sections ${{\Gamma'}^\bC}_{\! \! \! \! \! \! \text{cp}}(-):=\Gamma_{\text{cp}}'(-)\otimes_\bR \bC$, we obtain the complex vector species ${{\boldsymbol{\Gamma}'}^\bC}_{\! \! \! \! \! \! \text{cp}}(E^\ast)$, given by
\[
{{\boldsymbol{\Gamma}'}^\bC}_{\! \! \! \! \! \! \text{cp}}(E^\ast)[I]:= {{\Gamma'}^\bC}_{\! \! \! \! \! \! \text{cp}}\big ((E^\ast)^{ \boxtimes I }\big)
.\]

%We formalize a realization of these distributional sections as observables (called polynomial observables) via a (generalized) system of products on $\textbf{E}$ below.

%%%%%%%%%%%%%%%%%%%%%%%%%%%%%%%%%%%%%%%%%%%%%%%%%%%%%%%%%%%%%%%%%%%%%%%%
\section{Observables} 
%%%%%%%%%%%%%%%%%%%%%%%%%%%%%%%%%%%%%%%%%%%%%%%%%%%%%%%%%%%%%%%%%%%%%%%%

%References for this section are \cite[Section 7]{perturbative_quantum_field_theory}, \cite{bar15}, \cite[Section 1.2-1.3]{dutsch2019perturbative}. See also \cite{horm90} for theory of distributions. 

We now set $\Bbbk=\bC$. An off-shell \emph{observable} $\emph{\textsf{O}}$ is a smooth functional of field configurations into the complex numbers,
\[     
\emph{\textsf{O}}:\Gamma(E)\to \bC
,\qquad 
\Phi\mapsto \emph{\textsf{O}}(\Phi)  
.\]
The space of all observables is denoted $\text{Obs}$. We can pointwise multiply observables, sometimes called the \emph{normal ordered product}, so that observables form a commutative $\bC$-algebra,
\[ 
\text{Obs}\otimes \text{Obs} \to \text{Obs}, \qquad \emph{\textsf{O}}_1\otimes \emph{\textsf{O}}_2 \mapsto  \emph{\textsf{O}}_1 \cdot \emph{\textsf{O}}_2         
\]
where
\[      
\emph{\textsf{O}}_1 \cdot \emph{\textsf{O}}_2(\Phi):=  \! \! \! \underbrace{\emph{\textsf{O}}_1(\Phi)\emph{\textsf{O}}_2(\Phi)}_{\text{multiplication in $\bC$}} \! \!  \!  
.\]
Thus, we may form the commutative algebra in species $\textbf{U}_{\text{Obs}}$, given by \textcolor{blue}{(\refeqq{eq:algU})}. A \emph{linear observable} $\emph{\textsf{O}}$ is an observable which is additionally a linear functional, that is
\[
\emph{\textsf{O}}( \Phi_1 + \Phi_2 ) = \emph{\textsf{O}}(\Phi_1) + \emph{\textsf{O}}(\Phi_2)
\qquad \text{and} \qquad
\emph{\textsf{O}}(c \Phi) = c \emph{\textsf{O}}(\Phi) \qquad \text{for}\quad  c\in \bC
.\] 
The space of linear observables is denoted $\text{LinObs}$. In particular, for each spacetime event $x\in \cX$, we have the \emph{field observable} $\boldsymbol{\Phi}(x) \in \text{LinObs}$, given by
\[    
\boldsymbol{\Phi}(x):    \Gamma(E)\to \bC,  \qquad  \Phi\mapsto \Phi(x)
.\] 

We have the following morphism of species,
\[
\textbf{X} \otimes {{\boldsymbol{\Gamma}'}^\bC}_{\! \! \! \! \! \! \text{cp}}(E^\ast) \to \textbf{U}_{\text{Obs}},
\qquad
\tH_i \otimes \ga \mapsto 
\bigg(  \Phi \mapsto   \int_{x\in \cX} \big \la \ga(x), \Phi(x) \big \ra     \text{dvol}_\cX  \bigg)
.\]
It follows from \cite[Lemma 2.15]{bar15} that the colimit (as in e.g. \cite[Remark 15.7]{aguiar2010monoidal}) of the species which is the image of this morphism is all the linear observables $\text{LinObs}$. The currying of this map is given by
\[
\textbf{X}\to \cH \big ( {{\boldsymbol{\Gamma}'}^\bC}_{\! \! \! \! \! \! \text{cp}}(E^\ast) ,  \textbf{U}_{\text{Obs}}\big  ),
\qquad 
\tH_i\mapsto \boldsymbol{\Phi}_i=\boldsymbol{\Phi}    \]
where
\[
\boldsymbol{\Phi}(\ga):=\bigg( \Phi \mapsto  \int_{x\in \cX}\big \la \ga(x), \Phi(x) \big \ra     \text{dvol}_\cX   \bigg)
.\]
If we restrict to bump functions $b\in \Gamma_{\text{cp}}(E^\ast)\otimes_{\bR} \bC$, also called `smearing functions', then one might call the linear map
\[
\boldsymbol{\Phi}: \Gamma_{\text{cp}}(E^\ast)\otimes_{\bR} \bC \to   \text{Obs}
,\qquad
  b\mapsto \boldsymbol{\Phi}(b)
\] 
an observable-valued distribution, and this is sometimes referred to as `the (smeared) field'. The field observable $\boldsymbol{\Phi}(x)$ is recovered by evaluating $\boldsymbol{\Phi}$ on the Dirac delta function $\delta_x$ localized at $x$. One views $b$ as the ``smearing'' of a Dirac delta function, hence smearing functions and smeared field. 
%\cite[Lemma 2.15]{bar15}, Compactly supported distributional sections account for all linear observables $\Gamma(E)\to \bC$. 
%Note that smooth and compactly supported functions $b\in \cX \to \bC$, \emph{bump functions} or `smearing functions', are particular cases. 
%By integrating field configurations against compactly supported distributional sections, we obtain an isomorphism    
%\[
%\text{LinObs} \cong \Gamma'_{\text{cp}}(E^\ast)\cong  {{\boldsymbol{\Gamma}'}^\bC}_{\! \! \! \! \! \! \text{cp}}(E^\ast) [i]
%.\] 
%See \cite[Proposition 3.6]{Moerdijk91}. 
%\begin{figure}[t]
%	\centering
%	\includegraphics[scale=0.42]{prop}
%	\caption{The various propagators for a real scalar field (see e.g. \cite[Section A.2]{dutsch2019perturbative}) in the case $p=0$ are $\bC$-valued generalized functions on the one-dimensional braid arrangement $\bR^2/(1,1)$. The imaginary part is shown in blue.}
%	\label{fig:prop}
%\end{figure}
We extend the smeared field by defining the following morphism of species,
\[
\textbf{E} \otimes {{\boldsymbol{\Gamma}'}^\bC}_{\! \! \! \! \! \! \text{cp}}(E^\ast) \to \textbf{U}_{\text{Obs}}
,\qquad
\tH_I\otimes \ga_I 
\mapsto
\bigg(  \Phi \mapsto   \int_{\cX^I}
\big \la 
\ga_I(x_{i_1}, \dots, x_{i_n}),
\Phi(x_{i_1})\dots \Phi(x_{i_n}) 
\big \ra     \text{dvol}_{\cX^I}  \bigg)  
.\]
The colimit of the species which is the image of this morphism is the vector space of \emph{polynomial observables}, denoted
\[
\text{PolyObs}\subset \text{Obs}
.\]
If we restrict the limit $\mathscr{S}({{\boldsymbol{\Gamma}'}^\bC}_{\! \! \! \! \! \! \text{cp}}(E^\ast)) \to \text{Obs}[[\formj]]$ to finite series and set $\formj=1$, then we recover \cite[Definition 1.2.1]{dutsch2019perturbative}. The space of \emph{microcausal polynomial observables} $\mathcal{F}$ is the subspace 
\[
\mathcal{F} \subset \text{PolyObs}
\] 
consisting of those polynomial observables which satisfy a certain microlocal-theoretic condition called \emph{microcausality}, for which we direct the reader to \cite[Definition 1.2.1 (ii)]{dutsch2019perturbative}. Following \cite[Definition 1.3.4]{dutsch2019perturbative}, the space of \emph{local observables}
\[
\mathcal{F}_{\text{loc}}\subset \text{Obs}
\] 
consists of those observables obtained by integrating a polynomial with real coefficients in the field and its derivatives (`field polynomials') against a bump function $b\in \Gamma_{\text{cp}}(E^\ast)\otimes_\bR \bC$. Importantly, we have a natural inclusion
\[
\mathcal{F}_{\text{loc}} \hookrightarrow \mathcal{F}
,\qquad
\ssA \mapsto\  :\ssA: 
.\] 
Let $\mathcal{F}_{\text{loc}}[[\hbar]]$ and $\mathcal{F}[[\hbar]]$ denote the spaces of formal power series in $\hbar$ with coefficients in $\mathcal{F}_{\text{loc}}$ and $\mathcal{F}$ respectively, and let $\mathcal{F}((\hbar))$ denote the space of Laurent series in $\hbar$ with coefficients in $\mathcal{F}$.

Applying Moyal deformation quantization with formal Planck's constant $\hbar$, $\mathcal{F}[[\hbar]]$ is a formal power series $\ast$-algebra, called the (abstract, off-shell) \emph{Wick algebra}, with multiplication the Moyal star product \cite[Definition 2.1.1]{dutsch2019perturbative} defined with respect to the Wightman propagator $\Delta_{\text{H}}$ for the Klein-Gordan field \cite[Section 2.2]{dutsch2019perturbative}, 
\[
\mathcal{F}[[\hbar]] \otimes \mathcal{F}[[\hbar]] \to \mathcal{F}[[\hbar]]
,\qquad
\emph{\textsf{O}}_1 \otimes \emph{\textsf{O}}_2\\
\mapsto
\emph{\textsf{O}}_1\star_{\text{H}}\emph{\textsf{O}}_2
.\] 
We may form the algebra in species $\textbf{U}_{\mathcal{F}[[\hbar]]}$, or, after allowing negative powers of $\hbar$, $\textbf{U}_{\mathcal{F}((\hbar))}$. 

%Briefly, the wave front set of each $\ga_I$ must not include those points where all $n$ wave vectors are in the closed future cone or all in the closed past cone. 

%%%%%%%%%%%%%%%%%%%%%%%%%%%%%%%%%%%%%%%%%%%%%%%%%%%%%%%%%%%%%%%%%%%%%%%%
\section{Time-Ordered Products and S-Matrix Schemes} \label{sec:Time-Ordered Products}
%%%%%%%%%%%%%%%%%%%%%%%%%%%%%%%%%%%%%%%%%%%%%%%%%%%%%%%%%%%%%%%%%%%%%%%%

%The \emph{retarded propagator} $\Delta^{ \text{ret} }$ is the fundamental solution,
%\[    (   \square +m^2  ) \Delta^{ \text{ret} }=-\delta  , \qquad \text{with}\quad   \text{supp} %(\Delta^{ \text{ret} })\subseteq \overline{V}_+.  \]
%This is uniquely determined.

For $\ssA\in \mathcal{F}_{\text{loc}}[[\hbar]]$, let $\text{supp}(\ssA)$ denote the spacetime support of $\ssA$. Given a composition $G$ of $I$, we say that $\ssA_I\in \textbf{E}_{\mathcal{F}_{\text{loc}}[[\hbar]]}[I]$ \emph{respects} $G$ if
\[
\text{supp}({\ssA_{i_1}})\vee\! \! \wedge\ \text{supp}(\ssA_{i_2}) 
\qquad \quad \text{for all}\quad 
(i_1,i_2)\quad \text{such that}\quad  G|_{\{i_1, i_2\}}= (i_1,i_2) 
.\footnote{\ $G|_{\{i_1, i_2\}}= (i_1,i_2)$ means that $i_1$ and $i_2$ are in different lumps, with the lump containing $i_1$ appearing to the left of the lump containing $i_2$}\]  
Consider a system of $\text{T}$-products (as defined in \autoref{sec:T-Products, Generalized T-Products, and Generalized R-Products}) of the form
\[
\text{T}: \textbf{E}^\ast_+ \otimes \textbf{E}_{\mathcal{F}_{\text{loc}}[[\hbar]]}\to 
\textbf{U}_{\mathcal{F}((\hbar))},
\qquad
\tH_{(I)}\otimes \ssA_I\mapsto \text{T}_I(\tH_{(I)}\otimes \ssA_I)=\text{T}_I(\ssA_I)
.\]

\begin{figure}[t]
	\centering
	\includegraphics[scale=0.45]{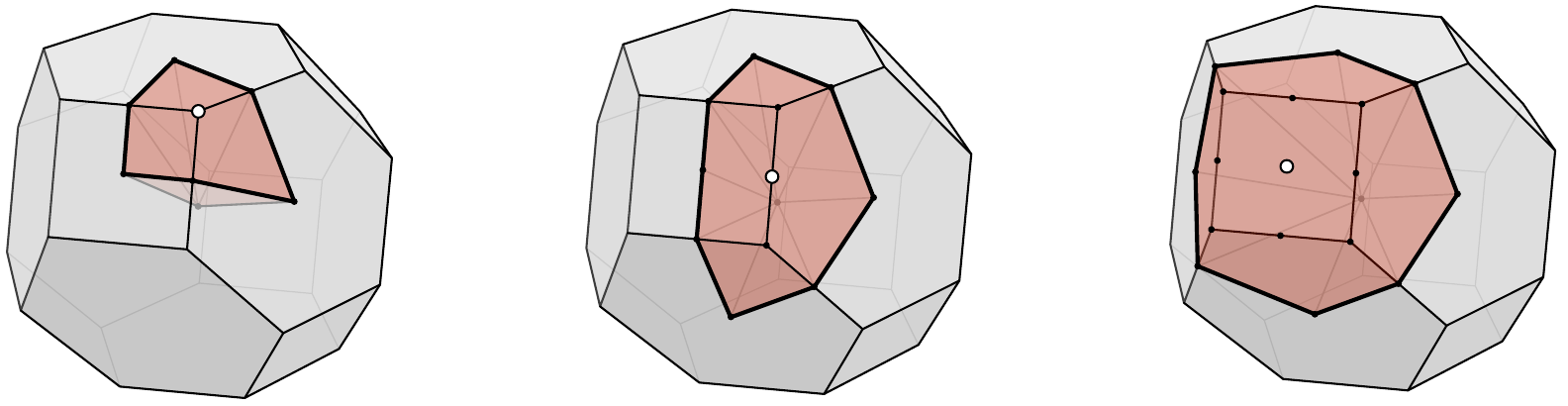}
	\caption{For $p=0$, configurations $\lambda\in \cX^I/\cX=\bR^I/(1,\dots,1)$ shown in red which respect the composition $G$ shown in white (i.e. \hbox{$\lambda(i_1)\vee\! \! \wedge\ \lambda(i_2)$} for all $(i_1,i_2)$ such that $G|_{\{i_1, i_2\}}=(i_1,i_2)$), depicted on the tropical toric compactification of the tropical torus $\bR^I/(1,\dots,1)=(\bT^\times)^I/\bT^\times$ with respect to the braid arrangement fan.}
	\label{fig:respects}
\end{figure}

%When interactions are adiabatically switched on, and we take the adiabatic limit, the region of factorization of scattering amplitudes (\autoref{prop:scattering}) will tend towards the boundary facets, where configurations are separated by infinite times

\noindent Since $\Sig$ is the free algebra on $\textbf{E}^\ast_+$, we have the unique extension to a system of generalized \hbox{$\text{T}$-products}
\[
\text{T}: \Sig \otimes \textbf{E}_{\mathcal{F}_{\text{loc}}[[\hbar]]}\to 
\textbf{U}_{\mathcal{F}((\hbar))},
\qquad
\text{T}_I(\tH_{F}\otimes \ssA_I):=\text{T}_{S_1}(\ssA_{S_1})\star_{\text{H}} \dots \star_{\text{H}}  \text{T}_{S_k}(\ssA_{S_k}) 
.\]
Then:
\begin{enumerate}[label={\arabic*.}]
\item 
(perturbation) we say that $\text{T}$ satisfies \emph{perturbation} if the singleton components $\text{T}_{i}$ are isomorphic to the inclusion $\mathcal{F}_{\text{loc}}[[\hbar]]\hookrightarrow \mathcal{F}((\hbar))$, that is
\[   
{\text{T}}_{i}( \ssA  )=  \,   :\! \ssA \! :       
\]
\item
(causal factorization) we say that $\text{T}$ satisfies \emph{causal factorization} if for all compositions $(S,T)$ of $I$ with two lumps, if $\ssA_I \in \textbf{E}_{\mathcal{F}_{\text{loc}}[[\hbar]]}[I]$ respects $(S,T)$\footnote{\ explicitly, $\text{supp}({\ssA_{i_1}})\vee\! \! \wedge\ \text{supp}(\ssA_{i_2})$ for all $i_1\in S$ and $i_2\in T$} then
\begin{equation}\label{eq:causalfac}
\text{T}_I(\tH_{(I)}\otimes\ssA_I)
=\text{T}_I(\tH_{(S,T)}\otimes\ssA_I).\footnote{\ or equivalently $\text{T}_I(\ssA_I)=\text{T}_S(\ssA_S)\star_{\text{H}} \text{T}_T(\ssA_T)$} 
\end{equation}
%(Recall that $\text{T}(\tH_{(S,T)}\otimes\ssA_I)=\text{T}(\tH_{(S)}\otimes \ssA_S)\star_{\text{H}} \text{T}(\tH_{(T)}\otimes\ssA_T)$.)
\end{enumerate}
Let a (fully normalized) \emph{system of time-ordered products} be a system of $\text{T}$-products which satisfies perturbation and causal factorization. The corresponding unique extension of $\text{T}$ to $\Sig$ is called the associated \emph{system of generalized time-ordered products}. After currying
\[     
\Sig \to \cH(  \textbf{E}_{\mathcal{F}_{\text{loc}}[[\hbar]]} ,  \textbf{U}_{\mathcal{F}((\hbar))} ), 
\qquad     
\tH_{F}\mapsto  \text{T}(S_1)\dots \text{T}(S_k)   
\]
the linear maps
\[
\text{T}(S_1)\dots \text{T}(S_k):  \mathcal{F}_{\text{loc}}[[\hbar]]^{\otimes I}  
\to
\mathcal{F}((\hbar))
,\qquad
\ssA_I\mapsto \text{T}_I(\tH_{F}\otimes \ssA_I)
\]
are called \emph{generalized time-ordered products}. The linear maps $\text{T}(I)$ are called \emph{time-ordered products}. After fixing a field polynomial, so that each $\ssA_{i_j}$ of $\ssA_I$ is determined by a bump function $b_{i_j}$, they are usually presented in generalized function notation as follows,
\[
\text{T}_I (\ssA_{i_1}  \otimes \cdots \otimes \ssA_{i_n} )
= 
\int_{\cX^I}   \text{T}(x_{i_1}, \dots , x_{i_n})  b_{i_1}( x_{i_1} ) \dots  b_{i_n}(x_{i_1}) 
dx_{i_1} \dots dx_{i_n}
\]
where $(x_{i_1}, \dots , x_{i_n}) \mapsto \text{T}(x_{i_1}, \dots , x_{i_n})$ is an ``operator-valued'' (Wick algebra-valued) generalized function. See e.g. \cite[Section 1.2]{ep73roleofloc}. In the following, recall the Tits product `$\, \triangleright$' from \autoref{sec:Tits}.

%Causal factorization is a condition on $\text{T}(I)$ in terms of $\text{T}(S)$ and $\text{T}(T)$. In this way, time-ordered products may be constructed inductively, as in \cite{ep73roleofloc}. %This result appears to be part of the physicists motivation for the Tits product. %This is an interesting question, since the species-theoretic Hopf structure of $\Sig$ does not appear explicitly in physics.

\begin{prop}\label{prob:causalfac}
Let 
\[
\text{T}: \textbf{E}^\ast_+ \otimes \textbf{E}_{\mathcal{F}_{\text{loc}}[[\hbar]]}\to 
\textbf{U}_{\mathcal{F}((\hbar))}
\] 
be a system of $\text{T}$-products which satisfies causal factorization. Given a composition \hbox{$G=(U_1,\dots,U_k)$} of $I$, and $\ssA_I \in \textbf{E}_{\mathcal{F}_{\text{loc}}[[\hbar]]}[I]$ which respects $G$, then 
\[ 
\text{T}_I( \mathtt{a} \otimes\ssA_I)
= 
\text{T}_I( \mathtt{a} \triangleright \tH_{G} \otimes\ssA_I)
\qquad  \text{for all} \quad 
\mathtt{a}\in \Sig[I]     
.\]
\end{prop}
\begin{proof}
We have
\[
\text{T}_I(\tH_{G}\otimes\ssA_I)
=
\underbrace{\text{T}_{U_1}(\ssA_{U_1})\star_{\text{H}} \dots \star_{\text{H}} \text{T}_{U_k}(\ssA_{U_k})
=
\text{T}_I(\ssA_I)}_{\text{by repeated applications of causal factorization}}
.\]
Observe that the action $\tH_F \mapsto \tH_F \triangleright \tH_{G}$, for $F\in \Sigma[I]$, replaces the lumps of $F$ with their intersections with $G$. But we just saw that $\text{T}_I(\ssA_I)=\text{T}_I(\tH_{G}\otimes\ssA_I)$, and so it follows that
\[
\text{T}_I( \tH_F \otimes\ssA_I)
= 
\text{T}_I( \tH_F \triangleright \tH_{G} \otimes\ssA_I)
.\]
Since the claim is true for the $\tH$-basis, it is true for all $\mathtt{a}\in \Sig[I]$. 
\end{proof}

\begin{cor}\label{cor:cor}
If $\mathtt{a} \triangleright \tH_{G} =0$, then 
\[
\text{T}_I(\mathtt{a} \otimes\ssA_I)=0
\]
for all $\ssA_I\in \textbf{E}_{\mathcal{F}_{\text{loc}}[[\hbar]]}[I]$ which respect $G$. 
\end{cor}

The restriction of $\text{T}$ to the primitive part Lie algebra is called the associated \emph{system of generalized retarded products},
\[
\text{R}:\Zie\otimes\textbf{E}_{\mathcal{F}_{\text{loc}}[[\hbar]]} \to \mathcal{F}((\hbar))  
.\]
The image of the Dynkin elements $\mathtt{D}_\cS$ under the currying of $\text{R}$ are the \emph{generalized retarded products} $\text{R}_\cS$, see e.g. \cite[Equation 79]{ep73roleofloc}. Combining this with \autoref{actondynkinwithtits} and \autoref{cor:cor}, we see that the generalized retarded products have certain interesting support properties.
 
Given a system of generalized time-ordered products
\[
\text{T}: \Sig \otimes \textbf{E}_{\mathcal{F}_{\text{loc}}[[\hbar]]}\to 
\textbf{U}_{\mathcal{F}((\hbar))}
\]
the $\text{T}$-exponential $\mathcal{S}=\mathcal{S}_{\mathtt{G}(1/\text{i}\hbar)}$\footnote{\ here we let $\Sig$ be defined over $\Bbbk=\bC((\hbar))$, i.e. the field of restricted Laurent series in $\hbar$ with coefficients in $\bC$} (defined in \textcolor{blue}{(\refeqq{eq:tsexp})}) for the group-like series 
\[
\mathtt{G}(1/\text{i}\hbar): \textbf{E} \to \Sig
,\qquad
\tH_{I} \mapsto   \dfrac{1}{\text{i}\hbar}\tH_{(I)}
\]
is called the associated perturbative \emph{S-matrix scheme}. Thus, $\mathcal{S}$ is the function 
\[
\mathcal{S}: \mathcal{F}_{\text{loc}}[[\hbar]] \to \mathcal{F}((\hbar))[[\formj]] 
,\qquad        
\ssA\mapsto  \mathcal{S}(\formj\! \ssA)  
=
\sum^\infty_{n=0} \bigg(\dfrac{1}{\text{i}\hbar}\bigg)^{\! n}\dfrac{\formj^n}{n!} \text{T}_n(\ssA^n)     
.\]
Given a choice of vector $\ssS_{\text{int}}\in \mathcal{F}_{\text{loc}}[[\hbar]]$, now called the adiabatically switched \emph{interaction}, recall from \autoref{sec:Perturbation of T-Products by Up Coderivation} that we also have the $\text{T}$-exponential perturbed by the up coderivation of $\textbf{E}^\ast$,
\[
\mathcal{S}_{\formg \ssS_{\text{int}}}: \mathcal{F}_{\text{loc}}[[\hbar]] \to \mathcal{F}((\hbar))[[\formg,\! \formj]] 
,\qquad        
\ssA\mapsto  \mathcal{S}_{\formg \ssS_{\text{int}}}(\formj\! \ssA)  
=
\mathcal{S}( \formg \ssS_{\text{int}}+ \formj \! \ssA)     
.\]
As in \cite[Section 15]{perturbative_quantum_field_theory}, we may extend $\bC[[\hbar,\formg,\! \formj]]$-linearly to obtain a map
\[
\mathcal{F}_{\text{loc}}[[\hbar, \formg,\! \formj]]\la \formg,\! \formj \ra  \to \mathcal{F}((\hbar))[[\formg,\! \formj]] 
\]
where $\mathcal{F}_{\text{loc}}[[\hbar, \formg,\! \formj]]\la \formg,\! \formj \ra$ denotes formal power series which are at least linear in $\formg$ or $\! \formj$. 

%%%%%%%%%%%%%%%%%%%%%%%%%%%%%%%%%%%%%%%%%%%%%%%%%%%%%%%%%%%%%%%%%%%%%%%%
\section{Interactions} 
%%%%%%%%%%%%%%%%%%%%%%%%%%%%%%%%%%%%%%%%%%%%%%%%%%%%%%%%%%%%%%%%%%%%%%%%

Given a choice of interaction $\ssS_{\text{int}} \in \mathcal{F}_{\text{loc}}[[\hbar]]$, and a system of fully normalized generalized \hbox{time-ordered} products
\[
\text{T}: \Sig \otimes \textbf{E}_{\mathcal{F}_{\text{loc}}[[\hbar]]}\to 
\textbf{U}_{\mathcal{F}((\hbar))}
,\]
we have the new system of interacting generalized time-ordered products which is obtained by applying the construction of \autoref{sec:Perturbation of Products and Series} to the retarded Steinmann arrow biderivation $\downarrow(-)$ for $\Sig$, as considered in \autoref{sec:Perturbation of T-Products by Steinmann Arrows},
\[ 
\wt{\text{T}}
:
\Sig\otimes\textbf{E}_{\mathcal{F}_{\text{loc}}[[\hbar]]} 
\to 
\textbf{U}_{\mathcal{F}((\hbar))[[\formg]]}, 
\qquad 
\wt{\text{T}}:
=\text{T}^{\textbf{E}}\circ  \check{\rho}_{\downarrow,\ssS_{\text{int}}} 
.\]
The associated \emph{generating function scheme} $\mathcal{Z}_{\formg \ssS_{\text{int}}}$ for interacting field observables (more generally for time-ordered products of interacting field observables) is the new $\text{T}$-exponential for the group-like series $\mathtt{G}(1/\text{i}\hbar)$, denoted $\mathcal{V}_{\formg \ssS_{\text{int}}}$ in \autoref{sec:Perturbation of T-Products by Steinmann Arrows}. Thus, $\mathcal{Z}_{\formg \ssS_{\text{int}}}$ is the function
\[
\mathcal{Z}_{\formg \ssS_{\text{int}}}
: 
\mathcal{F}_{\text{loc}}[[\hbar]] \to \mathcal{F}((\hbar))[[\formg,\! \formj]] 
,\qquad
\ssA \mapsto \mathcal{Z}_{\formg \ssS_{\text{int}}}(\formj\! \ssA)
\]
where
\[
\mathcal{Z}_{\formg \ssS_{\text{int}}}(\formj\! \ssA)
\! :=\!
\sum_{n=0}^\infty \!  \bigg(\dfrac{1}{\text{i}\hbar}\bigg)^{\! \!  n}\! \dfrac{\formj^n}{n!} \wt{\text{T}}_n(\ssA_n)
=\! 
\sum_{n=0}^\infty \sum_{r=0}^\infty \!  
\bigg(\dfrac{1}{\text{i}\hbar}\bigg)^{\! \!  r+n}\! 
\dfrac{\formg^{r} \formj^n}{r!\, n!}\, \!   \text{R}_{r;n} (\ssS_{\text{int}}^{\, r} ; \ssA^n)    
=
\mathcal{S}^{-1}( \formg \ssS_{\text{int}})\star_{\text{H}} \mathcal{S}(\formg \ssS_{\text{int}} +\formj\! \ssA )
.\]
Then 
\[
\ssA_{\text{int}}:
=
\wt{\text{T}}_i(\ssA)
=
\sum_{r=0}^\infty 
\bigg(\dfrac{1}{\text{i}\hbar}\bigg)^{\! r} 
\dfrac{\formg^{r}}{r!} \text{R}^r_1 (\ssS^{\, r}_{\text{int}}; \ssA)  
\in \mathcal{F}((\hbar))[[\formg]]
\] 
is the \emph{local interacting field observable} of $\ssA$. Bogoliubov's formula \textcolor{blue}{(\refeqq{eq:Bog})} now reads
\[
\ssA_{\text{int}}
=
\text{i} \hbar \, \dfrac{d}{d\formj}\Bigr|_{\formj=0}  \mathcal{Z}_{\formg \ssS_{\text{int}}}(\formj\! \ssA)
.\]
One views $\ssA_{\text{int}}$ as the deformation of the local observable $\ssA$ due to the interaction $\ssS_{\text{int}}$ being turned on. One can show that $\wt{\text{T}}$ does indeed land in $\textbf{U}_{\mathcal{F}[[\hbar,\formg]]}$ \cite[Proposition 2 (ii)]{dutfred00}. The perturbative interacting quantum field theory then has a classical limit \cite{collini2016fedosov}, \cite{MR4109798}. 

%This means that interacting field observables have a classical limit, being a formal deformation quantization of the classical interacting field theory.
%(, Hawkins-Rejzner 16, cor. 5.2). 

% and the map $\ssA\mapsto \ssA_{\text{int}}$ is somtimes called a quantum M{\o}ller operator. 

%%%%%%%%%%%%%%%%%%%%%%%%%%%%%%%%%%%%%%%%%%%%%%%%%%%%%%%%%%%%%%%%%%%%%%%%
\section{Scattering Amplitudes}\label{sec:scatterung}
%%%%%%%%%%%%%%%%%%%%%%%%%%%%%%%%%%%%%%%%%%%%%%%%%%%%%%%%%%%%%%%%%%%%%%%%

We finish with a translation of a standard result in pAQFT (see \cite[Example 15.12]{perturbative_quantum_field_theory}) into our notation, which relates S-matrix schemes as presented in \autoref{sec:Time-Ordered Products} to S-matrices used to compute scattering amplitudes, which are predictions of pAQFT that are tested with scattering experiments at particle accelerators. 

Following \cite[Definition 2.5.2]{dutsch2019perturbative}, the \emph{Hadamard vacuum state} $\la - \ra_0$ is the linear map given by
\[ 
\la - \ra_0
: 
\mathcal{F}[[\hbar,\formg]] \to \bC[[\hbar,\formg]]
,\qquad
\emph{\textsf{O}}\mapsto  \la  \emph{\textsf{O}}\,  \ra_0:= \emph{\textsf{O}}\, (\Phi=0)
.\]
Let $\ssS_{\text{int}} \in \mathcal{F}_{\text{loc}}[[\hbar]]$. We say that the Hadamard vacuum state $\la - \ra_0$ is \emph{stable} with respect to the interaction $\ssS_{\text{int}}$ if for all $\emph{\textsf{O}}\in \mathcal{F}[[\hbar,\formg]]$, we have
\begin{equation}\label{eq:vacstab}
\big\la\emph{\textsf{O}}\star_{\text{H}}\mathcal{S}(\formg \ssS_{\text{int}}) \big \ra_0
=
\big \la\emph{\textsf{O}}\, \big\ra_0
\big\la    \mathcal{S}(\formg \ssS_{\text{int}})   \big \ra_0    
\qquad \text{and} \qquad
\big\la    
\mathcal{S}^{-1}(\formg \ssS_{\text{int}})\star_{\text{H}} \emph{\textsf{O}}  \,  
\big\ra_0
=
\dfrac{1}{\big \la \mathcal{S}(\formg \ssS_{\text{int}})\big \ra_0 } \big \la \emph{\textsf{O}}\, \big \ra_0
.
\end{equation}
In situations where 
\[
\ssS_{\text{int}} \otimes \ssA_I\in \textbf{E}'_{\mathcal{F}_{\text{loc}}[[\hbar]]}[I]
\qquad \text{respects} \qquad 
(S,\ast,T)
\footnote{\ this is a composition of $\{\ast\}\sqcup I$}
\] 
%and $
%\text{supp}(\ssA_{i_1}) >< \text{supp}(\ssA_{i_2})
%$ 
%if $\{i_1,i_2\}\subseteq S$ or $\{i_1,i_2\}\subseteq T$, 
we can interpret free particles/wave packets labeled by $T$ coming in from the far past, interacting in a compact region according to the adiabatically switched interaction $\ssS_{\text{int}}$, and then emerging into the far future, labeled by $S$. For $\ssA_I\in \textbf{E}_{\mathcal{F}_{\text{loc}}[[\hbar]]}[I]$, let
\[\text{G}_I(\ssA_I):=\big\la 
\widetilde{\text{T}}(\ssA_I)  
\big\ra_0
.\]
If we fix the field polynomial of local observables to be $\text{P}(\Phi)=\Phi$, then $\ssA_I\mapsto \text{G}_I(\ssA_I)$ is the \hbox{time-ordered} $n$-point correlation function, or Green's function. They are usually presented in generalized function notation as follows,
\[
\text{G}_I (b_{i_1}  \otimes \cdots \otimes b_{i_n} )
= 
\int_{\cX^I} \Big\la   \text{T} \big ( \boldsymbol{\Phi}(x_{i_1}) \dots \boldsymbol{\Phi}(x_{i_n})\big ) \Big  \ra_0   b_{i_1}( x_{i_1} ) \dots  b_{i_n}(x_{i_n}) 
dx_{i_1} \dots dx_{i_n}
.\]
Note that to obtain the true Green's functions, we still have to take the adiabatic limit. 

%Note that $ \ssS_{\text{int}}$ is adiabatically switched, and so one should now take the algebraic adiabatic limit. 

\begin{prop} \label{prop:scattering}
If the Hadamard vacuum state $\la-\ra_0$ is stable with respect to $\ssS_{\text{int}} \in \mathcal{F}_{\text{loc}}[[\hbar]]$, and if $\ssS_{\text{int}} \otimes \ssA_I\in \textbf{E}'_{ \mathcal{F}_{\text{loc}[[\hbar]]}}[I]$ respects the composition $(S, \ast, T)$, then
\[  
\text{G}_I(\ssA_I)=  
\dfrac{1}
{\Big\la
\mathcal{S}(\formg\ssS_{\text{int}}) 
\Big\ra_0 }   
\Big\la \text{T}_S(\ssA_S)\star_{\text{H}}  \mathcal{S}(\formg\ssS_{\text{int}})\star_{\text{H}}  
\text{T}_T(\ssA_T)
\Big\ra_0
.\footnote{\ the element $\mathcal{S}(\formg\ssS_{\text{int}})\in \mathcal{F}((\hbar))[[\formg]]$ is called the perturbative S\emph{-matrix}}\]
\end{prop}
\begin{proof}
We have
\begin{align*}
\text{G}_I(\ssA_I) 
&=
\big\la 
\widetilde{\text{T}}(\ssA_I)  
\big\ra_0 \\[6pt] 
&= 
\bigg\la 
\sum_{r=0}^\infty \dfrac{\formg^{r}}{r!} 
\text{R}_{r;I}( \ssS^{\, r}_{\text{int}} ; \ssA_{I} )
\bigg\ra_0  \\[6pt] 
&= 
\bigg\la 
\sum_{r=0}^\infty \sum_{r_1 + r_2 =r}   \dfrac{\formg^{r}}{r_1! \, r_2!} 
\overline{\text{T}}^{r_1}_{\emptyset}(\ssS^{\, r_1}_{\text{int}}) \star_{\text{H}}  \text{T}^{r_2}_{I}(\ssS^{\, r_2}_{\text{int}} \ssA_{I}) 
\bigg\ra_0.
\end{align*}
To obtain the final line, we expanded the retarded products according to \textcolor{blue}{(\refeqq{eq:retardprod})}. Then, by causal factorization \textcolor{blue}{(\refeqq{eq:causalfac})}, we have
\[
\text{T}^{r_2}_{I}(\ssS^{\, r_2}_{\text{int}} \ssA_{I})
=
\text{T}_{S}(\ssA_{S})
\star_{\text{H}}
\text{T}^{r_2}_{ \emptyset}(\ssS^{\, r_2}_{\text{int}})
\star_{\text{H}}
\text{T}_{T}(\ssA_{T})
.\]
Therefore
\begin{align*}
\text{G}_I(\ssA_I) 
&=
\bigg\la 
\sum_{r=0}^\infty \sum_{r_1 + r_2 =r}   \dfrac{\formg^{r}}{r_1! \, r_2!} 
\overline{\text{T}}^{r_1}_{\emptyset}(\ssS^{\, r_1}_{\text{int}}) \star_{\text{H}} 
\text{T}_{S}(\ssA_{S})
\star_{\text{H}}
\text{T}^{r_2}_{\emptyset}(\ssS^{\, r_2}_{\text{int}})
\star_{\text{H}}
\text{T}_{T}(\ssA_{T})
\bigg\ra_0. \\[6pt]
&=
\bigg\la 
\sum_{r=0}^\infty \dfrac{\formg^{r}}{r!} 
\overline{\text{T}}^{r}_{\emptyset}(\ssS^{\, r}_{\text{int}}) \star_{\text{H}} 
\text{T}_{S}(\ssA_{S})
\star_{\text{H}}
\sum_{r=0}^\infty \dfrac{\formg^{r}}{r!} 
\text{T}^{r}_{\emptyset}(\ssS^{\, r}_{\text{int}})
\star_{\text{H}}
\text{T}_{T}(\ssA_{T})
\bigg\ra_0. \\[6pt]
&= 
\Big \la  
\mathcal{S}^{-1}(\formg \ssS_{\text{int}}) 
\star_{\text{H}} 
\text{T}_S(\ssA_S)
\star_{\text{H}}
\mathcal{S}(\formg \ssS_{\text{int}}) 
\star_{\text{H}}
\text{T}_T(\ssA_T)
\Big \ra_0  \\[6pt] 
&=
\dfrac{1}
{  \Big \la \mathcal{S}(\formg\ssS_{\text{int}}) \Big  \ra_0 }   
\Big\la \text{T}_S(\ssA_S)\star_{\text{H}}  \mathcal{S}(\formg\ssS_{\text{int}})\star_{\text{H}}  
\text{T}_T(\ssA_T)
\Big\ra_0.
\end{align*}
For the final step, we used vacuum stability \textcolor{blue}{(\refeqq{eq:vacstab})}. 
\end{proof}

%\newpage
\bibliographystyle{alpha}
\bibliography{steinmann}
%\nocite{*}
%\printindex

\end{document}